\documentclass[a4paper,11pt]{article}
\usepackage[a4paper,left=2cm,right=2cm,top=3cm,bottom=3cm]{geometry}
\usepackage[round]{natbib}
\usepackage{bbm}
\usepackage{amssymb}
\usepackage{amsmath}
\usepackage{fullpage}
\usepackage{amsthm}
\usepackage{graphicx}
\usepackage{algorithm}
\usepackage{algpseudocode}
\usepackage{authblk}

\algnewcommand\algorithmicinput{\textbf{INPUT:}}
\algnewcommand\INPUT{\item[\algorithmicinput]}
\algnewcommand\algorithmicoutput{\textbf{OUTPUT:}}
\algnewcommand\OUTPUT{\item[\algorithmicoutput]}
\algnewcommand\algorithmicinitialize{\textbf{Initialize:}}
\algnewcommand\initialize{\item[\algorithmicinitialize]}
\algnewcommand\algorithmicstage{\textbf{Stage}}
\algnewcommand\stage{\item[\algorithmicstage]}

\usepackage{enumitem}

\usepackage{hyperref}[]
\hypersetup{
    colorlinks=true,
    linkcolor=red,
    filecolor=magenta,      
    urlcolor=cyan,
      citecolor=blue,
    }

\usepackage{cleveref}

\newtheorem{lemma}{Lemma}
\newtheorem{proposition}{Proposition}
\newtheorem{corollary}{Corollary}

\newtheorem{definition}{Definition}

\newtheorem{theorem}{Theorem}

\newcommand{\1}  { {(1)}} 
\newcommand{\2}  { {(2)}}

\newcommand{\eps}{\varepsilon}
\newcommand{\E}{\mathbb E}

\DeclareMathOperator*{\argmin}{arg\,min}
\DeclareMathOperator*{\argmax}{arg\,max}

 \newcommand{\lb}{ \langle}
\newcommand{\rb}{ \rangle}

\newcommand{\lt}{ {\mathcal L^2}}
 \newcommand{\I}{ {\mathcal I}}
  \newcommand{\hk}{ { \mathcal H (K) }}
\newcommand{\kmh}{{K^{-1/2}}}
\newcommand{\kh}{{K^{1/2}}}

\newcommand{\shat}{\widehat{\Sigma}}

\newcommand{\shati}{\widehat{\Sigma}_{\mathcal{I}}}

\newcommand{\sx}{ {\Sigma}}

\newcommand{\weta}{ \widehat{\eta}}

\newcommand{\what}[1]{\widehat{#1}}
\newcommand{\wtilde}[1]{\widetilde{#1}}

\newcommand{\st}{ {(s,t]}}
\newcommand{\te}{ {(t,e]}}
\newcommand{\se}{ {(s,e]}}

\newcommand{\sem}{ {(s_m,e_m]}}
\newcommand{\semp}{ {(s_m',e_m']}}

\newcommand{\flst}{f^\lambda_{(s,t]}}
\newcommand{\fhast}{\widehat{f}_{(s,t]}}

\newcommand{\bhast}{\widehat{\beta}_{\st}}
\newcommand{\bhase}{\widehat{\beta}_{\se}}
\newcommand{\shast}{\widehat{\Sigma}_{\st}}
\newcommand{\bhate}{\widehat{\beta}_{(t,e]}}

\newcommand{\Gset}{\widetilde{G}^{s,e}_t}
\newcommand{\Gsetm}{\widetilde{G}^{s_m,e_m}_t}
\newcommand{\Whset}{\widehat{W}^{s,e}_t}
\newcommand{\Whsetm}{\widehat{W}^{s_m,e_m}_t}

\newcommand{\lft}{\left}
\newcommand{\rgt}{\right}
\newcommand{\boldI}{{\bf I}}

\newcommand{\Ktilde}{{\mathcal{K}}}
\newcommand{\T}{ {\mathcal T}}
\newcommand{\J}{ {\mathcal J}}
\newcommand{\B}{ {\mathcal B}}
\newcommand{\calP}{ {\mathcal P}}
\newcommand{\etahat}{ {\widehat{\eta}}}
\newcommand{\Khat}{ {\widehat{\Ktilde}}}

\newcommand{\calD}{{\mathcal D}}
\newcommand{\calQ}{ {\mathcal Q}}

\newcommand{\dau}{\partial}

\newtheorem{assumption}{Assumption}
\date{\vspace{-5ex}}

\begin{document}

\title{Estimation and Inference for Change Points in Functional  Regression Time Series}

\author[1]{Shivam Kumar}
\author[2]{Haotian Xu}
\author[3]{Haeran Cho}
\author[4]{Daren Wang}
\affil[1,4]{Department of ACMS, University of Notre Dame}
\affil[2]{Department of Statistics, University of Warwick}
\affil[3]{School of Mathematics, University of Bristol}

\maketitle

\begin{abstract}
   In this paper, we study the estimation and inference of change points under a functional linear regression model with changes in the slope function. We present a novel Functional Regression Binary Segmentation (FRBS) algorithm which is computationally efficient as well as achieving consistency in multiple change point detection. This algorithm utilizes the predictive power of piece-wise constant functional linear regression models in the reproducing kernel Hilbert space framework. We further propose a refinement step that improves the localization rate of the initial estimator output by FRBS, and derive asymptotic distributions of the refined estimators for two different regimes determined by the magnitude of a change.  To facilitate the construction of confidence intervals for underlying change points based on the limiting distribution, we propose a consistent block-type long-run variance estimator.  Our theoretical justifications for the proposed approach accommodate temporal dependence and heavy-tailedness in both the functional covariates and the measurement errors. Empirical  effectiveness of our methodology is demonstrated through extensive simulation studies and an application to the Standard and Poor’s 500 index dataset.
\end{abstract}

\section{Introduction}\label{sec: multiple CP}

Functional Data Analysis (FDA) studies data that are represented as random functions. The infinite dimension of functional data poses a significant challenge to the development of statistical methodologies. 
Functional Principal Component Analysis (FPCA), a pivotal approach in FDA, focuses on characterizing the dominant modes of variation in random functions. Seminal contributions to the development and application of FPCA include, for example, \citet{ramsay2005principal} and \citet{muller2005functional}. Another important approach in this area employs strategies based on Reproducing Kernel Hilbert Space (RKHS) for estimating the mean, covariance, and slope functions, as demonstrated in \cite{caiRkhsCovariance}.  Unlike non-parametric methods such as FPCA, the RKHS-based approach selects the most representative functional features in an adaptive manner from an RKHS. We refer to \citet{muller2016functional} for a comprehensive overview of the FDA. Extensive treatments of the subject can also be found in \citet{ramsay2002applied}, \citet{kokoszka2012functional}, \citet{hsingeubanktheoretical}, and \citet{kokoszka2017book}.

Functional time series analysis is an important area within FDA, focusing on functional data with temporal dependence. From the modeling perspective, \citet{Yao2000functional} focused on functional   regression via local linear modeling; \citet{kowal2017abayesian} investigated functional linear models; and \citet{kowal2019functional2017b} explored functional autoregression.  To analysis functional time series, \citet{panaretos2013fourier} employed a Fourier analysis-based approach and \citet{PanaretosFTS} considered the estimation of the dynamics of functional time series in a sparse sampling regime. We refer to \citet{konerstaicu2023second} for a comprehensive survey.

In this paper, we focus on a functional linear regression model with the slope function changing in a piece-wise constant manner. 
Given the data sequence $\lft\{(y_j,X_j )\rgt\}_{j=1}^n$, we consider the model  
	\begin{align}\label{eq: model}
	    y_{j} &= \lb \beta ^*_j ,  X_j \rb_\lt  +\eps_j , \qquad 1\le j \le n,
	\end{align}
 where $\{y_j\}_{j=1}^n$ are the scalar responses, $\{X_j\}_{j=1}^n$ the functional covariates, $\{\eps_j\}_{j=1}^n$ the centered noise sequence, and $\{\beta^*_j\}_{j=1}^n$ the true slope functions. Here, we denote $\lb \beta ^*_j ,  X_j \rb_\lt 
 =  \int \beta_j^*(u) X_j(u) du  $.
 We assume that there exists a collection of time points  $\{\eta_k\}_{k=0}^{\Ktilde+1} \subset \{0,1,\ldots,n\}$ with $0 =\eta_0 <\eta_1 < \ldots <\eta_\Ktilde < \eta_{\Ktilde+1} = n $ such that
 \begin{equation}\label{eq: model2}
        \beta^*_j \neq \beta^*_{j+1} \quad \text{if and only if} \quad j \in \{\eta_1, \ldots,\eta_\Ktilde\}.     
 \end{equation}
We refer to the model specified in \eqref{eq: model} and \eqref{eq: model2}  as the functional linear regression model with change points.  Our goals are twofold: to estimate the locations of the change points consistently, and to derive the limiting distributions of these estimators and consequently construct an asymptotically valid confidence interval around each change point.

The considered problems are part of the vast body of change point analysis. The primary interest of change point analysis is to detect the existence of change points and estimate change points' locations in various data types. \citet{wangun2020univariate} and \citet{sullivan2000change} have addressed the detection of changes in the mean and covariance of a sequence of fixed-dimensional multivariate data, while \citet{samworth2018high} and \citet{kaul2023covariance} have focused on high-dimensional settings.
Change point problems have been investigated in various settings, such as signal processing \citep{tartakovsky2012signal}, regression \citep{SeoShin2016lasso}, networks \citep{barnett2016network}, and factor models \citep{Chofactormodels}, to name but a few.
In functional settings, \citet{Dette} studied the detection of changes in the eigensystem, while \citet{LiLiHsing2022functional}, \citet{harris2022scalable}, and \citet{carlos2022change} considered problems related to detecting changes in the mean and  \citet{jiao2023covariance} that in the covariance. Change point detection problems within this context have also been investigated in the Bayesian framework, e.g.~\citet{ghoshal2021bayesian}.
Beyond estimation of change points, recently, the limiting distributions of change point estimators  have been studied in high-dimensional regression \citep{xu2022HDregression, kaul2021inference}, multivariate non-parametric \citep{carlos2023change} as well as functional \citep{aue2009functional, aue2018functional} settings.

Despite these contributions, the estimation and inference of  change points in functional linear regression settings remain   unaddressed, and this paper aims to fill this gap.  
To this end, we first propose a two-step procedure based on RKHS, to detect and locate the multiple change points. 
Then, we investigate limiting distributions of change point estimators and introduce a new methodology to construct a confidence interval for each change point. 
This requires the estimation of long-run variance in the presence of temporal dependence which is of independent interest on its own, as highlighted by studies such as  \cite{LRVkhismatullina2020multiscale} and \citet{kokoszkaLRV}.

The framework adopted in this study is general, accommodating heavy-tailedness and temporal dependence in both functional covariates and
noise sequences.  More specifically, our methodology only requires the existence of sixth moments and a polynomial decay of $\alpha$-mixing coefficients for both functional covariates and
noise sequences, which greatly expands its applicability.  We also allow  for  the number of change points, denoted by $\Ktilde$, to diverge with the sample size.
 
\subsection{List of contributions} 
We briefly summarise the main contributions made in this paper below. 
\begin{itemize}
    \item To the best of our knowledge, our work is the first attempt at estimating and inferring change points in functional linear regression settings.
    Our theory only requires weak moment assumptions as well as accommodates temporal dependence and the number of change points to increase with the sample size. Besides the error bound for change point localization, we establish the corresponding minimax lower bound, thereby demonstrating the optimality of the proposed change point estimator.

    \item To facilitate  the practicability of our inference procedure, we introduce a block-type long-run variance estimator and prove its consistency. This estimator is subsequently employed to construct an asymptotically valid confidence interval for each change point.

    \item We demonstrate the numerical performance of our proposed method through extensive numerical examples and  real data analysis using Standard and Poor's 500 index datasets. Our approach numerically outperforms alternative change point estimation methods that rely on FPCA or high-dimensional regression methods.

\end{itemize}

\subsection{Basics of RKHS}
\label{sec:prelim}
This section briefly reviews the basics of RKHS that are relevant to functional linear regression.  We refer to  \citet{wainwrighthdsbook} for a detailed introduction to RKHS.

For any compact set $\T$, denote the space of square-integrable functions defined on   $\T$ as 
 $$\lt(\T) =  \bigg \{ f: \T \to \mathbb R  :  \|f\|_\lt^2= \int_\T f^2(u) du < \infty \bigg  \} .$$   
For any $f,g \in \lt(\T)$,  let 
$$\lb f, g \rb_\lt = \int_\T f(u)g(u)\, du . $$   
For a linear map $F$ from $\lt(\T) $ to $\lt(\T)$, define $\|F\|_{\mathrm{op}} = \sup_{\substack{ \|h\|_\lt   = 1 }}  \|  F(h)\| _\lt$. A kernel function $R: \T \times \T \to \mathbb{R}$ is a   symmetric  and nonnegative definite function. 
 The   integral operator $L_R $ of $R $ is a linear map from $ \lt(\T)  $ to $ \lt(\T)$   defined as  
 $$L_R(f)(\cdot) = \int_{\T} R(\cdot,u) f(u) \,du .$$ Suppose  in addition that $R$  is bounded. Then, Mercer's theorem (e.g.\ Theorem 12.20 of \cite{wainwrighthdsbook}) implies that there exists a set of orthonormal eigenfunctions $\{\psi^R_l\}_{l=1}^\infty  \subset \lt(\T )$ and a sequence  of  nonnegative eigenvalues $\{\theta^R_l\}_{l=1}^\infty$ sorted non-increasingly, such that
$$
\quad R(u_1,u_2) = \sum_{l= 1}^\infty  \theta^R_l \psi^R_l(u_1) \psi^R_l(u_2). 
$$
Thus, we have that 
$ L_R(\psi^R_l) =  \theta_l^R \psi^R_l$.
Define the RKHS generated by  $R$ as 
$$  \mathcal{H}(R) =\bigg \{ f \in \lt(\T) :   \|f \| _{\mathcal{H}(R)} ^2 =   \sum_{l=1}^\infty \frac{\lb f, \psi^R_l \rb_\lt ^2 }{\theta^R_l} < \infty  \bigg\} . $$
For any $ f, g \in \mathcal{H}(R)$, denote 
\begin{equation}\label{eq: RKHS defintion}
    \lb f, g \rb_{\mathcal{H}(R)} = \sum_{l =  1}^\infty  \frac{\lb f, \psi^R_l \rb_\lt\lb g, \psi^R_l \rb_\lt}{\theta^R_l}.    
\end{equation} 
Define   
$$R^{1/2}(u_1,u_2) = \sum_{l =1 }^\infty   \sqrt{\theta^R_l} \psi^R_l(u_1) \psi^R_l(u_2).$$
Thus, $L_{R^{1/2} }(\psi^R_l) =  \sqrt { \theta_l^R} \psi^R_l$.
It follows that    $L_{R^{1/2}}:  \lt(\T) \to \mathcal{H}(R)$  is bijective and distance-preserving.
In addition, if $\{ \Phi_l \}_{l=1}^\infty$ is a   $\lt(\T)$ basis, then  $ \{ L_{R^{1/2}}(\Phi_l )\}_{l=1}^\infty$ is a basis of $\mathcal{H}(R)$.
For any $f,g \in \lt(\T)$, denote $$R[f,g]= \iint_{\T\times \T} f(u_1)R(u_1,u_2)g(u_2)\; du_1 du_2.$$ Let $R_1$ and $ R_2$ be any generic kernel functions. We denote the composition of $R_1$ and $R_2$ as 
$$
R_1   R_2 (u_1,u_2)  = \int_\T R_1(u_1,v) R_2(v,u_2) dv.
$$

\subsection{Notation and organization}
For two positive real number  sequences $\{a_j\}_{j= 1}^\infty $ and $\{b_j\}_{j= 1}^\infty $, we write $a_j \lesssim b_j$ or $a_j = O(b_j)$ if there exists  an absolute positive constant $C$ such that $a_j \le Cb_j$. We  denote  $a_j \asymp b_j$, if $a_j \lesssim b_j$ and $b_j \lesssim a_j$. We write   $a_j = o\lft( b_j\rgt)$ if $\lim_{j\to \infty} b_j^{-1} a_j \to 0$. For a sequence of $\mathbb{R}$-valued random variables $\{X_j\}_{j= 1}^\infty$, we denote  $X_j = O_\mathbb P\left(a_j\right)$ if $\lim _{M \rightarrow \infty} \lim \sup _j \mathbb{P}\left(\left|X_j\right| \geq M a_j\right)=0$. We denote $X_j=o_\mathbb P\left(a_j\right)$ if $\lim \sup _j  \mathbb{P}\left(\left|X_j \right| \geq M a_j\right)=0$ for all $M>0$. The convergences in distribution and probability are respectively denoted by $\overset{\mathcal{D}}{\longrightarrow}$ and $\overset{P}{\longrightarrow}$. 
With slight abuse of notations, for any positive integers $s$ and $e$ where $0 \le s < e < n$, we use  $\se$ to denote the set $\se \cap \{1, \ldots, n\}$.

The rest of the paper is organized as follows. Section~\ref{sec: CP estimation} introduces our new methodology  for estimating multiple change points within functional linear regression settings. Section~\ref{sec: theory CP} studies the theoretical properties of the proposed estimators, establishing their minimax optimality and limiting distributions. In Section~\ref{sec: CI}, we discuss the construction of confidence intervals around each change point and provide an asymptotically valid  procedure   for the long-run variance estimation. Finally, Section~\ref{sec: NR} demonstrates the superior performance of our proposed method through its application to both simulated and real-world datasets, highlighting its advantages over potential competitors. The implementation of the proposed methodology can be found at \href{https://github.com/civamkr/FRBS}{\fcolorbox{green}{white}{\textcolor{black}{\fontfamily{pcr}\selectfont{https://github.com/civamkr/FRBS.}}}}

\section{Change point estimation}\label{sec: CP estimation}
 
In this section, we introduce  our method for change point   estimation under the functional linear regression model defined in \eqref{eq: model}. To motivate our approach, we first consider a closely related  two-sample testing problem in the functional linear regression setting. Given data $\{ (y_j, X_j) \}_{j=1}^n$ generated from \eqref{eq: model},  consider 
$$
H_0: \beta^*_{s+1} = \ldots = \beta^*_{e} \quad \text{ vs.} \quad H_a: \beta^*_{s+1} = \ldots = \beta^*_{t} \neq \beta^*_{t+1} = \ldots = \beta^*_{e},
$$
where $ 0<s<t<e \le n $. In other words, we are interested in testing whether there is a change in the slope function at time $t$ within the interval $(s,e]$.  The corresponding   likelihood ratio statistic~is 
\begin{align}\label{eq: likelihood}
    \what{W}_t^{s,e} = \frac{\max_{\beta\in \hk} \mathfrak{L}\lft( { \{ y_j, X_j\} }_{j=s+1}^e, \beta \rgt) }{\max_{   \beta_1\in \hk   } \mathfrak{L} \lft( { \{ y_j, X_j \}  }_{j=s+1}^t, \beta_1 \rgt) 
    \max_{  \beta_2 \in \hk  }\mathfrak{L}  \lft( { \{  y_j, X_j   \} }_{j=t+1}^e, \beta_2 \rgt)}
\end{align}
where, assuming for the moment that $\{\epsilon_j\}_{j = 1}^n$ are i.i.d.~standard normal, we have the likelihood function $\mathfrak{L}( { \{ y_j, X_j \} }_{j=s+1}^e, \beta ) = \prod_{j = s + 1}^e (2\pi)^{-1/2} e^{- ( y_j - \lb X_j, \beta\rb_\lt )^2/2 } $, and $\hk$ denotes RKHS corresponding to kernel $K$ defined in \Cref{assume: model assumption spectral} below.
 Note that  \eqref{eq: likelihood} can be further simplified to 
\begin{equation}\label{eq: definition Whset 2}
    \what{W}_t^{s,e} = \sum_{j=s+1}^e \lft( y_j - \lb X_j, \what{\beta}_{(s,e] } \rb_\lt \rgt)^2 - \sum_{j=s+1}^t \lft( y_j - \lb X_j, \what{\beta}_{(s,t] } \rb_\lt \rgt)^2 - \sum_{j=t+1}^e \lft( y_j - \lb X_j, \what{\beta}_{(t,e] } \rb_\lt \rgt)^2,
\end{equation}
where $\what{\beta}_{(s,e] }$ is the maximum likelihood estimator of the slope function based on $\{ (y_j, X_j) \}_{j=s+1}^e $. In practice,   finding $\what{\beta}_{(s,e] }$ is a challenging task, as  the intrinsic dimension of $\hk$ is infinite.  Inspired by \citet{caiyuan2012minimax}, we consider the following penalized estimator 
\begin{equation}\label{def: bhast}
    \bhase  = \argmin_{\beta\in \hk } \left\{ \frac{1}{(e-s)}\sum_{j\in \se }\left( y_j - \lb X_j, \beta \rb_\lt \right)^2  + \lambda_{e-s} \|\beta\|^2_\hk \right\},
\end{equation}
where $\lambda_{e-s}  $ is a tuning parameter to ensure the smoothness of the estimator. While \eqref{def: bhast} is an optimization problem in an infinite-dimensional space, the solution can be found in a finite-dimensional subspace via the representer theorem in RKHS \citep{caiyuan2010rkhs}, and is therefore statistically sound and numerically robust. In this case,
the population counterpart of \eqref{eq: definition Whset 2} is 
\begin{equation}\label{eq: signal_population}
    {W}_t^{s,e} =  \frac{(t-s)(e-t)}{(e-s)}\sx[ \beta^*_\st - \beta^*_\te, \beta^*_\st - \beta^*_\te],
\end{equation}
where $\beta^*_\se = (e-s)^{-1}\sum_{j = s+1}^e \beta^*_j$, and   
$ \Sigma $ is the covariance operator of $\{ X_j\}_{j=1}^\infty$, the centered and stationary covariate sequence,  i.e.\ 
$ \Sigma (u_1, u_2) = \E ( X_1(u_1)X_1(u_2))$.
To illustrate the effectiveness of the likelihood ratio statistics  $\what{W}_t^{s,e}$ in revealing the location of a change point, we demonstrate in \Cref{fig_cusum} that displays 
$\what{W}_t^{s,e}$ and its population counterpart $W_t^{s,e}$ in a situation where the interval $(s,e]$ contains a single change point at $\eta$.
We observe that $\what{W}_t^{s,e}$ closely approximates $W_t^{s,e}$, which is a `tent-shape' function in $t$ and is maximized at $\eta$, and thus $\what{W}_t^{s,e}$ attains its maximum close to~$\eta$ (in fact, exactly at $\eta$ in this example).

\begin{figure}[h]
\centering
\includegraphics[width=10cm]{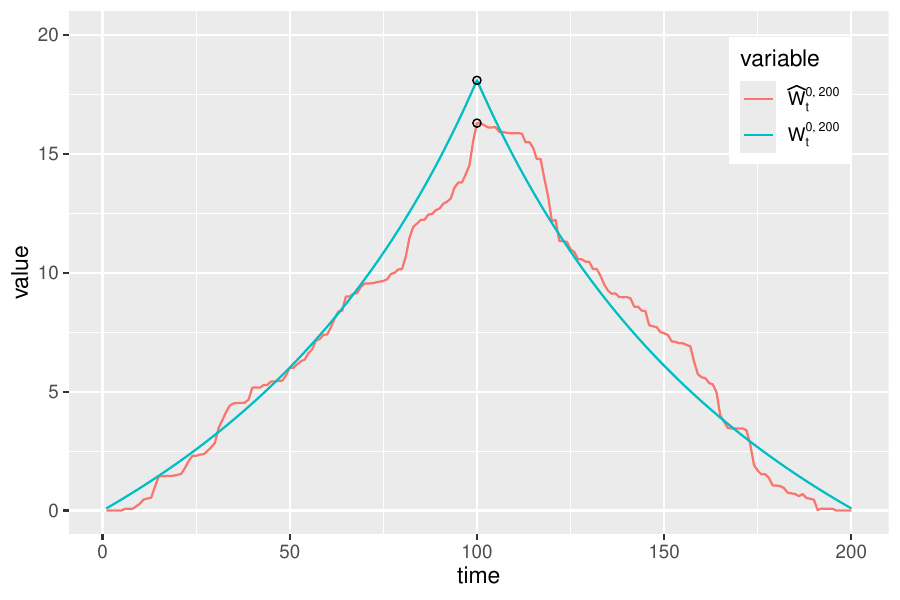}
\caption{Plot of $\widehat W_{t}^{s, e}$ and its population version ${W}_t^{s, e}$ with $s = 0$ and $e = 200$ over $t = 1, \dots, 199$. The data are simulated under Scenario~1 (\Cref{sec: NR_Simu}) with $n = 200$ and a change point occurs at $\eta = 100$. The estimator $\bhase$ is obtained with $\lambda_{e-s} = 0.2$. Both $\widehat W_{t}^{s, e}$ and ${W}_t^{s,e}$ achieve their maximum at $t = 100$.}
\label{fig_cusum}
\end{figure}

In what follows, we propose a two-step method for change point estimation in functional regression time series. In Step~$1$, we propose  a computationally efficient approach based on the statistic $\what{W}_t^{s,e}$, to generate preliminary change point estimators. Then, in Step~$2$, we utilize the preliminary estimators from Step~$1$ to develop the final estimators with  enhanced accuracy. 
 
\subsection*{Step~1: preliminary estimator}
\label{sec: SBS}
In Step~$1$, our goal is to consistently estimate change points with computational efficiency. Our approach utilizes the seeded binary segmentation first proposed in \citet{kovacs2023seeded}, which, with a set of deterministic intervals drawn as in Definition~\ref{def: Seeded Intervals} below, systematically scans for change points in the data at multiple resolutions.

\begin{definition}[Seeded intervals]
\label{def: Seeded Intervals} Let $n$ be the length of a given time series and $\Delta$ a given integer satisfying $0 < \Delta < n$.
Letting $M= \lft\lceil \log_{2} \lft( n/\Delta \rgt) \rgt\rceil + 1$ be the total number of layers, denote  
$\mathfrak{l}_k = n/2^{k-1}$ and $\mathfrak{b}_k = \mathfrak{l}_k/2 = n/2^k$,
for the layer index $k = 1, \ldots, M$. 
Then, the collection of seeded intervals is defined as 
\begin{align*}
\J = \bigcup_{k=1}^{M} \J_k \text{ \ where \ }  \J_k = \bigcup_{i=1}^{2^k - 1} \Big\{ \lft[ \big\lceil(i-1)\mathfrak{b}_k \big\rceil, \big\lfloor(i - 1)\mathfrak{b}_k + \mathfrak{l}_k \big\rfloor \Big]\rgt\},
\end{align*}
where $\J_k$ is the seeded intervals in the $k$-th layer.
 
\end{definition}
The total number of the intervals in $\mathcal{J}$ is bounded from above by
\begin{align}
|\J| = \sum_{k=1}^{ \lft\lceil \log_{2} \lft(\frac{n}{\Delta} \rgt) \rgt\rceil + 1} (2^k - 1) = 2^{\lft\lceil \log_{2} \lft(\frac{n}{\Delta} \rgt) \rgt\rceil + 2} - 3 - {\lft\lceil \log_{2} \lft(\frac{n}{\Delta} \rgt) \rgt\rceil} \le 8 \lft( \frac{n}{\Delta}\rgt).
\label{rm: card J}
\end{align}
 Note that by the construction of $\mathcal J$,
 each change point $\eta_k$ is contained in exactly two intervals that belong to the last layer of the seeded intervals, $\mathcal{J}_M$. We exploit this fact in Step~$2$ for the refined change point estimator.

\Cref{alg: FBS} outlines the procedure of computing the preliminary change point estimators, which is called with $(s, e] = (0, n]$. This algorithm recursively detects change points based on the likelihood ratio statistics $\{\Whsetm, \, s_m < t < e_m\}$ defined in \eqref{eq: definition Whset 2}. 
Specifically, using the set of seeded intervals, the algorithm iteratively identifies the shortest interval associated with a strong signal for a change (in the sense that $\Whsetm$ exceeds a threshold~$\tau$), an idea first proposed by \citet{baranowski2019narrowest} for detecting multiple change points in the mean of a univariate time series.
Upon detection of each change point, it stores the estimator and proceeds to search for further change points separately within the sections of the data determined by two consecutive estimators previously detected.
In the absence of a change point within a data section $(s, e]$, we expect all $\Whsetm, \, s < t < e$, to fall below the given threshold $\tau$, in which case the algorithm excludes the interval $(s, e]$ from further consideration. 
In addition to the threshold $\tau$, \Cref{alg: FBS} requires the choice of the regularization parameter $\lambda_{e, s}$ for the local estimation of the slope function, which takes the form $\lambda_{e-s} = \omega (e - s)^{-2r/(2r + 1)}$ with some $\omega > 0$ and $r$ that controls the regularity of the regression coefficient (see \Cref{assume: model assumption spectral}).
The choice of these tuning parameters are discussed in \Cref{sec: NR}.  
\begin{algorithm}[h]
 \caption{Functional Regression Binary Segmentation: FRBS$\lft(\se,\J, \omega, \tau\rgt)$.}\label{alg: FBS}
 \begin{algorithmic}
    \INPUT{Data $\{\lft(y_j, X_j \rgt)\}_{j=1}^n$},  seeded intervals $\J$, tuning parameters $\omega$, $\tau > 0$.
    \initialize If $\se = (0,n]$, the estimated change point set $\what{\B} = \emptyset$.

    \For{$ (s_m, e_m] \in \J$}
        \State compute $\Whsetm $ with $\lambda_{e_m-s_m} = \omega (e_m-s_m)^{-\frac{2r}{2r+1}} $
        \If{$(s_m, e_m] \subset (s,e]$}
            \State$b_m \gets \argmax_{s_m<t<e_m}  \Whsetm $
            \State$a_m \gets  \widehat{W}^{s_m,e_m}_{b_m} $
        \Else
        \State$a_m = 0$
        \EndIf
    \EndFor
    \State $\mathcal{M}^{s,e} := \{m: a_m>\tau\}$
    \If{$\mathcal{M}^{s,e} \neq \emptyset$}
        \State$m^* \gets \argmin_{m\in \mathcal{M}^{s,e}} |e_m - s_m| $
        \State$\widehat{\B} \gets \widehat{\B} \cup \{b_{m^*}\}$
        \State FRBS $\lft((s, b_{m^*}], \J, \omega, \tau\rgt)$
        \State FRBS $\lft((b_{m^*}, e], \J, \omega, \tau\rgt)$
    \EndIf
    \OUTPUT{$\widehat{\B} $}. 
 \end{algorithmic}
\end{algorithm}

\subsection*{Step~2: refined estimator}\label{sec: Final CP estimator}

Let $\widehat{\B} = \{\etahat_k, \, 1 \le k \le \Khat: \, \etahat_1 < \ldots < \etahat_{\Khat} \}$ denote the set of preliminary change point estimators returned by \Cref{alg: FBS}. 
In this step, we produce the refined estimators $\{\wtilde{\eta}_k\}_{k=1}^\Khat$ with enhanced accuracy.

For each $k = 1, \ldots, \Khat$, let  $(s_k^\1, e_k^\1]$ and $(s_k^\2, e_k^\2]$ be the two seeded intervals  in $\J_M$
that contains $\weta_k$, where $\J_M$ is the last layer of seeded intervals from \Cref{def: Seeded Intervals}.  
Further, we write 
\begin{equation}\label{eq: def s_k and e_k}
    s_k = \min\{s_k^\1, s_k^\2\}, \qquad e_k = \max\{e_k^\1, e_k^\2\}.
\end{equation}
The localization performance of our preliminary estimator ensures that, with high probability, the interval $(s_k, e_k]$ contains one and only one change point $\eta_k$ and it is sufficiently large (see~Lemma~C.3 in appendix).  
Over such $(s_k, e_k]$, we apply the following refinement procedure to further enhance the accuracy of the change point estimator, which also enables the investigation into the asymptotic distribution of the resultant estimator.  
For each $k$, let 
\begin{align}\label{eq: local refinement}
    \widetilde{\eta}_k &= \argmin_{s_k < t < e_k} \mathcal{Q}_k(t), \quad \text{where}
    \\
 \mathcal{Q}_k(t) &=   \sum_{j=s_k+1}^t \lft( Y_j - \lb X_j, \widehat{\beta}_{(s_k, \what{\eta}_k] } \rb_\lt \rgt)^2 + \sum_{j=t+1}^{e_k} \lft( Y_j - \lb X_j, \widehat{\beta}_{(\what{\eta}_k, e_k] } \rb_\lt \rgt)^2. \nonumber 
\end{align}
As shown later in Section~\ref{sec: Limiting distribution}, the refined estimator $\widetilde{\eta}_k$ attains the rate of localization matching the minimax lower bound, and thus is minimax optimal.

\section{Theoretical properties}\label{sec: theory CP}

In this section, we establish theoretical properties of the change point estimators proposed in 
\Cref{sec: CP estimation}.
We first introduce the required assumptions for the model~\eqref{eq: model}--\eqref{eq: model2}, which permit temporal dependence and heavy-tailedness in the data.

To quantify the degree of temporal dependence, we adopt the $\alpha$-mixing coefficient which is a standard tool commonly used in the time series literature. Recall that a stochastic process $\{Z_t\}_{t\in \mathbb{Z}}$ is said to be $\alpha$-mixing (strong mixing) if 
$$
\alpha(k) = \sup_{t \in \mathbb{Z}}\alpha\lft( \sigma(Z_s, s\le t), \, \sigma(Z_s, s\ge t + k)  \rgt) \to 0
$$
as $k \to \infty$, where we write $\alpha(\mathcal{A}, \mathcal{B}) = \sup_{A\in \mathcal{A}, B \in \mathcal{B}} \left| \mathbb P (A\cap B) - \mathbb P (A)\mathbb P (B) \right|$ for any two $\sigma$-fields $\mathcal{A}$ and $\mathcal{B}$. 
\Cref{assume: model assumption flr CP} concerns the distributions of the functional covariates and the noise sequence.  
\begin{assumption}\label{assume: model assumption flr CP}
 \ 
 \\
    (i)  The functional  covariate sequence $\{X_j\}_{j =1}^n \subset  \lt(\T)$  satisfies $\E[X_j] = 0$, $\E[\|X_j\|^2_\lt] < \infty$, and for any $f\in \lt(\T)$, there exists some constant $c > 0$ such that
    $$
    (\E[\lb X_j, f\rb_\lt^6] )^{1/6} \le c ( \E[\lb X_j, f \rb_\lt^2] )^{1/2}.
    $$
\
\\
   (ii) The noise sequence $\{\varepsilon_j\}_{j  = 1}^n \subset \mathbb R  $ satisfies $\E[\eps_j|X_j] = 0$ and $\E[\eps_j^6|X_j] < \infty$. 
 \
 \\
   (iii) The sequence $\{(X_j, \eps_j)\}_{j=1}^n$ is stationary and $\alpha$-mixing with the mixing coefficients sastisfying $\sum_{k=1}^\infty k^{1/3}\alpha^{1/3}(k) < \infty$.
\end{assumption}
Under \Cref{assume: model assumption flr CP}, the functional  covariate and the noise sequences are allowed to possess heavy tails. In particular, \Cref{assume: model assumption flr CP}~(i) assumes that the $6$-th moment of the random variable $\langle X_j,f \rangle_\lt $ is bounded  by its second moment, which  holds if e.g.\ each $X_j$ is a Gaussian process.  
Similar assumptions on the moments of the functional covariate are made in \citet{caiyuan2012minimax} for the investigation into the penalized slope estimator in~\eqref{def: bhast} in the stationary setting.  
The $\alpha$-mixing condition essentially requires that $\alpha(k) = o(1/k^4)$, allowing the mixing coefficient to decay at a polynomial rate. 
 
\begin{assumption}\label{assume: model assumption spectral}
 \ 
 \\
    (i)  
  The slope function satisfies $\beta^*_j \in \hk$, for all $j = 1, \ldots, n$, where $\hk$ is the RKHS generated by  the  kernel function $K$.
\
\\
(ii) It holds that 
    $$
    \kh\sx\kh (t_1,t_2) = \sum_{l\ge1} \mathfrak{s}_l \phi_l(t_1) \phi_l(t_2),
    $$ 
    where $\{\phi_l\}_{l =1}^\infty $ are the eigenfunctions and $\{\mathfrak{s}_l\}_{l =1}^\infty $ the corresponding eigenvalues satisfying    $\mathfrak{s}_l \asymp l^{-2r}$ for some   constant $r>1$.

\end{assumption}

\Cref{assume: model assumption spectral}~(i) requires  that  the slope functions are in the RKHS generated by the kernel function $K$, which regularizes the smoothness of the slope function. 
\Cref{assume: model assumption spectral}~(ii) requires that the   function $ \kh\sx\kh $ admits an eigen-decomposition with polynomial-decaying eigenvalues, which controls the regularity of regression prediction.  Both \Cref{assume: model assumption spectral}~(i) and (ii) are also found in  \citet{caiyuan2012minimax}. 
 
Under the model~\eqref{eq: model},  we define the change size of two consecutive slope functions as 
\begin{align}
\label{eq:kappa}
\kappa^2_k = \Sigma[ \beta^*_{\eta_k} - \beta^*_{\eta_{k+1}}, \beta^*_{\eta_k} - \beta^*_{\eta_{k+1}} ].
\end{align}
The form of $\kappa^2_k$ is closely related to the population counterpart (defined in~\eqref{def: bhast}) of the likelihood ratio statistics $\Whset$ in~\eqref{eq: definition Whset 2}. In fact, if the time interval $(s,e]$ contains only one change point $\eta_k$, the statistic $\widehat W^{s,e}_{\eta_k} $ converges asymptotically to ${W}_{\eta_k}^{s,e}$, which in turn satisfies
 \begin{equation}\label{eq: signal}
    {W}_{\eta_k}^{s,e} = \frac{(\eta_k -s)(e-\eta_k )}{(e-s)}\sx[ \beta^*_ {(s,\eta_k]  }- \beta^*_{( \eta_k,t ]  }, \beta^*_{(s,\eta_k]  }  - \beta^*_{( \eta_k,t ]  } ] = \frac{(\eta_k -s)(e-\eta_k )}{(e-s)} \kappa^2_k.
\end{equation} 

The detectability of each change point $\eta_k$ depends on both the change size $\kappa_k$ and how far it is from the adjacent change points. 
Therefore, we define the minimal change size and the minimal spacing of change points as
$$\kappa = \min_{1 \le k \le \Ktilde} \kappa_k \quad \text{and} \quad \Delta = \min_{1 \le k \le \Ktilde + 1} (\eta_k - \eta_{k - 1}),$$   
respectively.
\Cref{assume: SNR Signal to Noise ratio} specifies the signal-to-noise condition for the consistency of our method in terms of $\kappa$ and $\Delta$.
\begin{assumption}\label{assume: SNR Signal to Noise ratio}
Suppose that   
    $$
     \min\lft\{ \frac{\kappa^2 \Delta^2}{n \cdot n^{1/(2r+1)} \log^{1+ 2\xi}(n)}, \; \frac{\kappa^2 \Delta^{r/(2r+1)} }{ \log^{1+ 2\xi}(\Delta)} \rgt\}\to \infty,
    $$
    where $\xi >0$ is some constant and $r$ is defined in \Cref{assume: model assumption spectral}.
\end{assumption}

To establish the consistency of the preliminary estimators in \Cref{thm: Consistency}, it is sufficient to have 
$$\kappa^2 \Delta^2/(n \cdot n^{1/(2r+1)} \log^{1+ 2\xi}(n) ) \to \infty . $$
The additional requirement that $$\kappa^2 \Delta^{r/(2r+1)}/ \log^{1+2\xi}(\Delta) \to \infty,$$  is required  to derive the limiting distribution of the refined estimator in \Cref{thm: Limiting distribution}.
When $\Delta$ is of the same order as $n$, since $r > 1$, the first condition in \Cref{assume: SNR Signal to Noise ratio} dominates and it is simplified to $\kappa^2 \Delta^{2r/(2r+1)}  \to \infty$. Similar assumptions have been employed in \cite{padilla2021optimal, madrid2024change} for nonparametric change point analysis, where the smoothness of the density function plays a similar role as $r$.

\subsection{Consistency of the preliminary estimator}\label{sec: consistency} 
 
 We first present the main theorem establishing the consistency of  \Cref{alg: FBS}  and the associated rate of localization.
 
\begin{theorem}\label{thm: Consistency}
Suppose that Assumptions \ref{assume: model assumption flr CP}, \ref{assume: model assumption spectral}, and \ref{assume: SNR Signal to Noise ratio} hold. 
 Let $c_{\tau,1} > 32$ and $c_{\tau,2}\in (0,1/20)$ denote absolute constants.
 Suppose that $\tau$ satisfies
\begin{equation}\label{eq: range of tau}
    c_{\tau,1} \lft( \frac{n}{\Delta} \rgt) n^{1/(2r+1)} \log^{1+2\xi}(n) < \tau < c_{\tau,2} \kappa^2 \Delta,
\end{equation}
where $r$ and $\xi$ are defined in Assumptions \ref{assume: model assumption spectral} and \ref{assume: SNR Signal to Noise ratio}, respectively, and that $\omega >1/2$ is any finite constant.
Also,  let $\mathcal J$ be  seeded intervals constructed according to  \Cref{def: Seeded Intervals} with $\Delta$ defined in \Cref{assume: SNR Signal to Noise ratio}.
Then, $\text{FRBS}( (0,n], \mathcal J, \omega, \tau)$ outputs $ \widehat{\B} = \{\etahat_k\}_{k=1}^\Khat$ which satisfies
\begin{align*}
    \mathbb{P}\bigg( \Khat = \Ktilde; \, \max_{1 \le k \le \Ktilde} \kappa_k^2 \lft| \etahat_k - \eta_k  \rgt| \le C_1 \lft(\frac{n}{\Delta}\rgt) \Delta^{1/(2r+1)} \log^{1+ 2\xi}(n) \bigg) \to 1 \text{ \ as \ } n \to \infty,
\end{align*}
where $2<C_1< c_{\tau,1}/16$.
\end{theorem}
\Cref{thm: Consistency} shows that uniformly for all $ k =1,\ldots \mathcal K$,
$$|\what{\eta}_k - \eta_k | = O_{\mathbb P}\lft( \kappa_k^{-2} \lft(\frac{n}{\Delta}\rgt) \Delta^{1/(2r+1)} \log^{1+ 2\xi}(n)\rgt)  = o_{\mathbb P} (\Delta),$$
where the last equality follows from  \Cref{assume: SNR Signal to Noise ratio}.

\subsection{Consistency and the limiting distributions of  the refined estimators}
\label{sec: Limiting distribution}
In this subsection, we analyze the 
 consistency and the limiting distributions of the refined change point estimators.  In particular, we  demonstrate  that the limiting distribution of the refined change point estimator $\wtilde\eta_k$ is divided into two regimes determined  by the change  size $\kappa_k$: (i) the non-vanishing regime where $\kappa_k \to \varrho_k$  for some positive constant  $\varrho_k > 0$; and (ii) the vanishing regime where $  \kappa_k \to 0$.

\begin{theorem}\label{thm: Limiting distribution}
    Suppose that Assumptions \ref{assume: model assumption flr CP}, \ref{assume: model assumption spectral}, and \ref{assume: SNR Signal to Noise ratio} hold. Let $\lft\{ \wtilde{\eta}_k \rgt\}_{k=1}^\Khat$ denote the refined change point estimators obtained as in  \eqref{eq: local refinement} and assume that $\Khat = \mathcal{K}$.
    \begin{enumerate}[label=\Alph*.]
         \item (Non-vanishing regime)  
    For any given $k\in \{1, \ldots, \Khat\}$, suppose $ \kappa_k \to \varrho_k$ as $n\to \infty $, with $\varrho_k > 0$ being an absolute constant. Then    $\lft| \wtilde{\eta}_k - \eta_k \rgt| = O_\mathbb P(1)$.
    In addition,  as $n \to \infty$,
             $$
             \wtilde{\eta}_k - \eta_k \xrightarrow{\mathcal{D}} \argmin_{\gamma\in \mathbb{Z}} S_k(\gamma)
             $$
             where, for $\gamma \in \mathbb{Z}$, $S_k(\gamma)$ is a two-sided random walk defined as
             \begin{align*}
                 S_k({\gamma}) = 
                 \begin{cases}
                        \sum_{j = \gamma }^{-1} \lft\{ -2 \varrho_k \lft\lb X_j, \Psi_k \rgt\rb_\lt \eps_j + \varrho_k^2 \lft\lb X_j, \Psi_k \rgt\rb_\lt^2 \rgt\} & \text{for \ } \gamma<0,
                        \\
                        0 &  \text{for \ } \gamma = 0,
                        \\
                        \sum_{j=1}^{\gamma} \lft\{ 2 \varrho_k \lft\lb X_j, \Psi_k \rgt\rb_\lt \eps_j + \varrho_k^2 \lft\lb X_j, \Psi_k \rgt\rb_\lt^2 \rgt\} & \text{for \ } \gamma > 0,
                 \end{cases}
             \end{align*}
             with
    $$
    \Psi_k = \lim_{n\to \infty} \frac{\beta^*_{\eta_{k+1}}- \beta^*_{\eta_{k}}}{\sqrt{\Sigma\lft[ \beta^*_{\eta_{k+1}}- \beta^*_{\eta_{k}}, \beta^*_{\eta_{k+1}}- \beta^*_{\eta_{k}} \rgt]}}.
    $$

        \item (Vanishing regime) For any given $k\in \{1, \ldots, \Khat\}$, suppose  $\kappa_k \to 0$ as $n\to \infty $.   Then $\lft| \wtilde{\eta}_k - \eta_k \rgt| = O_\mathbb P (\kappa_k^{-2})$.
           In addition,  as $n \to \infty$,
             $$
             \kappa_k^2(\wtilde{\eta}_k - \eta_k) \xrightarrow{\mathcal{D}} \argmin_{\gamma\in \mathbb{R}} \{ |\gamma| + \sigma_\infty(k) \mathbb{W}(\gamma)\},
             $$
             where
\begin{equation}\label{eq: def long-run variance}
        \sigma_\infty^2(k) = 4 \lim_{n\to \infty} \text{Var}\lft( \frac{1}{n} 
 \sum_{j=1}^n \frac{\lb X_j, \beta^*_{\eta_{k}} - \beta^*_{\eta_{k+1}}  \rb_\lt \eps_j }{\kappa_k}  \rgt),
    \end{equation} 
    and $\mathbb{W}(\gamma)$ is a two-sided standard Brownian motion defined as
             \begin{equation*}
                  \mathbb{W}(\gamma) = 
                 \begin{cases}
                        \mathbb{B}_1(-\gamma) & \text{for \ } \gamma<0,
                        \\
                        0 & \text{for \ } \gamma = 0,
                        \\
                        \mathbb{B}_2(\gamma) & \text{for \ } \gamma > 0,
                 \end{cases}
             \end{equation*}
             with $\mathbb{B}_1(r)$ and $\mathbb{B}_2(r)$ denoting two independent standard Brownian motions.
         \end{enumerate}

\end{theorem}

 \Cref{thm: Limiting distribution} establishes the 
localization error bound with rate $\kappa_k^{-2}$ 
for the refined change point estimator as well as the corresponding limiting distributions. 
 The  localization error  bound   in  \Cref{thm: Limiting distribution} significantly improves upon that  attained by the preliminary estimator derived in \Cref{thm: Consistency}. Note that \Cref{thm: Limiting distribution} assumes $\Khat = K$, which holds asymptotically with probability tending to one by \Cref{thm: Consistency}. We make a similar condition in Theorems~\ref{thm: long run variance estimation} and \ref{thm: CI vanishing} below.
 
 In the following \Cref{lemma: lower bound}, we  further provide a matching lower bound to show that   the locatization error rate established in \Cref{thm: Limiting distribution} is minimax optimal.

\begin{lemma}\label{lemma: lower bound}
    Let $\{(y_j, X_j)\}_{j=1}^n$ be a functional regression time series following the models in~\eqref{eq: model}--\eqref{eq: model2} with $\Ktilde = 1$, and suppose that   Assumptions \ref{assume: model assumption flr CP} and \ref{assume: model assumption spectral} hold.
    Let $\mathbb P^n_{\kappa,\Delta}$ be the corresponding joint distribution. For any diverging sequence $\rho_n \to \infty $, consider the class of distributions
    $$
    \mathfrak{P} = \lft\{\mathbb P^n_{\kappa,\Delta}\; :\; \min\lft\{ \frac{\kappa^2 \Delta^2}{n \cdot n^{1/(2r+1)} \log^{1+ 2\xi}(n)}, \; \frac{\kappa^2 \Delta^{r/(2r+1)} }{ \log^{1+ 2\xi}(\Delta)} \rgt\} > \rho_n   \rgt\}.
    $$
    Then for sufficiently large  $n$, it holds that 
    $$
    \inf_{\weta}\; \sup_{\mathbb P \in \mathfrak{P}} \E\lft[ |\weta - \eta(\mathbb P)| \rgt] \ge  \frac{1}{\kappa^2 e^{2}}.
    $$
\end{lemma}
The class of distributions $\mathfrak{P}$ encompasses all possible scenarios where \Cref{assume: SNR Signal to Noise ratio} is satisfied.
\Cref{lemma: lower bound}  complements the upper bound   established in \Cref{thm: Limiting distribution} in both the vanishing and the non-vanishing regimes.  The matching bounds in \Cref{thm: Limiting distribution} and \Cref{lemma: lower bound}   indicate  that our refined estimator is minimax optimal.

\section{Confidence interval for the change points}\label{sec: CI}

In this section, we provide a practical way to construct confidence intervals for the true change points under the vanishing regime based on the limiting distribution derived  
in \Cref{thm: Limiting distribution}B. Since the limiting distribution in the vanishing regime contains an unknown long-run variance, we study a consistent estimator before proposing our method for constructing confidence intervals for true change points. 
 
\subsection{Long run variance estimation}

To utilize the limiting distribution  in the vanishing regime derived in \Cref{thm: Limiting distribution}B, we first need to consistently estimate the long-run variance $\sigma_\infty^2(k)$ defined in~\eqref{eq: def long-run variance}.
The long-run variance depends on the size of change $\kappa_k$ at the change point $\eta_k$ as defined in \eqref{eq:kappa}. To this end, we propose the  plug-in estimator 
\begin{equation}\label{eq: kappa hat}
        \what{\kappa}_k = \sqrt{\what{\Sigma}_{(s_k,e_k]} \lft[ \what{\beta}_{(s_k, \weta_k]} - \what{\beta}_{(\weta_k, e_k]}, \what{\beta}_{(s_k, \weta_k]} - \what{\beta}_{(\weta_k, e_k]} \rgt]},
\end{equation}
where  $s_k$ and $e_k$ are defined in  \eqref{eq: def s_k and e_k}, $\what{\beta}_{(s_k, \weta_k]} $ and $\what{\beta}_{(\weta_k, e_k]}$ are obtained in \eqref{def: bhast}, and 
$$ \widehat \Sigma_{(s_k, e_k)} (u_1, u_2) = \frac{1}{e_k-s_k}\sum_{j=s_k+1}^{e_k} X_j(u_1) X_j(u_2)$$
is the sample covariance operator for the functional data $\{ X_j\}_{j=s_k+1}^{e_k}$.
We show the consistency of~$\what{\kappa}_k$ in Lemma~D.1 in appendix. Note that for simplicity we use $\what{\beta}_{(s_k, \weta_k]} $ and $\what{\beta}_{(\weta_k, e_k]}$, which are the by-products of \Cref{alg: FBS}. It is also possible to construct a consistent estimator of $\kappa_k$ using the refitted slope functions after obtaining the refined change point estimators $\widetilde \eta_k$. For the same reason, we also use $\what{\beta}_{(s, e]} $ in \Cref{alg: LRV} to estimate the long-run variance.

For the estimation of $\sigma_\infty^2(k)$, we make use of  a block-type strategy which has previously been adopted by \cite{casini2021minimax} for the estimation of the long-run variance in a fixed-dimensional time series setting.  
In \Cref{alg: LRV}, we outline our proposal for the estimation of $\sigma_\infty^2(k)$.
Our proposed method  first partitions the data into mutually disjoint blocks of size $2q$ for some positive integer $q$, and filters out the intervals  that contain change point estimators and which are adjacent to them.
This filtering ensures that with high probability, the remaining intervals do not contain any change point. Let us denote the set of remaining intervals by $\mathcal P$.
For each given $\mathcal I = (m, m + 2q] \in \mathcal P$, we first compute the  statistic   
$$Z_j = \what{\kappa}_k^{-1}\cdot \lft\lb X_j, \what{\beta}_{(s_k, \weta_k]}  - \what{\beta}_{( \weta_k, e_k]} \rgt\rb_\lt \cdot \lft( y_ j  - \lft\lb X_j, \what{\beta}_{( m,  m+2q] } \rgt\rb_\lt \rgt)$$
at each $j \in \mathcal I$,
which approximates the sequence $\kappa_k^{-1} \lb X_j, \beta^*_{\eta_{k}} - \beta^*_{\eta_{k+1}} \rb_\lt \eps_j$. 
Then, we compute the scaled sample average of the centered sequence $Z_j - Z_{j + q}$, $j = m + 1, \ldots, m + q$, and denote it by $F_{\mathcal I}$.
The estimator $\what{\sigma}^2_\infty(k)$ is obtained as the average of the square of $F_{\mathcal I}$ over $\mathcal I \in \mathcal P$.

\begin{algorithm}[h]
 \caption{Long run variance estimation: LRV($\{\what{\beta}_{(s,e]}\}$, $\{\what{\kappa}_k\}_{k=1}^{\Khat}$, $\{\wtilde{\eta}_k\}_{k=1}^\Khat$,  $q $, $\omega$)}\label{alg: LRV}
 \begin{algorithmic}
    \INPUT{ Preliminary estimators $\{\what{\kappa}_k\}_{k=1}^{\Khat}$    and 
    $\{\wtilde{\eta}_k\}_{k=1}^\Khat$, tuning parameter $q  \in \mathbb Z^+$}.
    \State $\mathcal N \gets  \{1, \ldots, \lfloor n/2q \rfloor \} - \bigcup_{k=1}^\Khat \{\lfloor \wtilde{\eta}_k/2q \rfloor - 1,  \lfloor \wtilde{\eta}_k/2q \rfloor, \lfloor \wtilde{\eta}_k/2q \rfloor +1 \} $
   
    \State $\calP \gets    \bigcup_{i \in \mathcal N}  \{ ( 2q\cdot( i  - 1)  , \, 2q\cdot i ]  \} $ \Comment{$\mathcal P$ is a collection of mutually disjoint intervals in $(0,n]$. }
    
    \For{$ k = 1,\ldots, \Khat$}
    \State Compute  $s_k$ and $e_k$ as in \eqref{eq: def s_k and e_k},
        \For{$\mathcal I =(m , m +2q] \in \calP$}
            \For{$j  \in ( m ,  m +2q]  $}
                 \State $Z_j \gets \what{\kappa}_k^{-1}\cdot \lft\lb X_j, \what{\beta}_{(s_k, \weta_k]}  - \what{\beta}_{( \weta_k, e_k]} \rgt\rb_\lt \cdot \lft( y_ j  - \lft\lb X_j, \what{\beta}_{( m,  m+2q] } \rgt\rb_\lt \rgt)$ 
            \EndFor
            
            \State $F_{\mathcal I } \gets \sqrt{2} \,q^{-1/2} \lft\{ \sum_{j =m+1 }^{m+q} \big( Z_j - Z_{j+q} \big) \rgt\}$
        \EndFor
        \State $\what{\sigma}^2_\infty(k) \gets |\mathcal P| ^{-1} \sum_{\mathcal I  \in \calP} \lft(F_{\mathcal I }\rgt)^2 $
    \EndFor
    \OUTPUT{$ \lft\{\what{\sigma}^2_\infty(k) \rgt\}_{k=1}^\Khat$}. 
 \end{algorithmic}
\end{algorithm}

\begin{theorem}\label{thm: long run variance estimation}
    Suppose that all the assumptions of \Cref{thm: Limiting distribution} hold and that   $\Khat = \mathcal{K}$.  In   \Cref{alg: LRV}, let $\{ (s_k, e_k) \}_{k=1}^\Ktilde$ be defined as in \eqref{eq: def s_k and e_k}, $\{\wtilde{\eta}_k\}_{k=1}^\Khat$ be the refined estimators as in \eqref{eq: local refinement},  $\{\what{\kappa}_k\}_{k=1}^{\Khat}$ be  defined  as  in \eqref{eq: kappa hat},    and 
   $q$ be  an integer  satisfying 
    \begin{equation}\label{eq: long run q conditions}
        \lft( \frac{\log^{2+2\xi}(\Delta)}{\kappa^2} \rgt)^{\frac{2r+1}{2r-1}}  \ll q \ll \Delta.
    \end{equation}
    Denote by $ \lft\{\what{\sigma}^2_\infty(k) \rgt\}_{k=1}^\Khat$ the output of \Cref{alg: LRV}.  Then,
    for any given $k\in \{1, \ldots, \Khat\}$, $\what{\sigma}^2_\infty(k) \xrightarrow{\mathbb P} \sigma^2_\infty(k)$ as $n \to \infty$.
\end{theorem}

The choice of the tuning parameter $q$ needs to balance the bias   within each interval and the variance between all intervals in $\mathcal P$.  The practical choice of $q$ is outlined in \Cref{sec: NR_Simu}.

\subsection{Confidence interval  construction}\label{sec: CI construction}

In this subsection, we outline the practical procedure for constructing an asymptotically valid confidence interval in the vanishing regime for each change point. For any given $k\in \{1, \ldots, \Khat\}$ and confidence level $\alpha \in (0,1)$, the construction of a $100(1-\alpha)\%$ confidence interval for  $\eta_k$
is performed in two steps:
\\
\
\\
{\bf Step I.} \ Let $B \in \mathbb{N}$. For $b \in \{1, \ldots, B\}$, define
\begin{equation}\label{eq: uhat}
    \what{u}^{(b)}_k = \argmin_{r\in (-\infty , \infty )} \lft( |r| + \what{\sigma}_\infty(k) \mathbb W^{(b) } (r) \rgt)
\end{equation}
where $\what{\sigma}^2_\infty(k)$ is the long-run variance estimator obtained from \Cref{alg: LRV}, and
\begin{equation*}
    \mathbb W^{(b) } (r)  =
    \begin{cases}
    \frac{1}{\sqrt{n}} \sum_{j = \lfloor nr \rfloor}^{-1} z_j^{(b)} & \text{for \ } r <0,
    \\
    0 & \text{for \ } r =0,
    \\
    \frac{1}{\sqrt{n}} \sum_{j = 1}^{\lceil nr \rceil} z_j^{(b)} & \text{for \ } r > 0,
    \end{cases}
\end{equation*}
with  $\{z_j^{(b)}\}_{j=-\infty }^\infty $ being i.i.d.   standard normal random variables.
\\
\
\\
{\bf Step II.} Let $\what{q}_{ k,\alpha/2}  $ and $\what{q}_{ k,1 - \alpha/2} $ be the $\alpha/2$-quantile and $(1-\alpha/2)$-quantile  of the empirical distribution of $\{ \what{u}^{(b)}_k \}_{b=1}^B$. Then, the confidence interval for $\eta_k$ is constructed as
\begin{equation}\label{eq:conf interval}
    \lft[ \wtilde{\eta}_k - \frac{\what{q}_{ k,\alpha/2}  }{\what{\kappa}^2_k} , \;\; \wtilde{\eta}_k + \frac{\what{q}_{ k,1 - \alpha/2 }  }{\what{\kappa}^2_k} \rgt],
\end{equation}
where $\what{\kappa}^2_k$ is defined in \eqref{eq: kappa hat}.

\begin{theorem}\label{thm: CI vanishing}
    Suppose that all the assumptions of \Cref{thm: Limiting distribution} hold, and that $\Khat = \mathcal{K}$.
    For any given $k\in \{1, \ldots, \Khat\}$ and $b  = 1, \ldots, B $, let $\what{u}^{(b)}_k$ be defined as in~\eqref{eq: uhat}. 
    Then,  it holds that
    $$
    \frac{\kappa_k^2}{\what{\kappa}_k^2} \what{u}^{(b)}_k \xrightarrow{\calD} \argmin_{r \in \mathbb R} \big\{ |r| + \sigma_\infty(k) \mathbb W(r) \big\} \quad \text{as} \quad n \to \infty.
    $$
\end{theorem}
\Cref{thm: CI vanishing} implies  that the confidence intervals  proposed in \eqref{eq:conf interval} is asymptotic valid  in     the  {\it vanishing regime} considered in Theorem~\ref{thm: Limiting distribution}B.
Confidence interval construction under the non-vanishing regime remains a challenging problem as the limiting distribution involves random quantities of typically unknown distributions. There are some recent attempts on this problem \cite[e.g.][]{kaul2021inference,ng2022bootstrap,cho2022bootstrap}. However, to the best of our knowledge,  there are few theoretical studies  for confidence interval construction  under the non-vanishing regime in the presence of temporal dependence.  
 
\section{Numerical results}\label{sec: NR}

In this section, we perform numerical experiments on simulated and real datasets to investigate the performance of the proposed change point estimation and inference procedure, which contains three steps: (i) the preliminary estimation of the change points, (ii) the refinement of change point estimators and (iii) the construction of confidence intervals. Throughout, we refer to our combined procedure as  `FRBS'. 
\subsection{Simulation studies}\label{sec: NR_Simu}
\medskip
\noindent \textbf{Settings.} 
We modify the simulation settings of \cite{caiyuan2010rkhs} or \cite{caiyuan2012minimax} by introducing temporal dependence in $\{X_j\}_{j = 1}^n$ and changes in $\{\beta_j^*\}_{j = 1}^n$.  Specifically, we simulate data from the model described in \eqref{eq: model}, where the error process $\{\varepsilon_j\}_{j = 1}^n$ is a sequence of i.i.d.\ standard normal random variables, and $\{X_j\}_{j = 1}^n$ is a stationary process following 
\begin{align*}
	X_{j} &= \sum_{m = 1}^{50}\zeta_mZ_{m,j}\phi_m, \qquad 1\le j \le n,
\end{align*}
with $\phi_1 = 1$, $\phi_{m+1} = \sqrt{2}\cos(m\pi t)$ for $m \geq 1$ and $\zeta_m = (-1)^{m+1}m^{-1}$. For each $m \geq 1$, $\{Z_{m,j}\}_{j = 1}^n$ is independently generated as an autoregressive process, i.e.\ $Z_{m,j} = 0.3 Z_{m,j-1} + \sqrt{1-0.3^2} \cdot e_{m,j}$ with $e_{m,j} \overset{i.i.d.}{\sim} N(0,1)$. Note that $\zeta_m^2 = m^{-2}$ are the eigenvalues of the covariance function of $X_{j}$, and $\phi_m$ are the corresponding eigenfunctions.

\noindent Let
\begin{align*}
    \beta^{(0)} = 4\sum_{m = 1}^{50}(-1)^{m+1}m^{-4}\phi_m \quad \text{and} \quad \beta^{(1)} = (4-c_{\beta})\sum_{m = 1}^{50}(-1)^{m+1}m^{-2}\phi_m,
\end{align*}
where the coefficient $c_{\beta} \in \{0.5, 1\}$. 
We consider the slope functions
\begin{align*}
    \beta_j^* = \begin{cases}
        \beta^{(0)} & \text{for \ } j \in \{1, \dots, \eta_1\},\\
        \beta^{(1)} & \text{for \ } j \in \{\eta_1 + 1, \dots, \eta_2\},\\
        \vdots \\
        \beta^{\text{($\mathcal K$ mod $2$)}} & \text{for \ } j \in \{\eta_{\mathcal K} + 1, \dots, n\}.
    \end{cases}
\end{align*}
The cases with $c_{\beta} = 0.5$ and $c_{\beta} = 1$ correspond to the settings with small and large jump sizes, respectively.  We further assume that for each $j$, the random function $X_j$ is observed in an evenly spaced fixed grid with size $p = 200$.
We choose the reproducing kernel Hilbert space $\mathcal{H}(K)$ as the Sobolev space $$\mathcal{W}_2^1 = \{f \in L^2[0,1]\,:\, \|f^{(j)}\|_{\mathcal{L}^2} < \infty, \, j = 0, 1\},$$
with the corresponding reproducing kernel
\begin{align*}
    K(s,t) = \begin{cases}
        \frac{\cosh(s)\cosh(1-t)}{\sinh(1)} & 0 \leq s \leq t \leq 1,\\
        \frac{\cosh(t)\cosh(1-s)}{\sinh(1)} & 0 \leq t \leq s \leq 1.
    \end{cases}
\end{align*}
Note that the reproducing kernel and the covariance function of $X_j$ share a common ordered set of eigenfunctions \citep[see][]{caiyuan2012minimax}.

\medskip
\noindent \textbf{Evaluation measurements.} Let $\{\eta_k\}_{k = 1}^{\mathcal K}$ and $\{\widehat{\eta}_k\}_{k = 1}^{\widehat{\mathcal K}}$ be the set of true change points and a set of estimated change points, respectively. To assess the performance of different methods in localization, we report (i) the proportions (out of $200$ repetitions) of over- or under-estimating $\mathcal K$, and (ii) the average and the standard deviation of the scaled Hausdorff distances between $\{\eta_k\}_{k = 1}^{\mathcal K}$ and $\{\widehat{\eta}_k\}_{k = 1}^{\widehat{\mathcal K}}$ defined as
\begin{align*}
    d_{\mathrm{H}} = \frac{1}{n}\max\Big\{\max_{j = 0, \dots \Khat+1}\min_{k = 0, \dots, \mathcal K+1}|\widehat{\eta}_{j} - \eta_k|, \max_{k = 0, \dots , \mathcal K+1} \min_{j = 0, \dots, \widehat{\mathcal K}+1} |\widehat{\eta}_{j} - \eta_k| \Big\},
\end{align*}
where we set $\widehat{\eta}_0 = 1$ and $\widehat{\eta}_{\widehat{\mathcal K}+1} = n+1$. 
Given a confidence level $\alpha \in (0,1)$, we evaluate the performance of the proposed confidence intervals by measuring their coverage of $\eta_k$, defined as
\begin{align*}
    \text{cover}_k(1-\alpha)  = \mathbbm{1}\bigg\{ \eta_k \in \bigg[\widetilde{\eta}_k + \frac{\widehat{q}_u(\alpha/2)}{\widehat{\kappa}_k^{2}}, \, \widetilde{\eta}_k + \frac{\widehat{q}_u(1-\alpha/2)}{\widehat{\kappa}_k^{2}} \bigg]\bigg\},
\end{align*}
for each $k \in \{ 1, \dots, \mathcal K\}$.
To ensure the validity of the above definition, we compute the averaged coverage among all the repetitions where we obtain $\widehat{\mathcal K} = \mathcal K$.
\medskip

\noindent \textbf{Comparison.}
To the best of our knowledge, no competitor exists for the change point problem in the functional linear (scalar-on-function) regression setting we consider in this paper.  However, considering that functional covariates are typically observed as high-dimensional vectors, we adopt the estimation and inference procedure developed for change points in high-dimensional linear regression (referred to as `HDLR') \citep{xu2022HDregression} as a competitor.  Note that HDLR is analogous to FRBS in the sense that they are both two-step procedures producing preliminary and refined estimators. Thus, we compare their performance on both steps.
Additionally, we include a method that combines FPCA and the HDLR in estimating the change point location. More specifically, we first perform the FPCA on the functional covariates and then perform the HDLR using the $n \times K$ score matrix outputted by FPCA as the covariate matrix, which is referred to as `FPCA+LR'. To perform FPCA, we use the R package \texttt{fdapace} \citep{fdapace_R} with the default settings.
\medskip

\noindent \textbf{Selection of tuning parameters and estimation of unknown quantities.}
Four tuning parameters are involved in the proposed change point localization and inference procedures. These are the number of layers $M$ for the seeded intervals (see~\Cref{def: Seeded Intervals}), $\omega$ and $\tau$ for the FRBS algorithm (see~\Cref{alg: FBS}) and the block size $2q$ for long-run variance estimation (see~\Cref{alg: LRV}) in the confidence interval construction.  We set $M = \lceil \log_2(10) \rceil + 1$. 
In place of $\omega$, which is used in specifying $\lambda_{e - s}$, we propose to select a single $\lambda_{e - s} = \lambda$ along with the threshold $\tau$, adapting the cross-validation method proposed by \cite{rinaldo2021localizing}.
Specifically, we first divide $\{(y_j, X_j)\}_{j =1}^{n}$ into those with odd and even indices, respectively. For each possible combination of $\lambda \in \{0.1, 0.2, 0.3, 0.4, 0.5\}$ and $\tau \in \{1, 1.5, 2, 2.5, 3\} \times n^{2/5}$, we obtain the FRBS outputs ($\widehat{\mathcal{B}}$, $\widehat{\mathcal K}$ and $\{\widehat{\beta}_k\}_{k = 0}^{\widehat{\mathcal K}}$) based on the training set, and compute the least squared prediction error on the test set as the validation loss.  We select the combination of $\lambda$ and $\tau$ that minimize the validation loss. Following the discussion after \Cref{thm: long run variance estimation}, we set $q = \Big\lceil  \Big(\max_{1 \le k \le \widehat{\mathcal K}}\{e_k - s_k\}\Big)^{2/5}\Big/2 \Big\rceil$ with $\{(s_k, e_k)\}_{k = 1}^{\widehat{\mathcal K}}$ given in \eqref{eq: def s_k and e_k}.  We note that the simulation results remain robust against the choices of the tuning parameters $M$ and $q$.  

For the HDLR, we use the CV method in \cite{changepoints_R} to select the tuning parameters for the DPDU algorithm therein, with candidate sets $\lambda \in \{0.05, 0.1, 0.5, 1, 2, 3, 4, 5\}$ and $\tau \in \{5, 10, 15, 20, 25, 30, 35, 40\}$, and use the default values of the other tuning parameters.

\medskip
\noindent \textbf{Scenario: single change point.} Let $\mathcal K = 1$ and $\eta = n/2$.  We vary $n \in \{200, 400, 600, 800\}$, $c_{\beta} \in \{0.5, 1\}$ and fix $p = 200$. Tables \ref{tab:sce1_1} and \ref{tab:sce1_2} summarize the localization and inference performance of FRBS, HDLR and FPCA$+$LR. \Cref{tab:sce1_2} excludes the case with $n = 200$ where there are a large number of repetitions with mis-estimated $\mathcal{K}$ for all methods. 

In \Cref{tab:sce1_1}, comparing the Hausdorff distance computed with
the preliminary ($d_{\mathrm{H}}^{\mathrm{pre}}$) and the refined estimators ($d_{\mathrm{H}}^{\mathrm{fin}}$), we see that the refinement step improves the performance for all methods in consideration as $n$ increases and/or the jump size increases. 
The detection power improves with the sample size as evidenced by the decrease in the proportion of under-detection. 
At the same time, FRBS does not detect more false positives as the sample size increases, unlike HDLR and FPCA$+$LR. 
Overall, the proposed FRBS outperforms both competitors by a large margin, in its detection accuracy as well as localization performance, demonstrating the advantage of adopting a functional approach over the high-dimensional one of HDLR.
Although the RKHS and the covariance function of $X_j$ are well-aligned, the dimension reduction-based approach of FPCA$+$LR comes short of the RKHS-based FRBS.

\Cref{tab:sce1_2} shows that our proposed construction of confidence intervals performs well especially when the jump size is relatively high. In contrast, the intervals constructed based on HDLR perform poorly in capturing the change points, often with the intervals being too narrow.
All these observations suggest the benefit of adopting the proposed functional approach.

\begin{table}[htbp]
\centering
\caption{In Scenario 1, the proportions of under- ($\widehat{\mathcal K} < \mathcal K$),  over-detection ($\widehat{\mathcal K} > \mathcal K$), and  the average and standard deviation (in parentheses) of scaled Hausdorff distance   over 200 repetitions  are reported for FRBS, HDLR, and FPCA+LR. The single change point is located at $\eta = n/2$.}
\begin{tabular}{ccccc}
\hline
\multicolumn{5}{c}{$\mathcal K = 1$ and $p = 200$} \\
$n$   & $\widehat{\mathcal K} < \mathcal K$ & $\widehat{\mathcal K} > \mathcal K$ & $d_{\mathrm{H}}^{\mathrm{pre}}$ & $d_{\mathrm{H}}^{\mathrm{fin}}$ \\ \hline                     \multicolumn{5}{c}{FRBS, $c_{\beta} = 0.5$ (small jump size)}        \\
200 & 0.310   & 0.015         & 0.198 (0.212)  & 0.190 (0.216)  \\
400 & 0.105       & 0.045     & 0.091 (0.150)  & 0.084 (0.151)   \\
600 & 0.045       & 0.030     & 0.060 (0.111)  & 0.048 (0.108)   \\
800 & 0.005       & 0.020     & 0.033 (0.053)  & 0.020 (0.044)   \\
\multicolumn{5}{c}{FRBS, $c_{\beta} = 1$ (large jump size)}        \\
200 & 0   & 0.025         & 0.033 (0.053)  & 0.027 (0.045)  \\
400 & 0       & 0.025     & 0.018 (0.032)  & 0.013 (0.026)   \\
600 & 0       & 0.020     & 0.017 (0.037)  & 0.012 (0.035)   \\
800 & 0       & 0.010     & 0.013 (0.028)  & 0.009 (0.027)   \\
\hline
\multicolumn{5}{c}{HDLR, $c_{\beta} = 0.5$ (small jump size)}        \\
200 & 0.630   & 0.050         & 0.350 (0.196)  & 0.354 (0.195)  \\
400 & 0.275       & 0.070     & 0.200 (0.196)  & 0.201 (0.202)   \\
600 & 0.100       & 0.110     & 0.127 (0.146)  & 0.126 (0.159)   \\
800 & 0.080       & 0.105     & 0.112 (0.137)  & 0.105 (0.147)   \\
\multicolumn{5}{c}{HDLR, $c_{\beta} = 1$ (large jump size)}        \\
200 & 0.080   & 0.150         & 0.118 (0.138)  & 0.115 (0.151)  \\
400 & 0.005       & 0.180     & 0.066 (0.078)  & 0.063 (0.104)   \\
600 & 0       & 0.100     & 0.041 (0.061)  & 0.034 (0.076)   \\
800 & 0       & 0.135     & 0.040 (0.059)  & 0.033 (0.078)   \\
\hline
\multicolumn{5}{c}{FPCA+LR, $c_{\beta} = 0.5$ (small jump size)}        \\
200 & 0.715   & 0.020         & 0.385 (0.181)  & 0.386 (0.183)  \\
400 & 0.365       & 0.060     & 0.230 (0.212)  & 0.224 (0.222)   \\
600 & 0.180       & 0.070     & 0.150 (0.181)  & 0.140 (0.189)   \\
800 & 0.070       & 0.115     & 0.103 (0.134)  & 0.085 (0.147)   \\
\multicolumn{5}{c}{FPCA+LR, $c_{\beta} = 1$ (large jump size)}        \\
200 & 0.125   & 0.100         & 0.130 (0.156)  & 0.126 (0.174)  \\
400 & 0.025       & 0.110     & 0.067 (0.099)  & 0.053 (0.113)   \\
600 & 0       & 0.070     & 0.037 (0.057)  & 0.026 (0.074)   \\
800 & 0       & 0.050     & 0.027 (0.047)  & 0.017 (0.055)   \\
\hline
\end{tabular}
\label{tab:sce1_1}
\end{table}

\begin{table}[t]
\centering
\caption{ In Scenario 1, the averaged coverage and the average and standard deviation (in parentheses) of the width of the confidence intervals from FRBS and HDLR over 200 repetitions are reported. The single change point is located at $\eta = n/2$.}
\begin{tabular}{ccccc}
\hline
\multicolumn{5}{c}{$\mathcal K = 1$ and $p = 200$} \\
 & \multicolumn{2}{c}{$\alpha = 0.01$} & \multicolumn{2}{c}{$\alpha = 0.05$} \\
$n$   & $\mathrm{cover}(1-\alpha)$ & $\mathrm{width}(1-\alpha)$ & $\mathrm{cover}(1-\alpha)$ & $\mathrm{width}(1-\alpha)$ \\ \hline
\multicolumn{5}{c}{FRBS, $c_{\beta} = 0.5$ (small jump size)}                                                                             \\
400 & 0.982            & 109.923 (38.551)  & 0.935    & 72.441 (24.811)  \\
600 & 0.973            & 111.502 (28.749)  & 0.924    & 73.076 (19.474) \\
800 & 0.974            & 117.712 (32.737)  & 0.918    & 77.015 (20.509) \\
\multicolumn{5}{c}{FRBS, $c_{\beta} = 1$ (large jump size)}                                                                             \\
400 & 0.989            & 42.877 (9.479)  & 0.923  & 28.515 (6.469)  \\
600 & 0.974            & 43.505 (7.582)  & 0.969  & 28.299 (4.811)  \\
800 & 0.990            & 43.892 (7.125)  & 0.949    & 28.831 (4.435) \\
\hline
\multicolumn{5}{c}{HDLR, $c_{\beta} = 0.5$ (small jump size)}                                                                             \\
400 & 0.405            & 19.504 (4.493)  & 0.321    & 13.450 (2.944)  \\
600 & 0.424            & 23.690 (4.432)  & 0.367    & 16.000 (2.849) \\
800 & 0.497            & 26.650 (4.294)  & 0.387    & 17.975 (2.685) \\
\multicolumn{5}{c}{HDLR, $c_{\beta} = 1$ (large jump size)}                                                                             \\
400 & 0.748            & 18.153 (3.826)  & 0.681  & 12.712 (2.464)  \\
600 & 0.828            & 20.722 (3.553)  & 0.750  & 14.272 (2.344)  \\
800 & 0.896            & 22.879 (3.551)  & 0.821    & 15.740 (2.284) \\
\hline
\end{tabular}
\label{tab:sce1_2}
\end{table}

An additional simulation study with unequally-spaced two change points is provided in the appendix. These results also show the effectiveness of the proposed method.

\subsection{Real data analysis}
We consider the daily closing price of the S\&P 500 index, from Jan-02-2019 to Jan-19-2023\footnote{The data set is available at \url{https://fred.stlouisfed.org/series/SP500}}.  Inspired by a series of papers \citep[e.g.][]{kokoszka2012functional,kokoszka2013predictability}, which study the predictability of stock prices using the intraday cumulative returns curves, we regress the daily returns ($y_j$) on the intratime cumulative return curves ($X_j$) of the previous one-month (i.e.\ $21$ working days), and use our proposed FRBS as a tool to explore the potential changes in the relationship 
under the model \eqref{eq: model}. Specifically, we transform the closing price data ($P_j$) into the log-ratio of close price between two consecutive days ($y_j$), in percent,
\[
y_j = 100 \cdot \log(P_j/P_{j-1}),
\]
and the discretized $X_j = (X_j(1), \dots, X_j(20))^{\top}$, in percent,
\[X_j(k) = 100 \cdot \log(P_{j-k}/ P_{j-21} ), \quad k = 1, 2, \dots, 20. \]
With $j$ ranging as $j = 22, \dots, 1271$, the sample size is $n = 1250$. \Cref{fig} plots $y_j$ and $X_j$.
\begin{figure}[H]
\centering
\includegraphics[width=8cm]{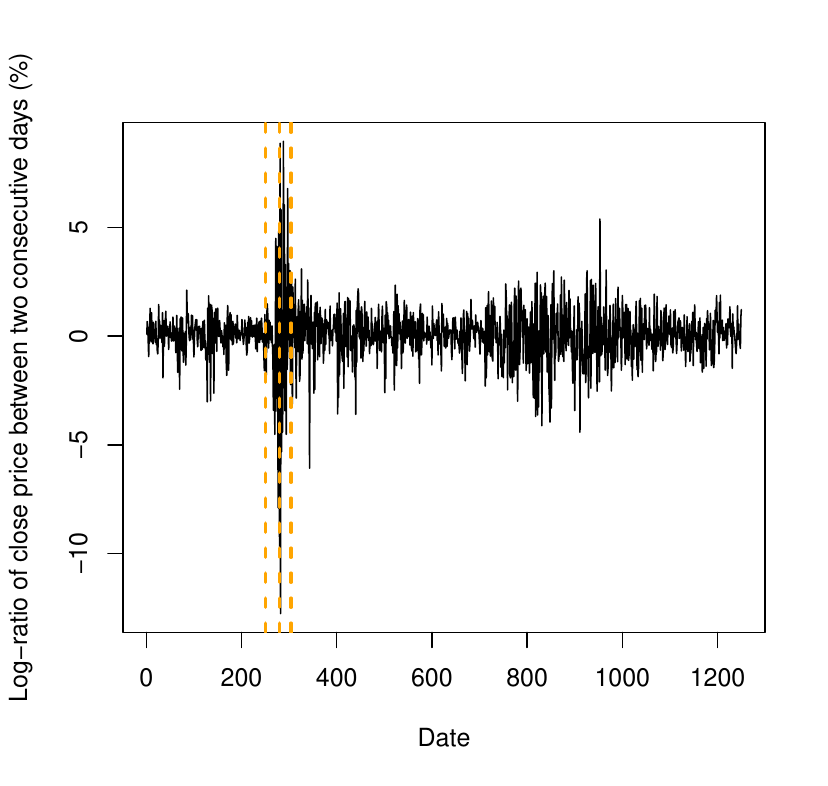}
\includegraphics[width=8cm]{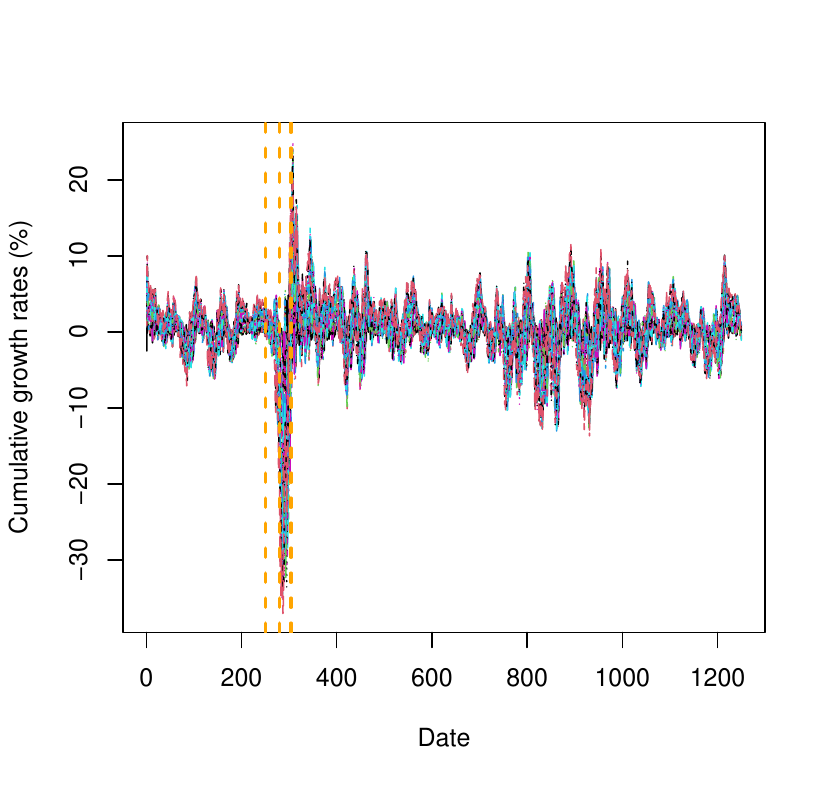}
\caption{The log-ratio of close price between two consecutive days ($y_j$, left); the cumulative growth rate ($X_j(k)$, $1 \le k \le 20$, right). The refined change point estimators are marked by dashed orange lines.}
\label{fig}
\end{figure}

With the tuning parameters selected as discussed in \Cref{sec: NR_Simu}, the proposed FRBS returns three change points at Jan-07-2020, Mar-11-2020, and May-07-2020 as the preliminary estimators and Jan-30-2020, Mar-11-2020, and Apr-16-2020 as the refined ones.  The first estimated change point, with a narrow $95\%$ confidence interval [Jan-28-2020, Feb-03-2020], coincides with the date when WHO officially declared a Public Health Emergency of International Concern. This period reflects investor's concerns about the pandemic's impact on the global economy which lad to increased market volatility and a significant sell-off.  The second estimated change point, with a $95\%$ confidence interval [Feb-20-2020, Mar-30-2020], matches the date when COVID-19 was characterized as a pandemic by WHO. This declaration confirmed the severity and global scale of the outbreak.  During this period, many countries implemented lockdown measures, which lad to huge volatility in financial markets and a sharp drop in the S\&P 500 index.   The third estimated change point reflects that the initial impact of COVID-19 gradually settled.  A series of economic and financial policies were introduced by the governments globally, and the market started to react to these policy changes. Our method produces a wide $95\%$ confidence interval as [Mar-05-2020, May-18-2021].  

In comparison, we consider the same transformed $y_j$ and $X_j$ but regard $X_j$ as a covariate vector of dimensional $20$ and use high-dimensional linear regression with change points \citep{xu2022HDregression} to study the relationship between $y_j$ and $X_j$. The HDLR algorithm outputs two change point estimators at Feb-18-2020 and Apr-14-2020.

\section{Discussion}
In this paper, we study the change point problem within the context of functional linear regression, with minimal assumptions while accommodating temporal dependence and heavy-tailed distributions. Our contribution includes deriving the consistency and the limiting distribution of the change point estimators, a novel advancement in this functional framework. Additionally, we propose a theoretically sound and numerically robust long-run variance estimator to enhance the practicality of our findings. We offer the numerical implementation of our proposed approach which is shown to perform well on synthetic and real datasets.

The theoretical analysis has illuminated several challenging and intriguing directions for future research. One direction could involve devising asymptotically valid confidence intervals in the non-vanishing regime with respect to the size of the change. Another direction could focus on developing methodologies to simultaneously distinguish between different regimes of the size of change, motivated by their difference in the limiting distribution in \Cref{thm: Limiting distribution}.

\newpage
\bibliographystyle{apalike}
\bibliography{references.bib}

\appendix
\section{Additional simulation results}
In this additional simulation, we consider the sample setting described in \Cref{sec: NR_Simu} with unequally-spaced two change points.

\medskip  
\noindent \textbf{Scenario: unequally-spaced two change points.} Let $\mathcal K = 2$ and the unequally-spaced change points $\{\eta_1, \eta_2\} = \{n/4, 5n/8\}$. We vary $n \in \{400, 600, 800\}$ and fix $p = 200$. \Cref{tab:sce2_1} shows the localization performance of both preliminary and final estimators improve as $n$ increases.  Due to the overall poor detection performance HDLR / FRBS when $c_{\beta} = 0.5$, we only report the results from the confidence intervals produced by FRBS for the setting with $c_{\beta} = 1$ in Tables \ref{tab:sce2_2} and \ref{tab:sce2_3}.  The comparison
between Tables \ref{tab:sce2_2} and \ref{tab:sce2_3} reveals that our inference procedure performs better when applied to $\eta_2$ associated with larger spacing with adjacent change points.

\begin{table}[htbp]
\centering
\caption{In Scenario 2, the proportions of under- ($\widehat{\mathcal K} < \mathcal K$),  over-detection ($\widehat{\mathcal K} > \mathcal K$), and  the average and standard deviation (in parentheses) of scaled Hausdorff distance   over 200 repetitions  are reported for FRBS and HDLR. The single change point is located at $\eta = n/2$.  The two change points are located at $\eta_1 = n/4$ and $\eta_2 = 5n/8$.}
\begin{tabular}{ccccc}
\hline
\multicolumn{5}{c}{$\mathcal K = 2$ and $p = 200$} \\ 
$n$   & $\widehat{\mathcal K} < \mathcal K$ & $\widehat{\mathcal K} > \mathcal K$ & $d_{\mathrm{H}}^{\mathrm{pre}}$ & $d_{\mathrm{H}}^{\mathrm{fin}}$ \\ \hline
\multicolumn{5}{c}{FRBS, $c_{\beta} = 0.5$ (small jump size)}                                                                             \\
400 & 0.460    & 0.015        & 0.194 (0.154)  & 0.190 (0.157)   \\
600 & 0.265        & 0.025    & 0.120 (0.134)  & 0.108 (0.139)   \\
800 & 0.150        & 0.015    & 0.081 (0.108)  & 0.068 (0.109)   \\
\multicolumn{5}{c}{FRBS, $c_{\beta} = 1$ (large jump size)}                                                                             \\
400 & 0.010    & 0.025        & 0.029 (0.042)  & 0.024 (0.042)   \\
600 & 0        & 0.015    & 0.020 (0.023)  & 0.011 (0.018)   \\
800 & 0        & 0.015    & 0.015 (0.021)  & 0.011 (0.019)   \\
\hline
\multicolumn{5}{c}{HDLR, $c_{\beta} = 0.5$ (small jump size)}                                                                             \\
400 & 0.930    & 0.015        & 0.275 (0.073)  & 0.300 (0.084)   \\
600 & 0.890        & 0.010    & 0.252 (0.062)  & 0.275 (0.080)   \\
800 & 0.895        & 0.015    & 0.249 (0.055)  & 0.273 (0.072)   \\
\multicolumn{5}{c}{HDLR, $c_{\beta} = 1$ (large jump size)}                                                                             \\
400 & 0.820    & 0.020        & 0.224 (0.045)  & 0.244 (0.058)   \\
600 & 0.900        & 0.015    & 0.229 (0.039)  & 0.244 (0.043)   \\
800 & 0.865        & 0.005    & 0.228 (0.039)  & 0.239 (0.039)   \\
\hline
\end{tabular}
\label{tab:sce2_1}
\end{table}

\begin{table}[htbp]
\centering
\caption{ In Scenario 2, the averaged coverage and the average and standard deviation (in parentheses) of the width of the confidence intervals from FRBS over 200 repetitions for $\eta_1  $ are reported.}
\begin{tabular}{ccccc}
\hline
\multicolumn{5}{c}{FRBS, $p = 200$} \\
 & \multicolumn{2}{c}{$\alpha = 0.01$} & \multicolumn{2}{c}{$\alpha = 0.05$} \\
$n$   & $\mathrm{cover}_1(1-\alpha)$ & $\mathrm{width}_1(1-\alpha)$ & $\mathrm{cover}_1(1-\alpha)$ & $\mathrm{width}_1(1-\alpha)$ \\ \hline
\multicolumn{5}{c}{$c_{\beta} = 1$ (large jump size)}                                                                        \\
400 & 0.959            & 43.591 (12.428) & 0.933 & 28.741 (7.823) \\
600 & 0.994            & 45.492 (12.807)  & 0.980  & 29.751 (8.081)  \\
800 & 0.979            & 43.477 (9.224)  & 0.958  & 28.456 (5.528)  \\
 \hline
\end{tabular}
\label{tab:sce2_2}
\end{table}

\begin{table}[htbp]
\centering
\caption{ In Scenario 2, the averaged coverage and the average and standard deviation (in parentheses) of the width of the confidence intervals from FRBS over 200 repetitions for $\eta_2  $ are reported.}
\begin{tabular}{ccccc}
\hline
\multicolumn{5}{c}{FRBS, $p = 200$} \\
 & \multicolumn{2}{c}{$\alpha = 0.01$} & \multicolumn{2}{c}{$\alpha = 0.05$} \\
$n$   & $\mathrm{cover}_2(1-\alpha)$ & $\mathrm{width}_2(1-\alpha)$ & $\mathrm{cover}_2(1-\alpha)$ & $\mathrm{width}_2(1-\alpha)$ \\ \hline
\multicolumn{5}{c}{$c_{\beta} = 1$ (large jump size)}                                                                                      \\
400 & 0.980            & 45.114 (12.250)   & 0.938    & 29.389 (7.501) \\
600 & 0.990            & 44.142 (10.934)  & 0.980  & 28.909 (6.844)  \\
800 & 0.990            & 44.067 (9.165)  & 0.958  & 28.798 (5.734)  \\ \hline
\end{tabular}
\label{tab:sce2_3}
\end{table}

\section{Proof of \Cref{thm: Consistency}}
\begin{proof}
For $(s_m,e_m] \in \J$ and for all $t\in \sem$ we define 
\begin{equation*}\label{eq: Gset functional cusum}
    \wtilde{G}_t^{s_m,e_m} = \frac{(t-s_m)(e_m-t)}{e_m-s_m}\Sigma\lft[ \beta^*_{(s_m,t]} - \beta^*_{(t,e_m]}, \beta^*_{(s_m,t]} - \beta^*_{(t,e_m]}  \rgt].    
\end{equation*}

For the interval $\sem$, consider the event
\begin{align}\nonumber
    \mathcal{A}(s_m,e_m) =& \, \lft\{ \text{for all } t \in \sem, \rgt.
    \\\nonumber 
    & \lft. \lft| \Whsetm - \Gsetm \rgt| - 0.5\Gsetm \le  \lft(\frac{n}{\Delta}\rgt) (\log^{1+2\xi} (n)) \lft(n^{1/(2r+1)} + 0.5 \rgt) \rgt\},
\end{align}
and define the event
\begin{equation}\label{eq: event A}
   \mathcal{A} = \bigcap_{(s_m,e_m] \in \J } \mathcal{A}(s_m,e_m).
\end{equation}
We established in \Cref{lemma: probability bound first lemma} that
$$
P\big( \mathcal{A} \big) \longrightarrow 1 \qquad \text{ as } n \to \infty.  
$$
All the analysis in the rest of this proof is under this asymptotically almost sure event $\mathcal{A}$.
The strategy here is to use an induction argument. 
Denote
$$
\vartheta_k = C_1 \frac{(n/\Delta) \Delta^{1/(2r+1)} \log^{1+ 2\xi}  n}{\kappa_k^2}.
$$

\
\\
{\bf Step 1:}
We show that, FRBS will consistently reject the existence of change points if they are no undetected change points in $\se$. By induction hypothesis, we have
$$
|\eta_k - s| \le \vartheta_k, \qquad |e - \eta_{k+1}| \le \vartheta_{k+1}.
$$
For each $\sem \in \J$ such that $\sem \subset \se$, there are four possible cases which are outlined below
\begin{enumerate}[label = \roman*)]
    \item $s_m < \eta_k < \eta_{k+1} < e_m$ with $\eta_k - s_m\le \vartheta_k  \text{ and } \eta_{k+1} - e_m \le \vartheta_{k+1}$
    \item $ \eta_{k} \le  s_m< e_m \le \eta_{k+1} $ with $s_m - \eta_{k} \le \vartheta_k$ and $\eta_{k+1}-e_m\le \vartheta_{k+1}$,
    \item $\eta_{k-1}< s_m  \le \eta_k < e_m \le \eta_{k+1} $ with $\eta_k - s_m \le \vartheta_k$,
    \item $\eta_{k-1}\le s_m < \eta_k \le  e_m < \eta_{k+1} $ with $ e_m -\eta_k \le \vartheta_k$.   
\end{enumerate}
We shall consider the first case, all other cases are simpler and could be handled similarly. There are two previously detected change point $\eta_k$ and $\eta_{k+1}$ in $\sem$ and we are going to show that FRBS shall not detect any change point in $\sem$. 
On the event $\mathcal{A}$ we write that 
\begin{align*}
    \forall t \in \semp \qquad \Whsetm &\le \frac{3}{2}\Gsetm +  \lft(\frac{n}{\Delta}\rgt) \log^{1+2\xi} (n) \lft(n^{1/(2r+1)} + \frac{1}{2} \rgt)
    \\
    &\le 3\kappa_k^2 (\eta_k -s_m) + 3\kappa_{k+1}^2(e-\eta_{k+1}) + 2\lft(\frac{n}{\Delta}\rgt) n^{1/(2r+1)} \log^{1+2\xi} (n)   
    \\
    &\le \lft(8C_1 + 2 \rgt) \lft(\frac{n}{\Delta}\rgt) n^{1/(2r+1)} \log^{1+2\xi} (n) < \tau
\end{align*}
where the second last line follows from \Cref{lemma: pop cusum two point projection} and the last line just follows from the definition of $\tau$.
\
\\
{\bf Step 2:}
We show that FRBS will correctly detect the existence of an undetected change point in $\se$. In this case, there exists some change point, $\eta_k$ in $\se$, such that
$$
\min\{\eta_k -s, e - \eta_k\} > \Delta - \vartheta_k,
$$
for some $1\le k \le \Ktilde$. Realize that $\Delta - \vartheta_k > 4\Delta/5$, asymptotically. For this step, it is sufficient to show that the set $\mathcal{M}^{s,e}$, form \Cref{alg: FBS}, is not empty. From the construction of intervals in $\J$ and from \Cref{lemma: interval in J}, we can always find an interval $\sem\in \J$ such that $\sem \subset \se$ containing $\eta_k$ such that 
\begin{equation}\label{eq: thm a2}
    e_m - s_m \le \Delta , \qquad\text{and}\qquad \min\{\eta_k - s_m, e_m - \eta_k\} \ge \Delta/5.    
\end{equation}
\
\\
On the event $\mathcal{A}$, we have
\begin{align}\label{eq: fth1}
    \max_{s_m<t\le e_m} \Whsetm  \ge  \widehat{W}^{s_m, e_m}_{\eta_k}   \ge \frac{1}{2} \widetilde{G}^{s_m, e_m}_{\eta_k} - \lft(\frac{n}{\Delta}\rgt) \log^{1+2\xi} (n) \lft(n^{1/(2r+1)} + \frac{1}{2} \rgt).    
\end{align}
Since $\eta_k$ is the only change point in $\sem$, using \eqref{eq: thm a2}, we write that
\begin{equation}\label{eq: fth111}
     \widetilde{G}^{s_m,e_m}_{\eta_k} = \kappa_k^2 \frac{(\eta_k -s_m)(e_m- \eta_k)}{(e_m-s_m)} \ge \frac{1}{2}\kappa_k^2 \min\{\eta_k - s_m, e_m - \eta_k\} \ge \frac{1}{10} \kappa_k^2 \Delta.
\end{equation}
We may extend \eqref{eq: fth1} to have
\begin{align*}
    \max_{s_m<t\le e_m}  \Whsetm  &\ge \frac{1}{2} \widetilde{G}^{s_m, e_m}_{\eta_k} - \lft(\frac{n}{\Delta}\rgt) \log^{1+2\xi} (n) \lft(n^{1/(2r+1)} + \frac{1}{2} \rgt)
    \\
    &\ge \frac{1}{20} \kappa_k^2 \Delta - \lft(\frac{n}{\Delta}\rgt) \lft( \log^{1+2\xi} (n) \rgt) \lft(n^{1/(2r+1)} + \frac{1}{2} \rgt)
    \\
    &\ge \frac{1}{20} \kappa_k^2 \Delta - o\lft( \kappa_k^2 \Delta \rgt) > \tau.
\end{align*}
where the second last line follows from \eqref{eq: fth111} and the last line follows from \Cref{assume: SNR Signal to Noise ratio}. Therefore $\mathcal{M}^{s,e} \neq \emptyset$.
\
\\
{\bf Step 3:}
 This is the localization step. We have $\mathcal{M}^{s,e} \neq \emptyset$. Let $b = b_{m^*}$ be the chosen point in \Cref{alg: FBS}. Let $(s_{m^*},e_{m^*}]$ be the corresponding interval. Since it is the narrowest one, we have $(e_{m^*} - s_{m^*}) \le (e_m-s_m) \le \Delta$, where $ \sem$ is the interval picked at \eqref{eq: thm a2}. Therefore, $(s_{m^*},e_{m^*}]$ can contains exactly one change point $\eta_k$.
\
\\
Without loss of generality, let's assume that $b > \eta_k$. Additionally, we shall assume that $(b - \eta_k) > \frac{3}{\kappa_k^2}$; if not, the localization rate follows directly.
Since 
$$
 \widehat{W}^{s_{m^*},e_{m^*}}_b \ge  \widehat{W}^{s_{m^*},e_{m^*}}_{\eta_k}.
$$
We write that
\begin{align*}
    &\sum_{j=s_{m^*}+1}^b \lft(Y_j - \lft\lb X_j, \what{\beta}_{(s_{m^*},b]} \rgt\rb_\lt \rgt)^2 + \sum_{j=b+1}^{e_{m^*}} \lft(Y_j - \lft\lb X_j, \what{\beta}_{(b,e_{m^*}]} \rgt\rb_\lt \rgt)^2
    \\
     \le &\sum_{j=s_{m^*}+1}^{\eta_k} \lft(Y_j - \lft\lb X_j, \what{\beta}_{(s_{m^*},{\eta_k}]} \rgt\rb_\lt \rgt)^2 + \sum_{j={\eta_k}+1}^{e_{m^*}} \lft(Y_j - \lft\lb X_j, \what{\beta}_{({\eta_k},e_{m^*}]} \rgt\rb_\lt \rgt)^2,
\end{align*}
which is equivalent to 
\begin{align}\label{eq: fth 11a}
         & \lft(\lft(\frac{b-\eta_k}{b - s_{m^*}}\rgt) + 1  \rgt)^2 \sum_{j=\eta_k+1}^{b} \lft\lb X_j, \beta^*_{\eta_{k+1}}  -\beta^*_{{\eta_k}} \rgt\rb_\lt^2 
         \\\label{eq: fth 11b}
         \le& \lft( \sum_{j=s_{m^*}+1}^{\eta_k} \lft(Y_j - \lft\lb X_j, \what{\beta}_{(s_{m^*},{\eta_k}]} \rgt\rb_\lt \rgt)^2 - \sum_{j=s_{m^*}+1}^{\eta_k} \lft(Y_j - \lft\lb X_j, \beta^*_{\eta_{k}} \rgt\rb_\lt \rgt)^2  \rgt)
        \\\label{eq: fth 11c}
        & + \lft(\sum_{j={\eta_k}+1}^{e_{m^*}} \lft(Y_j - \lft\lb X_j, \what{\beta}_{({\eta_k},e_{m^*}]} \rgt\rb_\lt \rgt)^2 - \sum_{j={\eta_k}+1}^{e_{m^*}} \lft(Y_j - \lft\lb X_j, \beta^*_{\eta_{k+1}} \rgt\rb_\lt \rgt)^2 \rgt)
        \\\label{eq: fth 11d}
        & + \lft(\sum_{j=s_{m^*}+1}^b \lft(Y_j - \lft\lb X_j, \beta^*_{(s_{m^*},b]} \rgt\rb_\lt \rgt)^2 -\sum_{j=s_{m^*}+1}^b \lft(Y_j - \lft\lb X_j, \what{\beta}_{(s_{m^*},b]} \rgt\rb_\lt \rgt)^2 \rgt)
        \\\label{eq: fth 11e}
        & + \lft( \sum_{j=b+1}^{e_{m^*}} \lft(Y_j - \lft\lb X_j, \beta^*_{\eta_{k+1}}\rgt\rb_\lt \rgt)^2 - \sum_{j=b+1}^{e_{m^*}} \lft(Y_j - \lft\lb X_j, \what{\beta}_{(b,e_{m^*}]} \rgt\rb_\lt \rgt)^2 \rgt)
        \\\label{eq: fth 11f}
        &+ 2 \lft(\frac{b - \eta_k}{b - s_{m^*}}\rgt)\lft( \sum_{j= s_{m^*} + 1}^{b} \lft\lb X_j, \beta^*_{\eta_{k}} - \beta^*_{(s_{m^*},b]} \rgt\rb_\lt \eps_j  \rgt) 
        \\\label{eq: fth 11g}
        &+2\lft( \sum_{j= \eta_k +1}^{b} \lft\lb X_j, \beta^*_{\eta_{k+1}} - \beta^*_{\eta_{k}} \rgt\rb_\lt \eps_j  \rgt).
\end{align}
Therefore we have,
$$
\eqref{eq: fth 11a}  \le \big| \eqref{eq: fth 11b} \big| + \big|\eqref{eq: fth 11c} \big| + \big|\eqref{eq: fth 11d} \big| + \big|\eqref{eq: fth 11e} \big| + \big|\eqref{eq: fth 11f} \big| + \big|\eqref{eq: fth 11g} \big| .
$$

\
\\
{\bf Step 3A: the order of magnitude of \eqref{eq: fth 11b}, \eqref{eq: fth 11c}, \eqref{eq: fth 11d} and \eqref{eq: fth 11e}.}
Following from the \Cref{lemma: new localization}, we have
\begin{align*}
   \big| \eqref{eq: fth 11b} \big| &= O_p\lft(\lft(n/\Delta\rgt)  (\eta_p - s_{m^*}) \delta_{\eta_p - s_{m^*}} \log^{1+\xi}(\eta_p - s_{m^*}) \rgt),
   \\
    \big|\eqref{eq: fth 11c} \big| &= O_p\lft(  \lft(n/\Delta\rgt) (e_{m^*} - \eta_p ) \delta_{e_{m^*} - \eta_p} \log^{1+\xi}(e_{m^*} - \eta_p) \rgt),
    \\
    \big|\eqref{eq: fth 11d} \big| &= O_p\lft(  \lft(n/\Delta\rgt) (b - s_{m^*}) \delta_{b - s_{m^*}} \log^{1+\xi}(b - s_{m^*}) \rgt),
    \\
    \big|\eqref{eq: fth 11e} \big| &= O_p\lft(  \lft(n/\Delta\rgt)  (e_{m^*} - b) \delta_{e_{m^*} - b} \log^{1+\xi}(e_{m^*} - b) \rgt),
\end{align*}
which lead us to
\begin{equation*}
    \big| \eqref{eq: fth 11b} \big| + \big|\eqref{eq: fth 11c} \big| + \big|\eqref{eq: fth 11d} \big| + \big|\eqref{eq: fth 11e} \big| = O_p \lft( \lft(n/\Delta\rgt)\Delta^{1/(2r+1)} \log^{1+\xi} (\Delta) \rgt).
\end{equation*}

\
\\
{\bf Step 3B: the order of magnitude of \eqref{eq: fth 11f} and \eqref{eq: fth 11g}.}
Observe that from \Cref{lemma: general markov type bound II} we may have
$$
\E\lft[ \frac{1}{\kappa_k^2} \lft|\sum_{j=\eta_k +1}^{ t'} \lft\lb X_j, \beta^*_{\eta_{k+1}} - \beta^*_{\eta_k} \rgt\rb_\lt\eps_j\rgt|^2\rgt] = O(t' - \eta_k),
$$
and using \Cref{lemma: max to Op by peeling}, we may write
\begin{equation}\label{eq: III term in thm2}
   \max_{1\le k \le \Ktilde} \max_{1/\kappa_k^2< t' < \eta_{k+1}} \lft|\frac{1}{\sqrt{(t'-\eta_k)}(\log^{1+\xi} \lft( (t' - \eta_k) \kappa_k^2 \rgt) + 1) } \frac{1}{\kappa_k} \sum_{j=\eta_k +1}^{ t'} \lft\lb X_j, \beta^*_{\eta_{k+1}} - \beta^*_{\eta_k} \rgt\rb_\lt\eps_j\rgt|^2 = O_p\lft(\Ktilde\rgt).  
\end{equation}
Following this, we have
$$
 \lft|\lft( \sum_{j= \eta_k +1}^{b} \lft\lb X_j, \beta^*_{\eta_{k+1}} - \beta^*_{\eta_{k}} \rgt\rb_\lt \eps_j  \rgt) \rgt| = O_p\lft( \sqrt{\Ktilde} \sqrt{(b- \eta_p)}\kappa_k \lft\{ \log^{1+\xi} \lft((b-\eta_k) \kappa_k^2\rgt) + 1 \rgt\} \rgt),
$$
which is a bound on \eqref{eq: fth 11g}.
Similarly using \eqref{eq: III term in thm2}, we get 
\begin{align*}
    &\lft|  \lft(\frac{b - \eta_k}{b - s_{m^*}}\rgt)\lft( \sum_{j= s_{m^*} + 1}^{b} \lft\lb X_j, \beta^*_{\eta_{p}} - \beta^*_{(s_{m^*},b]} \rgt\rb_\lt \eps_j  \rgt) \rgt| 
    \\
    =& O_p\lft( \sqrt{\Ktilde} \lft[\frac{b - \eta_k}{b - s_{m^*}}\rgt] \sqrt{(b- s_{m^*})}\kappa_k \lft\{ \log^{1+\xi}\lft( (b- s_{m^*}) \kappa_k^2\rgt) + 1 \rgt\} \rgt)
    \\
    =& O_p\lft( \sqrt{\Ktilde}  \sqrt{(b- \eta_k)}\kappa_k \log^{1+\xi} \lft((b- s_{m^*}) \kappa_k^2\rgt) \rgt),
\end{align*}
where we use $\frac{b - \eta_p}{b - s_{m^*}} \le 1$ and $\log \lft((b- s_{m^*}) \kappa_k^2\rgt) > 1$ in the last line. This bound \eqref{eq: fth 11f}.
Therefore
$$
\big|\eqref{eq: fth 11f} \big| + \big|\eqref{eq: fth 11g} \big| = O_p\lft( \sqrt{\Ktilde} \sqrt{(b- \eta_p)}\kappa_k \lft\{ \log^{1+\xi} ((b-\eta_p) \kappa_k^2) \rgt\} \rgt).
$$
\
\\
{\bf Step 3C: the lower bound of \eqref{eq: fth 11a}:}
Observe that from \Cref{lemma: general markov type bound I} we may have
$$
\E\lft[ \frac{1}{\kappa_k^2} \lft|\sum_{j=\eta_k +1}^{ t'} \lft(\lft\lb X_j, \beta^*_{\eta_{k+1}} - \beta^*_{\eta_k} \rgt\rb_\lt^2  - \kappa_k^2 \rgt)\rgt|^2\rgt] = O(t' - \eta_k),
$$
and using \Cref{lemma: max to Op by peeling}, we may write
\begin{equation}\label{eq: IV term}
     \max_{1\le k \le \Ktilde} \max_{\frac{1}{\kappa_k^2} +\eta_k< t' < \eta_{k+1} } \lft| \frac{1}{\sqrt{(t'-\eta_k)}(\log^{1+\xi} ((t'- \eta_k)\kappa_k^2) + 1) } \frac{1}{\kappa_k} \sum_{j=\eta_k +1}^{t'} \lft(\lft\lb X_j, \beta^*_{\eta_{k+1}} - \beta^*_{\eta_k} \rgt\rb_\lt^2  - \kappa_k^2 \rgt)\rgt|^2 = O_p\lft(\Ktilde\rgt).    
\end{equation}
Following \eqref{eq: IV term}, we may write
\begin{align}\nonumber
    \eqref{eq: fth 11a} &=  \lft( \frac{b - \eta_k}{b - s_{m^*}} + 1 \rgt)^2 \sum_{j=\eta_k+1}^{b} \lft\lb X_j, \beta^*_{\eta_{k+1}}  -\beta^*_{{\eta_k}} \rgt\rb_\lt^2 
    \\\nonumber
    &=  \lft( \frac{b - \eta_k}{b - s_{m^*}} + 1 \rgt)^2  \lft[ (b - \eta_k) \kappa_k^2 - O_p\lft( \sqrt{\Ktilde} \sqrt{(b - \eta_k)}\kappa_k \lft\{ \log^{1+\xi} ((b - \eta_k)\kappa_k^2) + 1 \rgt\} \rgt) \rgt]
    \\\nonumber
    &\ge  (b - \eta_k)\kappa^2_k - O_p\lft( \sqrt{\Ktilde} \sqrt{ (b - \eta_k) \kappa_k^2} \lft( \log^{1+\xi} ((b - \eta_k)\kappa_k^2) \rgt) \rgt),
\end{align}
where we use $ \frac{b - \eta_k}{b - s_{m^*}} +1 \ge 1$ and $\log \lft((b- s_{m^*}) \kappa_k^2\rgt) > 1$ in the last line.

\
\\
Following from step 3A, step 3B and step 3C, we get
\begin{align*}
    &(b - \eta_k)\kappa^2_p - O_p\lft( \sqrt{\Ktilde} \sqrt{(b- \eta_k)}\kappa_k \lft\{ \log^{1+\xi} (b-\eta_k) \kappa_k^2) \rgt\} \rgt) 
    \\
    \le &O_p \lft( (n/\Delta) \Delta^{1/(2r+1)} \log^{1+\xi} (\Delta) \rgt) + O_p\lft( \sqrt{\Ktilde} \sqrt{(b- \eta_k)}\kappa_k \lft\{ \log^{1+\xi} (b-\eta_k) \kappa_k^2) \rgt\} \rgt),    
\end{align*}
with $\Ktilde\le n/\Delta$, it implies

\begin{equation}\label{eq: consistency localization}
    (b - \eta_k) \kappa_p^2 = O_p \lft( (n/\Delta) \Delta^{1/(2r+1)} \log^{1+\xi} (\Delta) \rgt).
\end{equation}
This concludes the induction step when $\se$ contains an undetected change point.
\end{proof}

\subsection{Technical results for the proof of \Cref{thm: Consistency}}
\begin{lemma}\label{lemma: interval in J}
Let $\se \subset (0,n]$ be given. Let $\eta_k$ be a point in $\se$. Suppose $\min\{\eta_k -s, e - \eta_k\} > 4\Delta/5$. Then there exists an interval $(s_m,e_m] \in \J \cap \se$ containing $\eta_k$ such that
$$
e_m - s_m \le \Delta , \qquad\text{and}\qquad \min\{\eta_k - s_m, e_m- \eta_k\} \ge \Delta/5. 
$$
\end{lemma}

\begin{proof}
    There are at most two intervals in each layer $\J_k$, for $1\le k \le M$, that contains any given point. We shall consider the layer with $\mathfrak{l}_k = \Delta$ and $\mathfrak{b}_k = \Delta/2$. Without loss of generality, let $\lft(\frac{(i-1)\Delta}{2}, \frac{(i-1)\Delta}{2} + \Delta \rgt]$ and $\lft(\frac{i\Delta}{2}, \frac{i\Delta}{2} + \Delta \rgt]$ are intervals containing $\eta_k$.

    \
    \\
    {\bf Case I:} Suppose $\eta_k - i\Delta/2 > (i+1)\Delta/2 -\eta_k$. Observe that $\eta_k - i\Delta/2 \ge \Delta/4$. The interval $(s_m,e_m] =\lft(\frac{i\Delta}{2}, \frac{i\Delta}{2} + \Delta \rgt]$ satisfies the required property because $\eta_k - s_m = \eta_k - i\Delta/2 \ge \Delta/4$ and $e_m - \eta_k > i\Delta/2 + \Delta - \lft( (i-2)\Delta/2 + \Delta \rgt) = \Delta/2 $.

    \
    \\
    {\bf Case II:} Suppose $\eta_k - i\Delta/2 \le (i+1)\Delta/2 -\eta_k $. Using arguments akin to the previous case, the interval $\left(\frac{(i-1)\Delta}{2}, \frac{(i-1)\Delta}{2} + \Delta \right]$ emerges as the necessary interval.
\end{proof}

\subsubsection{Large probability event}
Recall for any $a>0$, $\delta_a \asymp a^{-2r/(2r+1)}$.
\begin{lemma}\label{lemma: probability bound first lemma}
    Let $\xi > 0$. Then, as $n \to \infty$, we have
\begin{align*}
        P \bigg( \forall (s_m, e_m]\in \J&, \quad \forall t \in (s_m, e_m], \quad 
        \\
        &\lft| \Whsetm - \Gsetm \rgt| - 0.5\Gsetm \le   \lft(\frac{n}{\Delta}\rgt) \log^{1+2\xi} (n) \lft(n^{1/(2r+1)} + 0.5 \rgt)  \bigg) \to  1.   
\end{align*}
\end{lemma}

\begin{proof}
    Let $\sem \in \J$ be fixed. For notational simplicity, denote $s = s_m$ and $e=e_m$.
    Denote 
    \begin{equation}\nonumber
    W^{*s,e}_t = \sum_{j=s+1}^e \lft( Y_j - \lb X_j, \beta^*_\se \rb_\lt \rgt)^2 - \sum_{j=s+1}^t \lft( Y_j - \lb X_j, \beta^*_\st \rb_\lt \rgt)^2 - \sum_{j=t+1}^e \lft( Y_j - \lb X_j, \beta^*_\te \rb_\lt \rgt)^2.
    \end{equation}
     We show in Step~$1$ that
    \begin{equation}\label{eq: step1 probablity bound}
        \max_{s < t \le e} \lft| \Whset - W^{*s,e}_t\rgt| = O_p\lft( (e-s)^{1/(2r+1)} \log^{1+\xi} (e-s) \rgt).
    \end{equation}
    In Step~$2$, we show that
    \begin{equation}\label{eq: step2 probablity bound}
         \max_{s<t\le e} \frac{1}{ \sqrt{ \Gset \log^{1+\xi} (t-s) }} \lft| W^{*s,e}_t - \Gset \rgt| = O_p\lft( 1 \rgt),
    \end{equation}
    when $\Gset \neq 0$. It follows from using $4ab \le (a+b)^2$ at \eqref{eq: step2 probablity bound} that
    \begin{equation}\nonumber
        \max_{s<t\le e} \frac{1}{ 0.5\lft(\Gset + \log^{1+\xi} (t-s) \rgt)} \lft| W^{*s,e}_t - \Gset \rgt| = O_p\lft( 1 \rgt).
    \end{equation}
    Therefore,
    \begin{equation}\label{eq: ntest1}
        P \bigg(\forall t \in \se,  \lft| \Whset - \Gset\rgt| - 0.5\Gset \le   \lft(n^{1/(2r+1)} + 0.5 \rgt) \log^{1+2\xi} (n)\bigg) \to 1 ,
    \end{equation}
    as $n \to \infty$. The factor $\log^\xi(n)$ is to make the event asymptotically almost surely. When $\se$ has no change point, we have $W^{*s,e}_t = \Gset = 0$ and \eqref{eq: ntest1} trivially holds. Following the cardinality of $\J$ at \eqref{rm: card J}, the main result now follows from the union bound.
    
    \
    \\
    {\bf Step 1:}
    Using $(a-b)^2 - (a-c)^2 = (b-c)^2 - 2(a-c)(b-c)$, we may write
    \begin{align}\nonumber
         &\Whset - W^{*s,e}_t
         \\\nonumber
         =& \underbrace{\sum_{j=s+1}^e \lft\lb X_j, \bhase - \beta^*_\se \rgt\rb_\lt^2}_{\B_1} - \underbrace{\sum_{j=s+1}^t \lft\lb X_j, \bhast - \beta^*_\st \rgt\rb_\lt^2}_{\B_2} - \underbrace{\sum_{j=t+1}^e \lft\lb X_j, \bhate - \beta^*_\te \rgt\rb_\lt^2}_{\B_3}
         \\\nonumber
         & \quad + \underbrace{2\sum_{j=s+1}^e \lft\lb X_j, \beta^*_\se -\bhase  \rgt\rb_\lt \eps_j}_{\B_4} - \underbrace{2\sum_{j=s+1}^t \lft\lb X_j, \beta^*_\st -\bhast \rgt\rb_\lt\eps_j}_{\B_5} - \underbrace{2\sum_{j=t+1}^e \lft\lb X_j, \beta^*_\te -\bhate\rgt\rb_\lt \eps_j}_{\B_6}
         \\\nonumber
         & \quad + \underbrace{2\sum_{j=s+1}^e \lft\lb X_j, \beta^*_\se -\bhase  \rgt\rb_\lt \lft\lb X_j, \beta_j^* - \beta^*_\se \rgt\rb_\lt}_{\B_7} - \underbrace{2\sum_{j=s+1}^t \lft\lb X_j, \beta^*_\st -\bhast \rgt\rb_\lt \lft\lb X_j, \beta_j^* - \beta^*_\st \rgt\rb_\lt}_{\B_9}
         \\\label{eq: new3a}
         &\quad - \underbrace{2\sum_{j=t+1}^e \lft\lb X_j, \beta^*_\te -\bhate\rgt\rb_\lt\lft\lb X_j, \beta_j^* - \beta^*_\te \rgt\rb_\lt}_{\B_9}.
     \end{align}
    We will show the technique to bound $\B_1, \B_2, \B_3$ and the result for $\B_4, \B_5, \B_6$ and $\B_7, \B_8, \B_9$ follows from the same outlined idea and the corresponding \Cref{lemma: consistency flr} and \Cref{lemma: unfold empirical bound} respectively. 

\
\\
Observe that for $|\B_2|$
\begin{align*}
    &\max_{s<t\le e} \sum_{j=s+1}^t \lft\lb X_j, \bhast - \beta^*_\st \rgt\rb_\lt^2 
    \\
    =& \max_{s<t\le e} (t-s) \shat_\st\lft[\bhast - \beta^*_\st,\bhast - \beta^*_\st \rgt]
    \\
    \le& \max_{s<t\le e} \lft( (t-s)^{1/(2r+1)} \log^{1+\xi} (t-s) \rgt) \max_{s<t\le e} \lft( \frac{\delta_{t-s}^{-1}}{\log^{1+\xi} (t-s)} \shat_\st\lft[\bhast - \beta^*_\st,\bhast - \beta^*_\st \rgt] \rgt)
    \\
    =& (s-e)^{1/(2r+1)} \log^{1+\xi} (e-s) O_p(1),
\end{align*}
where the last line follows from the fact that $z \mapsto z^a \log z $ is strictly increasing for any $a \ge 0$ and the \Cref{lemma: empirical excess risk}. 

\
\\
For $|\B_1|$, at $t=e$, we have
\begin{align*}
    \sum_{j=s+1}^e \lft\lb X_j, \bhase - \beta^*_\se \rgt\rb_\lt^2 = (e-s)^{1/(2r+1)} \log^{1+\xi} (e-s) O_p(1).
\end{align*}

The bound for the term $ |\B_3|$ follows by same arguement as $\B_1$.
This establish \eqref{eq: step1 probablity bound}.

    \
    \\
    {\bf Step 2:}
    Let $\Gset \neq 0$. Note that
    $$
    \beta^*_\st - \beta^*_\se = \lft(\frac{e-t}{e-s}\rgt) \lft( \beta^*_\st - \beta^*_\te \rgt)
    $$
    and 
    $$
    \beta^*_\te - \beta^*_\se = \lft(\frac{t-s}{e-s}\rgt) \lft( \beta^*_\te - \beta^*_\st \rgt).
    $$

    \
    \\
    Using $(a-b)^2 - (a-c)^2 = (b-c)^2 - 2(a-c)(b-c)$, we may write
    \begin{align}\nonumber
        W^{*s,e}_t =& \sum_{j=s+1}^t \lft\lb X_j, \beta^*_\st - \beta^*_\se \rgt\rb_\lt^2 + \sum_{j=t+1}^e \lft\lb X_j, \beta^*_\te - \beta^*_\se \rgt\rb_\lt^2
         \\\nonumber
         & \quad + 2\sum_{j=s+1}^t \lft\lb X_j,\beta^*_\se - \beta^*_\st \rgt\rb_\lt\eps_j + 2\sum_{j=t+1}^e \lft\lb X_j, \beta^*_\se - \beta^*_\te \rgt\rb_\lt \eps_j
         \\\nonumber
         & \quad + 2\sum_{j=s+1}^t \lft\lb X_j, \beta^*_\se -\beta^*_\st \rgt\rb_\lt \lft\lb X_j, \beta_j^* - \beta^*_\st \rgt\rb_\lt + 2\sum_{j=t+1}^e \lft\lb X_j, \beta^*_\se -\beta^*_\te \rgt\rb_\lt\lft\lb X_j, \beta_j^* - \beta^*_\te \rgt\rb_\lt
         \\\nonumber
         =& \lft(\frac{e-t}{e-s}\rgt)^2 \sum_{j=s+1}^t \lft\lb X_j, \beta^*_\st - \beta^*_\te \rgt\rb_\lt^2 + \lft(\frac{t-s}{e-s}\rgt)^2 \sum_{j=t+1}^e \lft\lb X_j, \beta^*_\te - \beta^*_\st \rgt\rb_\lt^2
         \\\nonumber
         & \quad + 2 \lft(\frac{e-t}{e-s}\rgt) \sum_{j=s+1}^t \lft\lb X_j,\beta^*_\te - \beta^*_\st \rgt\rb_\lt\eps_j + 2 \lft(\frac{t-s}{e-s}\rgt) \sum_{j=t+1}^e \lft\lb X_j, \beta^*_\st - \beta^*_\te \rgt\rb_\lt \eps_j
         \\\nonumber
         & \quad + 2 \lft(\frac{e-t}{e-s}\rgt) \sum_{j=s+1}^t \lft\lb X_j, \beta^*_\te -\beta^*_\st \rgt\rb_\lt \lft\lb X_j, \beta_j^* - \beta^*_\st \rgt\rb_\lt
         \\\nonumber
         & \quad + 2 \lft(\frac{t-s}{e-s}\rgt) \sum_{j=t+1}^e \lft\lb X_j, \beta^*_\st -\beta^*_\te \rgt\rb_\lt\lft\lb X_j, \beta_j^* - \beta^*_\te \rgt\rb_\lt.
    \end{align}
Observe that 
$$\Gset = \lft( \frac{(t-s)^2(e-t)}{(e-s)^2} + \frac{(t-s)(e-t)^2}{(e-s)^2}   \rgt)\Sigma\lft[ \beta^*_\st - \beta^*_\te, \beta^*_\st - \beta^*_\te  \rgt].
$$
Also, $\lft(\frac{e-t}{e-s}\rgt)^2 \le 1$, $\lft(\frac{t-s}{e-s}\rgt)^2 \le 1$, $\sum_{j=s+1}^t \Sigma[\beta^*_\te -\beta^*_\st, \beta_j^* - \beta^*_\st] = 0$ and $\sum_{j=t+1}^e \Sigma[\beta^*_\st -\beta^*_\te, \beta_j^* - \beta^*_\te] = 0 $. Using the triangle inequality, we may write

\begin{align}\nonumber
    &\lft| W^{*s,e}_t - \Gset \rgt| 
    \\\label{eq: f41}
    \le &  \lft(\frac{e-t}{e-s}\rgt)\lft| \sum_{j=s+1}^t \lft( \lft\lb X_j, \beta^*_\st - \beta^*_\te \rgt\rb_\lt^2 - \Sigma[\beta^*_\st - \beta^*_\te, \beta^*_\st - \beta^*_\te] \rgt) \rgt| 
    \\\label{eq: f42}
    +& \lft(\frac{t-s}{e-s}\rgt)\lft| \sum_{j=t+1}^e \lft( \lft\lb X_j, \beta^*_\te - \beta^*_\st \rgt\rb_\lt^2 - \Sigma[\beta^*_\st - \beta^*_\te, \beta^*_\st - \beta^*_\te] \rgt) \rgt|
         \\\label{eq: f43}
        + &   2 \lft(\frac{e-t}{e-s}\rgt)\lft| \sum_{j=s+1}^t \lft\lb X_j,\beta^*_\te - \beta^*_\st \rgt\rb_\lt\eps_j \rgt| + 2 \lft(\frac{t-s}{e-s}\rgt) \lft| \sum_{j=t+1}^e \lft\lb X_j, \beta^*_\st - \beta^*_\te \rgt\rb_\lt \eps_j \rgt|
         \\\label{eq: f44}
         +&  2 \lft(\frac{e-t}{e-s}\rgt)\lft| \sum_{j=s+1}^t \lft\lb X_j, \beta^*_\te -\beta^*_\st \rgt\rb_\lt \lft\lb X_j, \beta_j^* - \beta^*_\st \rgt\rb_\lt -  \Sigma[\beta^*_\te -\beta^*_\st, \beta_j^* - \beta^*_\st] \rgt|
         \\\label{eq: f45}
         +&   2 \lft(\frac{t-s}{e-s}\rgt)\lft| \sum_{j=t+1}^e \lft\lb X_j, \beta^*_\st -\beta^*_\te \rgt\rb_\lt\lft\lb X_j, \beta_j^* - \beta^*_\te \rgt\rb_\lt - \Sigma[\beta^*_\st -\beta^*_\te, \beta_j^* - \beta^*_\te] \rgt|.
\end{align}
Our approach involves bounding each of the six terms through four distinct sub-steps. In Step~$2$A, we establish the bound for equations \eqref{eq: f41} and \eqref{eq: f42}. Progressing to Step~$2$B, we derive the bound for equation \eqref{eq: f43}. Moving on to Step~$2$C, we obtain the bound for equations \eqref{eq: f44} and \eqref{eq: f45}. Notably, all these derived bounds are uniform across $t \in \se$. The final step, Step~$2$D, amalgamates these outcomes into a coherent result.

\
\\
{\bf Step 2A.}
Using \Cref{lemma: general markov type bound I} we have
\begin{align*}
     &\E\lft[\lft| \sum_{j=s+1}^t \lft( \lft\lb X_j, \beta^*_\st - \beta^*_\te \rgt\rb_\lt^2 - \Sigma[\beta^*_\st - \beta^*_\te, \beta^*_\st - \beta^*_\te] \rgt) \rgt|^2\rgt] = O(t-s) \Sigma[\beta^*_\st - \beta^*_\te, \beta^*_\st - \beta^*_\te]. 
\end{align*}
Writing $\Sigma[\beta^*_\st - \beta^*_\te, \beta^*_\st - \beta^*_\te] = \frac{(e-s)}{(t-s)(e-t)}\Gset$, we may also write it as
$$
\E\lft[  \frac{(e-t)(t-s)}{(e-s)} \frac{1}{\Gset}\lft|  \sum_{j=s+1}^t \lft( \lft\lb X_j, \beta^*_\st - \beta^*_\te \rgt\rb_\lt^2 - \Sigma[\beta^*_\st - \beta^*_\te, \beta^*_\st - \beta^*_\te] \rgt) \rgt|^2\rgt] = O(t-s)
$$
Using the \Cref{lemma: maximal inequality useful}, we may write
\begin{align}\nonumber
     &\E\lft[ \max_{s<t\le e} \frac{(e-t)}{(e-s)} \frac{1}{ \log^{1+\xi} (t-s) } \frac{1}{\Gset}\lft|  \sum_{j=s+1}^t \lft( \lft\lb X_j, \beta^*_\st - \beta^*_\te \rgt\rb_\lt^2 - \Sigma[\beta^*_\st - \beta^*_\te, \beta^*_\st - \beta^*_\te] \rgt) \rgt|^2\rgt] = O(1)
     \\\nonumber
    &\E\lft[ \max_{s<t\le e} \frac{(t-s)}{(e-s)}\frac{1}{ \log^{1+\xi} (t-s) } \frac{1}{\Gset}\lft|  \sum_{j=t+1}^e \lft( \lft\lb X_j, \beta^*_\te - \beta^*_\st \rgt\rb_\lt^2 - \Sigma[\beta^*_\st - \beta^*_\te, \beta^*_\st - \beta^*_\te] \rgt) \rgt|^2\rgt] = O(1),
\end{align}
This lead us to 
$$
\eqref{eq: f41} + \eqref{eq: f42} = O_p\lft( \sqrt{ \Gset \log^{1+\xi} (t-s) }   \rgt).
$$
\
\\
{\bf Step 2B.}
Using \Cref{lemma: general markov type bound II}, we may have 
$$
\E\lft[  \frac{(e-t)(t-s)}{(e-s)} \frac{1}{\Gset}\lft| \sum_{j=s+1}^t \lft\lb X_j,\beta^*_\te - \beta^*_\st \rgt\rb_\lt\eps_j \rgt|^2 \rgt] = O(t-s).
$$
And again from \Cref{lemma: maximal inequality useful}, it follows that
\begin{align}\nonumber
     &\E\lft[ \max_{s<t\le e} \frac{(e-t)}{(e-s) \log^{1+\xi} (t-s) } \frac{1}{\Gset}\lft| \sum_{j=s+1}^t \lft\lb X_j,\beta^*_\te - \beta^*_\st \rgt\rb_\lt\eps_j \rgt|^2\rgt] = O(1)
     \\\nonumber
     \implies&\E\lft[ \max_{s<t\le e} \frac{(t-s)}{(e-s) \log^{1+\xi} (t-s) } \frac{1}{\Gset}\lft| \sum_{j=t+1}^e \lft\lb X_j,\beta^*_\st - \beta^*_\te \rgt\rb_\lt\eps_j \rgt|^2\rgt] = O(1),
\end{align}
This lead us to 
$$
\eqref{eq: f43} = O_p\lft( \sqrt{ \Gset \log^{1+\xi} (t-s) }   \rgt).
$$

\
\\
{\bf Step 2C.}
Using \Cref{lemma: general markov type bound II}, we may have 
$$
\E\lft[  \frac{(e-t)(t-s)}{(e-s)} \frac{1}{\Gset}\lft| \sum_{j=s+1}^t \lft\lb X_j, \beta^*_\te -\beta^*_\st \rgt\rb_\lt \lft\lb X_j, \beta_j^* - \beta^*_\st \rgt\rb_\lt -  \Sigma[\beta^*_\te -\beta^*_\st, \beta_j^* - \beta^*_\st] \rgt|^2 \rgt] = O(t-s).
$$
And again from \Cref{lemma: maximal inequality useful}, it follows that
\begin{align}\nonumber
     &\E\lft[ \max_{s<t\le e} \frac{(e-t)}{(e-s) \log^{1+\xi} (t-s) } \frac{1}{\Gset}\lft| \sum_{j=s+1}^t \lft\lb X_j, \beta^*_\te -\beta^*_\st \rgt\rb_\lt \lft\lb X_j, \beta_j^* - \beta^*_\st \rgt\rb_\lt -  \Sigma[\beta^*_\te -\beta^*_\st, \beta_j^* - \beta^*_\st] \rgt|^2\rgt]
     \\\nonumber
     &= O(1)
     \\\nonumber
     &\E\lft[ \max_{s<t\le e} \frac{(t-s)}{(e-s) \log^{1+\xi} (t-s) } \frac{1}{\Gset}\lft| \sum_{j=t+1}^e \lft\lb X_j, \beta^*_\st -\beta^*_\te \rgt\rb_\lt\lft\lb X_j, \beta_j^* - \beta^*_\te \rgt\rb_\lt - \Sigma[\beta^*_\st -\beta^*_\te, \beta_j^* - \beta^*_\te] \rgt|^2\rgt]
     \\\nonumber
     &= O(1),
\end{align}
This lead us to 
$$
\eqref{eq: f44} + \eqref{eq: f45} = O_p\lft( \sqrt{ \Gset \log^{1+\xi} (t-s) }   \rgt).
$$

\
\\
{\bf Step 2D.}
Combining the results in Step~$2$A, Step~$2$B and Step~$2$C, we get
    $$
    \max_{s<t\le e} \frac{1}{\sqrt{ \Gset \log^{1+\xi} (t-s) }} \lft| W^{*s,e}_t - \Gset \rgt| = O_p\lft( 1 \rgt).
    $$  
\end{proof}

\begin{lemma}\label{lemma: new localization}
    Let $\xi >0$ and $\se \subset (0,n]$. Suppose $ \eta_{k-1} < s < \eta_k < e < \eta_{k+1}$. Then we have uniformly for all $ t\in \se$,
    \begin{align}\label{eq: bias1}
    &\sum_{j=s+1}^t \lft(Y_j - \lft\lb X_j, \bhast \rgt\rb_\lt \rgt)^2 - \sum_{j=s+1}^t \lft(Y_j - \lft\lb X_j, \beta^*_\st \rgt\rb_\lt \rgt)^2 = O_p\lft( (t-s) \delta_{t-s} \log^{1+\xi} (t-s) \rgt),
    \\
    &\sum_{j=t+1}^e \lft(Y_j - \lft\lb X_j, \bhate \rgt\rb_\lt \rgt)^2 - \sum_{j=t+1}^e \lft(Y_j - \lft\lb X_j, \beta^*_\te \rgt\rb_\lt \rgt)^2 = O_p\lft( (e-t) \delta_{e-t} \log^{1+\xi} (e-t) \rgt).        
    \end{align}
    Consequently, following from the union bound we have uniformly for all $\sem \in \J$ and for all $t\in \sem$
    \begin{equation}\label{eq: bias2}
         \sum_{j=s_m+1}^t \lft(Y_j - \lft\lb X_j, \what{\beta}_{(s_m,t]} \rgt\rb_\lt \rgt)^2 - \sum_{j=s_m+1}^t \lft(Y_j - \lft\lb X_j, \beta^*_{(s_m,t]} \rgt\rb_\lt \rgt)^2 = O_p\lft( \lft(n/\Delta\rgt) (t-s_m) \delta_{t-s_m} \log^{1+\xi} (t-s_m) \rgt).    
    \end{equation}
\end{lemma}

\
\\
\begin{proof}
    Observe that
    \begin{align*}
        &\sum_{j=s+1}^t \lft(Y_j - \lft\lb X_j, \bhast \rgt\rb_\lt \rgt)^2 - \sum_{j=s+1}^t \lft(Y_j - \lft\lb X_j, \beta^*_\st \rgt\rb_\lt \rgt)^2
        \\
        = & \sum_{j=s+1}^t \lft\lb X_j, \bhast -\beta^*_\st \rgt\rb_\lt^2 + 2\sum_{j=s+1}^t \lft\lb X_j, \beta^*_\st - \bhast \rgt\rb_\lt\eps_j + 2\sum_{j=s+1}^t \lft\lb X_j, \beta^*_\st - \bhast \rgt\rb_\lt \lft\lb X_j, \beta^*_j - \bhast \rgt\rb_\lt.
    \end{align*}
    Following from \Cref{lemma: empirical excess risk}, we have uniformly
    $$
    \lft|\sum_{j=s+1}^t \lft\lb X_j, \bhast -\beta^*_\st \rgt\rb_\lt^2 \rgt| = O_p\lft( (t-s) \delta_{t-s} \log^{1+\xi} (t-s) \rgt).
    $$
    From \Cref{lemma: consistency flr}
    $$
    \lft| \sum_{j=s+1}^t \lft\lb X_j, \beta^*_\st - \bhast \rgt\rb_\lt\eps_j\rgt| = O_p\lft( (t-s) \delta_{t-s} \log^{1+\xi} (t-s) \rgt).
    $$
    From \Cref{lemma: unfold empirical bound}
    $$
    \lft| \sum_{j=s+1}^t \lft\lb X_j, \beta^*_\st - \bhast \rgt\rb_\lt \lft\lb X_j, \beta^*_j - \bhast \rgt\rb_\lt \rgt| = O_p\lft( (t-s) \delta_{t-s} \log^{1+\xi} (t-s) \rgt).
    $$
    The \eqref{eq: bias1} of this lemma follows from these three bounds. Given the cardinality of $\J$ in \eqref{rm: card J}, the expression \eqref{eq: bias2} follows from \eqref{eq: bias1} by the union bound.
\end{proof}

\subsubsection{Population CUSUM of functional data}
All the notation used in this subsection are specific to this subsection only. We use these general results to prove some results earlier. 
\begin{assumption}
    Let $\{\mathfrak{f}_i\}_{i=1}^m \in \lt$. Assume there are $\{\mathfrak{n}_p\}_{p=0}^{K+1} \subset \{0,1,\ldots,m\}$ such that $0 =\mathfrak{n}_0 <\mathfrak{n}_1 < \ldots <\mathfrak{n}_K < \mathfrak{n}_{K+1} = m $ and 
    $$
    \mathfrak{f}_t \neq \mathfrak{f}_{t+1} \qquad \text{if and only if} \qquad t \in \{\mathfrak{n}_1, \ldots,\mathfrak{n}_p\}.
    $$
\end{assumption}

\
\\
Let $\inf_{1\le p \le K} \|\mathfrak{f}_{\mathfrak{n}_p} - \mathfrak{f}_{\mathfrak{n}_{p+1}} \|_\lt^2 = \inf_{1\le p \le K} \mathfrak{K}_p^2 = \mathfrak{K}^2 $.

\
\\
For $0 \le s < t <e\le m $, the CUSUM statistics is
\begin{equation}\label{eq: population functional cusum}
    \widetilde{\mathfrak{f}}_t^{s,e} = \sqrt{\frac{e-t}{(e-s)(t-s)}}\sum_{i=s+1}^{t} \mathfrak{f}_i - \sqrt{\frac{t-s}{(e-s)(e-t)}}\sum_{i=t+1}^{e} \mathfrak{f}_i.
\end{equation}

It can be easily shown that the CUSUM statistics at \eqref{eq: population functional cusum} are translational invariant. Consequently assuming $\sum_{i=1}^m \mathfrak{f}_i = 0$, we may also write 

\begin{equation}\label{eq: population functional cusum2}
    \widetilde{\mathfrak{f}}_t^{s,e} = \left( \sum_{i=s+1}^t \mathfrak{f}_i -\frac{t}{e-s}\sum_{i=s+1}^e \mathfrak{f}_i \right)/\sqrt{\frac{(t-s)(e-t)}{(e-s)}} = \left( \sum_{i=s+1}^t \mathfrak{f}_i \right)/\sqrt{\frac{(t-s)(e-t)}{(e-s)}} .
\end{equation}
The form at \eqref{eq: population functional cusum2} is useful proving many important properties of CUSUM.

\
\\
The \Cref{lemma: pop cusum one cp} below follows directly from the definition of CUSUM statistics.

\begin{lemma}\label{lemma: pop cusum one cp}
    Suppose $(s,e]$ contains only one change point $ \mathfrak{n}_p$, then
    \begin{equation*}
    \|\wtilde{\mathfrak{f}}^{s,e}_t\|^2_\lt =   
        \begin{cases}
         \frac{t-s}{(e-s)(e-t)}(e-\mathfrak{n}_p)^2 \mathfrak{K}_p^2, \quad t\le \mathfrak{n}_p\\
         \frac{e-t}{(e-s)(t-s)}(\mathfrak{n}_p - s)^2 \mathfrak{K}_p^2, \quad t\ge \mathfrak{n}_p.       
    \end{cases}
    \end{equation*}
    Consequently, we may write
    $$
        \max_{s<t\le e} \|\wtilde{\mathfrak{f}}^{s.e}_t\|^2_\lt = \frac{(e-\mathfrak{n}_p)(\mathfrak{n}_p-s)}{(e-s)} \mathfrak{K}_p^2.
    $$
\end{lemma}

\
\\
\begin{lemma}\label{lemma: pop cusum 1st cp projection}
    Let $\se$ be such that 
    $$
    \mathfrak{n}_{p-1} \le s < \mathfrak{n}_{p} < e.
    $$Then for any $s<t\le\mathfrak{n}_p$,
    $$
        \|\widetilde{\mathfrak{f}}_t^{s,e}\|^2_\lt = \frac{(t-s)(e-\mathfrak{n}_p)}{(\mathfrak{n}_p - s)(e -t)} \|\widetilde{f}_{\mathfrak{n}_p}^{s,e}\|^2_\lt.    
    $$
    Consequently, we may write
    \begin{equation}
            \max_{s<t\le e} \|\widetilde{\mathfrak{f}}_t^{s,e}\|^2_\lt = \max_{\mathfrak{n}_p \le t \le e} \|\widetilde{\mathfrak{f}}_t^{s,e}\|^2_\lt.        
    \end{equation}

\end{lemma}

\
\\
\begin{proof}
With the form outlined at \eqref{eq: population functional cusum2}
    \begin{align*}
        \|\widetilde{\mathfrak{f}}_t^{s,e}\|^2_\lt &= \frac{e-s}{(t-s)(e-t)} \lft\|\sum_{i=s+1}^t \mathfrak{f}_i\rgt\|_\lt^2 = \frac{(e-s) (t-s)^2}{(t-s)(e-t)} \|\mathfrak{f}_{\mathfrak{n}_p}\|_\lt^2
        \\
        &=\frac{(t-s)(e-\mathfrak{n}_p)}{(\mathfrak{n}_p-s)(e-t)} \frac{(e-s)}{(\mathfrak{n}_p - s)(e-\mathfrak{n}_p)} (\mathfrak{n}_p - s)^2 \|\mathfrak{f}_{\mathfrak{n}_p}\|_\lt^2 = \frac{(t-s)(e-\mathfrak{n}_p)}{(\mathfrak{n}_p-s)(e-t)} \|\widetilde{\mathfrak{f}}_t^{s,e}\|^2_\lt
    \end{align*}
    
\end{proof}
\
\\
\begin{lemma}\label{lemma: pop cusum two point projection}
    Let $(s,e]$ contains exactly two change points $\mathfrak{n}_p$ and $\mathfrak{n}_{p+1}$. Then
    $$
    \max_{s<t<e} \|\widetilde{\mathfrak{f}}_t^{s,e}\|^2_\lt \le 2(e - \mathfrak{n}_{p+1})\mathfrak{K}_{p+1}^2 + 2(\mathfrak{n}_p - s)\mathfrak{K}_p^2.
    $$
\end{lemma}

\begin{proof}
    Let 
    \begin{align*}
        \mathfrak{g}_t = 
        \begin{cases}
            \mathfrak{f}_{\mathfrak{n}_{p+1}}, &\qquad\qquad\text{if} \quad s\le t \le \mathfrak{n}_p
            \\
            \mathfrak{f}_t, &\qquad\qquad\text{if} \quad \mathfrak{n}_p + 1 \le t \le \mathfrak{n}_{p+1}.
        \end{cases}
    \end{align*}
    Then $\forall t \ge \mathfrak{n}_{r}$
    \begin{align}\nonumber
        &\wtilde{\mathfrak{f}}_t^{s,e} - \wtilde{\mathfrak{g}}_t^{s,e} = \sqrt{\frac{(e-s)}{(e-t)(t-s)}} \lft( \sum_{i = s+1}^{\mathfrak{n}_p} \mathfrak{f}_i - \sum_{i = s+1}^{\mathfrak{n}_p} \mathfrak{g}_i + \sum_{i = \mathfrak{n}_p+1}^{t} \mathfrak{f}_i - \sum_{i = \mathfrak{n}_p+1}^{t} \mathfrak{g}_i  \rgt) 
        \\\nonumber
        &\qquad\qquad= \sqrt{\frac{(e-s)}{(e-t)(t-s)}} (\mathfrak{n}_p - s) (\mathfrak{f}_{\mathfrak{n}_p} - \mathfrak{f}_{\mathfrak{n}_{p+1}}).
        \\\label{eq: f27a}
        \implies& \|\wtilde{\mathfrak{f}}_t^{s,e} - \wtilde{\mathfrak{g}}_t^{s,e}\|^2_\lt = \frac{(e-s)(\mathfrak{n}_p -s)}{(e-t)(t-s)} (\mathfrak{n}_p - s) \mathfrak{K}_p^2 \le (\mathfrak{n}_p - s) \mathfrak{K}_p^2 .
    \end{align}
Observe that
\begin{align}\label{eq: f27b}
   \max_{s<t\le e} \|\wtilde{\mathfrak{g}}^{s,e}_{t}\|^2_\lt = \|\wtilde{\mathfrak{g}}^{s,e}_{\mathfrak{n}_{p+1}}\|^2_\lt = \frac{ (e-\mathfrak{n}_{p+1})(\mathfrak{n}_{p+1}-s)}{(e-s)  }  \mathfrak{K}_{p+1}^2 \le (e-\mathfrak{n}_{p+1})\mathfrak{K}_{p+1}^2
\end{align}
where the equality follows from the fact that $g_t$ just have one change point and \Cref{lemma: pop cusum one cp}. Observe that
\begin{align}\nonumber
    \max_{s<t\le e} \|\wtilde{\mathfrak{f}}^{s,e}_{t}\|^2_\lt = \max_{\mathfrak{n}_p \le t\le e} \|\wtilde{\mathfrak{f}}^{s,e}_{t}\|^2_\lt &\le 2 \max_{\mathfrak{n}_p \le t\le e}\| \wtilde{\mathfrak{f}}^{s,e}_{t} - \wtilde{\mathfrak{g}}^{s,e}_{t}\|^2_\lt + 2 \max_{\mathfrak{n}_p \le t\le e} \|\wtilde{\mathfrak{g}}^{s,e}_{t}\|^2_\lt
    \\\nonumber
    &\le 2(\mathfrak{n}_p - s)\mathfrak{K}_p^2 + 2 (e-\mathfrak{n}_{p+1})\mathfrak{K}_{p+1}^2,
\end{align}
where the first line follows from \Cref{lemma: pop cusum 1st cp projection} and the triangle inequality, and the last line follows from \eqref{eq: f27a} and \eqref{eq: f27b}.    
\end{proof}

\section{Proof of \Cref{thm: Limiting distribution}}
Prior to presenting the proof of the main theorem, we will establish the existence and finiteness of the long-run variance.

\begin{lemma}\label{lemma: Long Run Variance}
    Suppose the \Cref{assume: model assumption flr CP} hold. For $k\in \{1, \ldots, \Ktilde \}$, the long-run variance defined in \eqref{eq: def long-run variance} exists and is finite.
\end{lemma}

\begin{proof}
Denote 
$$
Z_j =  \frac{\lb X_j, \beta^*_{\eta_{k}} - \beta^*_{\eta_{k+1}}  \rb_\lt \eps_j }{\kappa_k}.
$$
Observe that
$$
\E\lft[ |Z_1|^3 \rgt] \le \sqrt{\E\lft[ \frac{\lb X_j, \beta^*_{\eta_{k}} - \beta^*_{\eta_{k+1}}  \rb_\lt^6 }{\kappa_k^6} \rgt]} \sqrt{\E\lft[ \eps_j^6 \rgt]} = O\lft( \E^{3/2}\lft[ \frac{\lb X_j, \beta^*_{\eta_{k}} - \beta^*_{\eta_{k+1}}  \rb_\lt^2 }{\kappa_k^2} \rgt]\rgt) = O(1),
$$
where the second last equality follows from \Cref{assume: model assumption flr CP}. Given that we have $\sum_{k=1}^\infty \alpha^{1/3}(k) < \infty$ which is implied by $\sum_{k=1}^\infty k^{1/3} \alpha^{1/3}(k) < \infty$ in \Cref{assume: model assumption flr CP}, all the conditions of Theorem 1.7 of \citet{ibragimov1962some}. It follows from therein that $\sigma^2_\infty(k)$ exists and is finite.

\end{proof}

\paragraph{Sketch of the proof of \Cref{thm: Limiting distribution}.}
We refer to A1 and B1 jointly as uniform tightness. Their proof proceeds in multiple steps where we control diverse errors associated with time series functional linear regression modelling uniformly over the seeded intervals. 
Let
\begin{equation}\label{eq: Q-star}
     \mathcal{Q}^*_k(t) = \sum_{j = s_k + 1}^{t} \lft( Y_j - \lb X_j, \beta^*_{\eta_k} \rb_\lt \rgt)^2 + \sum_{j =t+ 1 }^{e_k} \lft( Y_j - \lb X_j, \beta^*_{\eta_{k+1}} \rb_\lt \rgt)^2
\end{equation}
be the population version of the objective function in \eqref{eq: local refinement}. Observe that $\wtilde{\eta}_k$ is the minimiser of $\calQ_k(t)$ and $\eta_k$ is the minimiser of the $\calQ_k^*(t)$. Establishing the limiting distribution in A2, involves understanding the behavior of both $\calQ^*_k(\eta_k + t) - \calQ^*_k(\eta_k)$ and $\calQ_k(\eta_k + t) - \calQ_k(\eta_k)$, for fixed~$t$. 
We show that $\max_{t}\lft|\calQ_k(\eta_k + t) - \calQ_k(\eta_k) -  \calQ^*_k(\eta_k + t) + \calQ^*_k(\eta_k)\rgt| = o_p(1)$, which in turn hinges on the convergence of $\widehat{\beta}_{(s_k, \what{\eta}_k] }$ to $\beta^*_{\eta_k}$ and symmetrically, that of $\widehat{\beta}_{(\what{\eta}_k, e_k]}$ to $\beta^*_{\eta_{k+1}}$ in an appropriate norm. This establishes that $\calQ^*_k(\eta_k + t) - \calQ^*_k(\eta_k)$ and $\calQ_k(\eta_k + t) - \calQ_k(\eta_k)$ have asymptotically same distribution. We then proceed to show that $\calQ^*_k(\eta_k + t) - \calQ^*_k(\eta_k)$ converges strongly to $S_k (t)$, and consequently, $\calQ_k(\eta_k + t) - \calQ_k(\eta_k)$ converges to $S_k (t)$ in distribution.

Finally, we leverage the Argmax continuous mapping theorem (e.g.\ Theorem 3.2.2 of \citealp{van1996weak}) to translate the convergence from the functional to the minimizer of the functional, which leads to A2. In this regime, it is noteworthy that $t$ is only taking discrete values, and we are not invoking any central limit theorems.

In the vanishing regime, additional complexities arise. Since $\kappa_k$ converges to $0$, in the light of tightness demonstrated in B1, 
we invoke the functional CLT and establish that $Q_k^* (\eta_k+ t \kappa_k^{-2} ) \,-\, Q_k^* (\eta_k)$ converges in distribution to a two-sided Brownian motion $\mathbb{W}(t)$, where $1/\kappa_k^{2}$ acts as a local sample size.
The subsequent steps parallels the non-vanishing case but 
additional intricacies arise due to the convergence behavior as $\kappa_k \to 0$.

\
\\
\begin{proof}[Proof of \Cref{thm: Limiting distribution}]
    Let $1 \le k \le  \Ktilde$ be given. By construction and \Cref{lemma: refined interval}, $(s_k, e_k]$ contains only one change point $\eta_k$ and
    $$
    \eta_k - s_k \ge \Delta/5, \qquad \qquad e_k - \eta_k \ge \Delta/5,
    $$
    for large enough $n$. Recall for any $a>0$, $\delta_a \asymp a^{-2r/(2r+1)}$.
    
    \
    \\
    Let $\wtilde{\eta}_k$ denote the minimiser at \eqref{eq: local refinement}. Without loss of generality assume the minimiser $\wtilde{\eta}_k = \eta_k + \gamma$, with $\gamma>0$. The results presented here assume that what we establish in \Cref{thm: Consistency} holds. 

    \
    \\
    \begin{center}
        \underline{\textbf{Uniform tightness:} $\kappa^2_k |\wtilde{\eta}_k - \eta_k | = O_p(1)$}
    \end{center}
    Assume $\gamma\ge \max\{1/\kappa_k^2, 2\}$, if not, the uniform tightness follows directly. Let $\mathcal{Q}_k$ be defined as in \eqref{eq: local refinement}. Since $\mathcal{Q}_k(\eta_k + \gamma)$ is a minimum, we may write
    \begin{align}\nonumber
        0 &\ge \mathcal{Q}_k(\eta_k + \gamma) - \mathcal{Q}_k(\eta_k) = \sum_{j = \eta_k + 1}^{\eta_k + \gamma} \lft( Y_j - \lb X_j,  \widehat{\beta}_{(s_k, \what{\eta}_k] } \rb_\lt \rgt)^2 - \sum_{j = \eta_k + 1}^{\eta_k + \gamma} \lft( Y_j - \lb X_j,  \widehat{\beta}_{(\what{\eta}_k, e_k] } \rb_\lt \rgt)^2. 
    \end{align}
    The preceding inequality is equivalent to
    \begin{align}\nonumber
        0 &\ge \lft( \sum_{j = {\eta_k +1}}^{{\eta_k } + \gamma} \lft( Y_j - \lb X_j,  \widehat{\beta}_{(s_k, \what{\eta}_k] } \rb_\lt \rgt)^2 - \sum_{j = {\eta_k +1}}^{\eta_k + \gamma} \lft( Y_j - \lb X_j,  \beta^*_{\eta_{k}}\rb_\lt \rgt)^2  \rgt)
        \\\nonumber
         & - \lft( \sum_{j = {\eta_k +1}}^{\eta_k + \gamma}  \lft( Y_j - \lb X_j,  \widehat{\beta}_{(\what{\eta}_k, e_k] } \rb_\lt \rgt)^2  - \sum_{j = {\eta_k +1}}^{\eta_k + \gamma} \lft( Y_j - \lb X_j, \beta^*_{\eta_{k+1}}   \rb_\lt \rgt)^2 \rgt)
        \\\nonumber
        &  + \lft( \sum_{j = {\eta_k +1}}^{\eta_k + \gamma} \lft( Y_j - \lb X_j,  \beta^*_{\eta_{k}} \rb_\lt \rgt)^2 - \sum_{j = {\eta_k +1}}^{\eta_k + \gamma} \lft( Y_j - \lb X_j, \beta^*_{\eta_{k+1}}   \rb_\lt \rgt)^2 \rgt),
    \\\label{eq: s3a}
        &= \lft( \sum_{j = {\eta_k +1}}^{\eta_k + \gamma} \lft( Y_j - \lb X_j,  \widehat{\beta}_{(s_k, \what{\eta}_k] } \rb_\lt \rgt)^2 - \sum_{j = {\eta_k +1}}^{\eta_k + \gamma} \lft( Y_j - \lb X_j,  \beta^*_{\eta_k} \rb_\lt \rgt)^2 \rgt) 
        \\\label{eq: s3b}
         &  - \lft( \sum_{j = {\eta_k +1}}^{\eta_k + \gamma}  \lft( Y_j - \lb X_j,  \widehat{\beta}_{(\what{\eta}_k, e_k] } \rb_\lt \rgt)^2  - \sum_{j = {\eta_k +1}}^{\eta_k + \gamma} \lft( Y_j - \lb X_j, \beta^*_{\eta_{k+1}}   \rb_\lt \rgt)^2 \rgt)
        \\\label{eq: s3c}
        &+  \lft(2 \sum_{j = {\eta_k +1}}^{\eta_k + \gamma} \lb X_j, \beta^*_{\eta_{k+1}} - \beta^*_{\eta_{k}}   \rb_\lt\eps_j\rgt) 
        \\\label{eq: s3d}
        &+ \lft( \sum_{j = {\eta_k +1}}^{\eta_k + \gamma} \lb X_j,  \beta^*_{\eta_k} - \beta^*_{\eta_{k+1}} \rb_\lt ^2\rgt). 
    \end{align}
Therefore, we have
$$
 \eqref{eq: s3d} \le \big|\eqref{eq: s3c}\big| + \big|\eqref{eq: s3b}\big| + \big|\eqref{eq: s3a}\big|.
$$

\
\\
Recall $\delta_{\Delta}=\Delta^{-2r/(2r+1)}$. Observe that
\begin{equation}\nonumber
    \begin{aligned}
       &\what{\eta_k} -  s_k \ge \Delta/5 \implies \delta_{\what{\eta}_k -  s_k} = O\lft(\delta_{\Delta}\rgt),
       \\
       &e_k - \what{\eta_k} \ge \Delta/5  \implies \delta_{e_k - \what{\eta}_k} = O\lft(\delta_{\Delta}\rgt).
    \end{aligned}
\end{equation}
Also using 
$\gamma \ge 1/\kappa_k^2$
and $r>1$, we have
\begin{equation}\label{eq: 3quater gamma and half gamma}
    \begin{aligned}
     &\delta_\gamma^{3/4} \sqrt{\log^{1+\xi} \gamma} = O\lft( \gamma^{-1/2} \rgt) = O\lft( \kappa_k \rgt),
    \\
    &\delta_\gamma^{1/2} \sqrt{\log^{1+\xi} \gamma} = O\lft( \gamma^{-1/3} \rgt) = O \lft(  \kappa_k^{2/3} \rgt). 
\end{aligned}
\end{equation}
From \Cref{assume: SNR Signal to Noise ratio}, we get
\begin{equation}
    \begin{aligned}\label{eq: full first and back end}
    &\delta_{\weta - s} \log^{1+2\xi} (\weta -s) = O\lft(\delta_{\Delta} \log^{1+2\xi} (\Delta) \rgt) = o(\kappa_k^2),
    \\
     &\delta_{e - \weta} \log^{1+2\xi} (e - \weta) = O\lft(\delta_{\Delta} \log^{1+2\xi} (\Delta) \rgt) = o(\kappa_k^2).
\end{aligned}    
\end{equation}
With \eqref{eq: 3quater gamma and half gamma} and \eqref{eq: full first and back end} we get
\begin{equation}\label{eq: half and quater lambda}
    \begin{aligned}
     &\delta_\gamma^{3/4} \sqrt{\log^{1+\xi} \gamma} \lft(\delta_{\weta_k - s_k}\rgt)^{1/4} \sqrt{\log^{1+2\xi} (\weta_k - s_k)} = o(\kappa_k^2)
    \\
     &\delta_\gamma^{1/2} \sqrt{\log^{1+\xi} \gamma} \lft(\delta_{\weta_k - s_k}\rgt)^{1/2} \sqrt{\log^{1+2\xi}(\weta_k - s_k)} = o(\kappa_k^2)
    \\
    &\delta_\gamma^{3/4} \sqrt{\log^{1+\xi} \gamma} \lft(\delta_{e_k -\weta_k }\rgt)^{1/4} \sqrt{\log^{1+2\xi}(e_k -\weta_k)} = o(\kappa_k^2)
    \\
     &\delta_\gamma^{1/2} \sqrt{\log^{1+\xi} \gamma} \lft(\delta_{e_k -\weta_k }\rgt)^{1/2} \sqrt{\log^{1+2\xi} (e_k -\weta_k)} = o(\kappa_k^2).
\end{aligned}
\end{equation}
Also, we have from \Cref{thm: Consistency} that
\begin{equation}\label{eq: kappa square term}
    \begin{aligned}
    \lft| \frac{\weta_k - \eta_k}{\weta_k - s} \rgt| \lesssim \lft| \frac{\weta_k - \eta_k}{\Delta} \rgt| = o_p(1),
    \\
    \lft| \frac{\weta_k - \eta_k}{\weta_k - e} \rgt| \lesssim \lft| \frac{\weta_k - \eta_k}{\Delta} \rgt| = o_p(1).
\end{aligned}
\end{equation}

\
\\
{\bf Step 1: the order of magnitude of \eqref{eq: s3a}.}
 Following from \Cref{lemma: tightness}, we have
 \begin{align}\nonumber
     \eqref{eq: s3a} = \;& O_p\lft( \gamma \delta_\gamma^{1/2} \sqrt{\log^{1+\xi} \gamma} \lft(\delta_{\weta_k - s_k}\rgt)^{1/2} \sqrt{\log^{1+2\xi}(\weta_k - s_k)} \rgt)
     \\\nonumber
     &+ O_p\lft(\gamma \delta_\gamma^{3/4} \sqrt{\log^{1+\xi} \gamma} \lft(\delta_{\weta_k - s_k}\rgt)^{1/4} \sqrt{\log^{1+2\xi} (\weta_k - s_k)}\rgt)
     \\\nonumber
     & + O_p\lft( \gamma\kappa_k \lft(\delta_{\weta_k - s_k}\rgt)^{1/2} \sqrt{\log^{1+2\xi}(\weta_k - s_k)}  \rgt) + O_p\lft( \gamma  \delta_{\weta_k - s_k} \log^{1+2\xi} (\weta_k - s_k) \rgt)
     \\\nonumber
     & + O_p\lft( \sqrt{\gamma}\kappa_k \lft\{ \log^{1+\xi} (\gamma\kappa_k^2) + 1 \rgt\} \rgt) + O_p\lft( \gamma \frac{\weta_k - \eta_k}{\weta_k - s_k} \kappa_k^2 \rgt)
     \\\label{eq: I}
     =& \; o_p(\gamma \kappa_k^2) + O_p\lft( \sqrt{\gamma}\kappa_k \lft\{ \log^{1+\xi} (\gamma\kappa_k^2) + 1 \rgt\} \rgt),
 \end{align}
where the last line follows from \eqref{eq: half and quater lambda} and \eqref{eq: kappa square term}.

\
\\
{\bf Step 2: the order of magnitude of \eqref{eq: s3b}.}
 Following from \Cref{lemma: tightness}, we have
 \begin{align}\nonumber
     \eqref{eq: s3b} =& \; O_p\lft( \gamma \delta_\gamma^{1/2} \sqrt{\log^{1+\xi} \gamma} \lft(\delta_{e_k - \weta_k}\rgt)^{1/2} \sqrt{\log^{1+2\xi}(e_k - \weta_k)} \rgt) 
     \\\nonumber
     &+ O_p\lft(\gamma \delta_\gamma^{3/4} \sqrt{\log^{1+\xi} \gamma} \lft(\delta_{e_k - \weta_k}\rgt)^{1/4} \sqrt{\log^{1+2\xi} (e_k - \weta_k)}\rgt)
     \\\nonumber
     & + O_p\lft( \gamma\kappa_k \lft(\delta_{e_k - \weta_k}\rgt)^{1/2} \sqrt{\log^{1+2\xi}(e_k - \weta_k)}  \rgt) + O_p\lft( \gamma  \delta_{\weta - s} \log^{1+2\xi} (e_k - \weta_k) \rgt)
     \\\nonumber
     & + O_p\lft( \sqrt{\gamma}\kappa_k \lft\{ \log^{1+\xi} (\gamma\kappa_k^2) + 1 \rgt\} \rgt) + O_p\lft( \gamma \frac{\weta_k - \eta_k}{e_k - \weta_k} \kappa_k^2 \rgt)
     \\\label{eq: II}
     =& \; o_p(\gamma \kappa_k^2) + O_p\lft( \sqrt{\gamma}\kappa_k \lft\{ \log^{1+\xi} (\gamma\kappa_k^2) + 1 \rgt\} \rgt).
 \end{align}
where the last line follows from \eqref{eq: half and quater lambda} and \eqref{eq: kappa square term}.

\
\\
{\bf Step 3: the order of magnitude of \eqref{eq: s3c}.}
Following from \Cref{lemma: general markov type bound II} and \Cref{lemma: max to Op by peeling}, we have
\begin{equation}\label{eq: III term in thm2 proof}
     \max_{1/\kappa_k^2< \gamma < \eta_{k+1} -\eta_k} \lft|\frac{1}{\sqrt{\gamma}(\log^{1+\xi} \lft( (\gamma) \kappa_k^2 \rgt) + 1) } \frac{1}{\kappa_k} \sum_{j=\eta_k +1}^{ \eta_k + \gamma} \lft\lb X_j, \beta^*_{\eta_{k+1}} - \beta^*_{\eta_k} \rgt\rb_\lt\eps_j\rgt|^2 = O_p\lft(1\rgt).  
\end{equation}
Using \eqref{eq: III term in thm2 proof}, we have
\begin{equation}\label{eq: III}
    \eqref{eq: s3c} = O_p\lft( \sqrt{\gamma}\kappa_k \lft\{ \log^{1+\xi} (\gamma\kappa_k^2) + 1 \rgt\} \rgt).
\end{equation}

\
\\
{\bf Step 4: lower bound of \eqref{eq: s3d}.} Following from \Cref{lemma: general markov type bound I} and \Cref{lemma: max to Op by peeling}, we have
\begin{equation}\label{eq: IV term in thm2 proof}
     \max_{1/\kappa_k^2 < \gamma < \eta_{k+1} -\eta_k } \lft| \frac{1}{\sqrt{\gamma}(\log^{1+\xi} (\gamma\kappa_k^2) + 1) } \frac{1}{\kappa_k} \sum_{j=\eta_k +1}^{\eta_k + \gamma} \lft(\lft\lb X_j, \beta^*_{\eta_{k+1}} - \beta^*_{\eta_k} \rgt\rb_\lt^2  - \kappa_k^2 \rgt)\rgt|^2 = O_p\lft(1\rgt).    
\end{equation}
Using \eqref{eq: IV term in thm2 proof}, we have
\begin{equation}\label{eq: IV}
    \eqref{eq: s3d} \ge \gamma \kappa_k^2 + O_p\lft( \sqrt{\gamma}\kappa_k \lft\{ \log^{1+\xi} (\gamma\kappa_k^2) + 1 \rgt\} \rgt).
\end{equation}

\
\\
Combining \eqref{eq: I}, \eqref{eq: II}, \eqref{eq: III} and \eqref{eq: IV}, we have uniformly for all $\gamma \ge \frac{1}{\kappa_k^2}$ 
$$
\gamma \kappa_k^2 + O_p\lft( \sqrt{\gamma}\kappa_k \lft\{ \log^{1+\xi} (\gamma\kappa_k^2) + 1 \rgt\} \rgt) \le  O_p\lft( \sqrt{\gamma}\kappa_k \lft\{ \log^{1+\xi} (\gamma\kappa_k^2) + 1 \rgt\} \rgt) + o_p(\gamma \kappa_k^2),
$$
which gives us
$$
\kappa_k^2 \lft| \wtilde{\eta}_k - \eta_k \rgt| = O_p(1). 
$$

\
\\
\begin{center}
    {\bf \underline{Limiting distribution:}}    
\end{center}
Recall the definition of $\mathcal{Q}^*_k(\cdot)$ from \eqref{eq: Q-star}. For any given $k \in \{1, \ldots, \Ktilde\}$, given the end points $s_k$ and $e_k$
 and the true coefficients $\beta^*_{\eta_k}$ and $\beta^*_{\eta_{k+1}}$, we have
\begin{align*}
     \eqref{eq: s3c} + \eqref{eq: s3d} &=  \sum_{j = {\eta_k +1}}^{\eta_k + \gamma} \lft( Y_j - \lb X_j,  \beta^*_{\eta_{k}} \rb_\lt \rgt)^2 - \sum_{j = {\eta_k +1}}^{\eta_k + \gamma} \lft( Y_j - \lb X_j, \beta^*_{\eta_{k+1}}   \rb_\lt \rgt)^2 
     \\
     &= Q^*_k(\eta_k + \gamma) - Q^*_k(\eta_k).     
\end{align*}
Following from the proof of uniform tightness, we have uniformly in $\gamma$, as $n\to \infty$, that
$$
\lft| \mathcal{Q}_k(\eta_k + \gamma) - \mathcal{Q}_k(\eta_k) - \lft( \mathcal{Q}_k^*(\eta_k + \gamma) - \mathcal{Q}_k^*(\eta_k) \rgt) \rgt| \le \big|\eqref{eq: s3a}\big| + \big|\eqref{eq: s3b}\big| + \big|\eqref{eq: s3c}\big| + \big|\eqref{eq: s3d}\big|   \xrightarrow{\mathbb{P}} 0.
$$
With Slutsky's theorem, it is sufficient to find the limiting distribution of $\mathcal{Q}_k^*(\eta_k + \gamma) - \mathcal{Q}_k^*(\eta_k)$ when $n\to \infty$.

\
\\
{\bf Non-vanishing regime.} For $\gamma>0$, we have that when $n\to \infty$
\begin{align*}
    &\mathcal{Q}_k^*(\eta_k + \gamma) - \mathcal{Q}_k^*(\eta_k) = \sum_{j = {\eta_k +1}}^{\eta_k + \gamma} \lft( Y_j - \lb X_j,  \beta^*_{\eta_{k}} \rb_\lt \rgt)^2 - \sum_{j = {\eta_k +1}}^{\eta_k + \gamma} \lft( Y_j - \lb X_j, \beta^*_{\eta_{k+1}}   \rb_\lt \rgt)^2
    \\
    = & \sum_{j=\eta_k + 1}^{\eta_k + \gamma} \lft\{ 2 \lft\lb X_j, \beta^*_{\eta_{k+1}} - \beta^*_{\eta_k} \rgt\rb_\lt \eps_j + \lft\lb X_j, \beta^*_{\eta_{k+1}} - \beta^*_{\eta_k} \rgt\rb_\lt^2 \rgt\} \xrightarrow{ \mathcal{D} } \sum_{j=\eta_k + 1}^{\eta_k + \gamma} \lft\{ 2 \varrho_k \lft\lb X_j, \Psi_k \rgt\rb_\lt \eps_j + \varrho_k^2 \lft\lb X_j, \Psi_k \rgt\rb_\lt^2 \rgt\}.
\end{align*}
For $\gamma<0$, we have when $n\to \infty$
\begin{align*}
    &\mathcal{Q}_k^*(\eta_k + \gamma) - \mathcal{Q}_k^*(\eta_k) = \sum_{j = {\eta_k +1}}^{\eta_k + \gamma} \lft( Y_j - \lb X_j,  \beta^*_{\eta_{k}} \rb_\lt \rgt)^2 - \sum_{j = {\eta_k +1}}^{\eta_k + \gamma} \lft( Y_j - \lb X_j, \beta^*_{\eta_{k+1}}   \rb_\lt \rgt)^2
    \\
    = & \sum_{j=\eta_k+ \gamma + 1}^{\eta_k } \lft\{ 2 \lft\lb X_j, \beta^*_{\eta_{k+1}} - \beta^*_{\eta_k} \rgt\rb_\lt \eps_j + \lft\lb X_j, \beta^*_{\eta_k} - \beta^*_{\eta_{k+1}} \rgt\rb_\lt^2 \rgt\}
    \\
    &\qquad\qquad\qquad\xrightarrow{ \mathcal{D} } \sum_{j=\eta_k+ \gamma + 1}^{\eta_k } \lft\{ -2 \varrho_k \lft\lb X_j, \Psi_k \rgt\rb_\lt \eps_j + \varrho_k^2 \lft\lb X_j, \Psi_k \rgt\rb_\lt^2 \rgt\},
\end{align*}
where the last line follows because pointwise convergence implies convergence in $\lb \;, \; \rb_\lt$.

\
\\
From stationarity, the Slutsky's theorem and the Argmax continuous
mapping theorem (e.g. Theorem 3.2.2 of \citet*{van1996weak}), we have
$$
\wtilde{\eta}_k - \eta_k \xrightarrow{\mathcal{D}} \argmin_{\gamma} S_k(\gamma).
$$

\
\\
{\bf Vanishing regime.} Let $m= \kappa_k^{-2}$, and we have that $m \to \infty$ as $n\to \infty$.
For $\gamma>0$, we have that
\begin{align}\nonumber
    &\mathcal{Q}_k^*(\eta_k + \gamma m) - \mathcal{Q}_k^*(\eta_k) = \sum_{j = {\eta_k +1}}^{\eta_k + \gamma m} \lft( Y_j - \lb X_j,  \beta^*_{\eta_{k}} \rb_\lt \rgt)^2 - \sum_{j = {\eta_k +1}}^{\eta_k + \gamma m} \lft( Y_j - \lb X_j, \beta^*_{\eta_{k+1}}   \rb_\lt \rgt)^2
    \\\nonumber
    = & \sum_{j=\eta_k + 1}^{\eta_k + \gamma m} \lft\{ 2 \lft\lb X_j, \beta^*_{\eta_{k+1}} - \beta^*_{\eta_k} \rgt\rb_\lt \eps_j + \lft\lb X_j, \beta^*_{\eta_{k+1}} - \beta^*_{\eta_k} \rgt\rb_\lt^2 \rgt\}
    \\\label{eq: th2}
    = & \frac{2}{\sqrt{m}} \sum_{j=\eta_k + 1}^{\eta_k + \gamma m} \lft\{ \frac{\lft\lb X_j, \beta^*_{\eta_{k+1}} - \beta^*_{\eta_k} \rgt\rb_\lt} {\kappa_k} \eps_j \rgt\} + \frac{1}{m}  \sum_{j=\eta_k + 1}^{\eta_k + \gamma m} \lft\{ \frac{ \lft\lb X_j, \beta^*_{\eta_{k+1}} - \beta^*_{\eta_k} \rgt\rb_\lt^2}{\kappa_k^2} -1  \rgt\} +  \frac{1}{m}  \sum_{j=\eta_k + 1}^{\eta_k + \gamma m} 1.
\end{align}
Following from the definition of the long-run variance and \Cref{lemma: CLT alpha-mixing}, we have 
\begin{equation}\label{eq: th21}
    \frac{1}{\sigma_\infty(k)}\frac{2}{\sqrt{m}} \sum_{j=\eta_k + 1}^{\eta_k + \gamma m} \lft\{ \frac{\lft\lb X_j, \beta^*_{\eta_{k+1}} - \beta^*_{\eta_k} \rgt\rb_\lt} {\kappa_k} \eps_j \rgt\} \xrightarrow{\calD} \mathbb{B}_2(\gamma).
\end{equation}
We also have
\begin{equation}\label{eq: th21a}
    \E\lft[ \lft| \frac{1}{m}  \sum_{j=\eta_k + 1}^{\eta_k + \gamma m} \lft\{ \frac{ \lft\lb X_j, \beta^*_{\eta_{k+1}} - \beta^*_{\eta_k} \rgt\rb_\lt^2}{\kappa_k^2} -1  \rgt\} \rgt|^2 \rgt] = O\lft( \frac{\gamma}{m} \rgt) \to 0, 
\end{equation}
following from \eqref{eq: markov I first equation} in \Cref{lemma: general markov type bound I}. Using \eqref{eq: th21a}, \eqref{eq: th21} and $\frac{1}{m}\sum_{j=\eta_k + 1}^{\eta_k + \gamma m} 1 \to \gamma$ in \eqref{eq: th2}, we write
$$
\mathcal{Q}_k^*(\eta_k + \gamma m) - \mathcal{Q}_k^*(\eta_k)
\xrightarrow{ \mathcal{D} } \sigma_\infty(k) \mathbb{B}_2(\gamma) + \gamma
$$
where $\mathbb{B}_2(\gamma)$ is a standard Brownian motion.


\
\\
Similarly, for $\gamma < 0$, we may have when $n \to \infty$
$$
\mathcal{Q}_k^*(\eta_k + \gamma m) - \mathcal{Q}_k^*(\eta_k) \xrightarrow{\calD} -\gamma + \sigma_\infty(k)\mathbb{B}_1(-\gamma),
$$
where $\mathbb{B}_1(r)$ is a standard Brownian motion. Let $Z_j^* = \frac{\lb X_j, \beta_{\eta_{k+1}} - \beta_{\eta_{k}} \rb_\lt \eps_j}{\kappa_k}$. To see the independence of $\mathbb{B}_1(r)$ and $\mathbb{B}_2(r)$ note that
\begin{align}\nonumber
    \frac{1}{m} \E\lft[ \lft(\sum_{t = -m\gamma}^{-1} Z_t^* \rgt) \lft( \sum_{t = 1}^{m\gamma} Z_t^*\rgt) \rgt] &= \frac{1}{m} \lft\{ \sum_{k=1}^{m\gamma} k \E[Z_1 Z_{1+k}] + \sum_{k= m\gamma+1}^{2m\gamma} (2q-k) \E[Z_1 Z_{1+k}]  \rgt\}
    \\\nonumber
    &\le \frac{1}{m} \sum_{k=1}^{2m\gamma} k \lft| \E\lft[  Z_1 Z_{1+k}  \rgt] \rgt| \le \frac{(2m\gamma)^{2/3}}{m} \sum_{k=1}^{2m\gamma} k^{1/3} \|Z_1\|_3^2 \alpha^{1/3}(k) = O\lft(\frac{1}{m^{1/3}} \rgt) \to 0,
\end{align}
where the second last inequality follows from \Cref{lemma: mixing covariance inequality} and stationarity and the last inequality follows from $\sum_{k=1}^\infty k^{1/3} \alpha^{1/3}(k) < \infty$ and
\begin{equation}\nonumber
    \|Z_1\|_3 = \E^{1/3}\lft[ \frac{1}{\kappa_k^3}\lft\lb X_1, \beta^*_{\eta_k} - \beta^*_{\eta_{k+1}} \rgt\rb^3 \E\lft[\eps_1^3|X_1\rgt]  \rgt] \le O(1) \E^{1/6}\lft[ \frac{1}{\kappa_k^6}\lft\lb X_1, \beta^*_{\eta_k} - \beta^*_{\eta_{k+1}} \rgt\rb^6  \rgt] = O(1),    
\end{equation}
which follows from \Cref{assume: model assumption flr CP}.

\
\\
From the Slutsky's theorem and the Argmax continuous mapping theorem we have
$$
\wtilde{\eta}_k - \eta_k \xrightarrow{\mathcal{D}} \argmin_{\gamma} \lft\{ |\gamma| + \sigma_\infty(k)\mathbb{W}(\gamma) \rgt\}.
$$
\end{proof}

\subsection{Technical result for the proof of \Cref{thm: Limiting distribution}}
Recall for any $a>0$, $\delta_a \asymp a^{-2r/(2r+1)}$.
\begin{lemma}\label{lemma: tightness}
    Let $ \eta_{k-1}< s <  \eta_k < e < \eta_{k+1}$ be fixed. Let $\xi >0$. Then,
    \begin{align*}
        &\max_{\substack{\gamma  \in \lft(1/\kappa_k^2 , \eta_{k+1} - \eta_k \rgt) }} \frac{1}{\mathcal{H}(\weta_k-s,\gamma)}\lft(\sum_{j = \eta_k +1 }^{\eta_k + \gamma}\lft( Y_j - \lb X_j,  \widehat{\beta}_{(s,\weta_k]} \rb_\lt \rgt)^2 - \sum_{j = \eta_k + 1}^{\eta_k + \gamma} \lft( Y_j - \lb X_j,  \beta^*_{\eta_k} \rb_\lt \rgt)^2  \rgt) = O_p(1)
    \end{align*}
    where for any $t \in \se$
    \begin{align*}
            \mathcal{H}(t-s, \gamma)= & \gamma \lft\{ \lft( \delta_{t-s}^{1/4} + \delta_{\gamma}^{1/4} \rgt) \delta_{\gamma}^{1/2} \sqrt{\log^{1+\xi} \gamma} + \lft(\kappa_k + \delta_{t-s}^{1/2}\sqrt{\log^{1+\xi} (t-s)} \rgt) \delta_{t-s}^{1/4} \rgt\} \sqrt{\log^{1+\xi} (t-s)} \delta_{t-s}^{1/4} 
            \\
            + &\kappa_k\sqrt{\gamma}\lft\{ \log^{1+\xi} (\gamma\kappa_k^2) + 1 \rgt\} + \lft|\frac{t-\eta_k}{t-s} \rgt|\gamma \kappa_k^2       
    \end{align*}

\end{lemma} 

\
\\
\begin{proof}
    Let $t>\eta_k$. The case when $t \le \eta_k$ follows similar to the proof outlined below. Observe that
    $$
    \beta^*_\st - \beta_{\eta_k}^* = \lft( \frac{t-\eta_k}{t -s} (\beta^*_{\eta_{k+1}} - \beta^*_{\eta_k}) \rgt).
    $$
    We may write the expression
    \begin{align}\nonumber
        & \lft(\sum_{j = \eta_k +1 }^{\eta_k + \gamma}\lft( Y_j - \lb X_j,  \widehat{\beta}_\st \rb_\lt \rgt)^2 - \sum_{j = \eta_k + 1}^{\eta_k + \gamma} \lft( Y_j - \lb X_j,  \beta^{*}_{\eta_k} \rb_\lt \rgt)^2  \rgt) 
        \\\label{eq: 3a1}
        =& \lft(\sum_{j = \eta_k +1 }^{\eta_k + \gamma}\lft( Y_j - \lb X_j,  \widehat{\beta}_\st \rb_\lt \rgt)^2 - \sum_{j = \eta_k + 1}^{\eta_k + \gamma} \lft( Y_j - \lb X_j,  \beta^*_\st\rb_\lt \rgt)^2  \rgt) 
        \\\label{eq: 3a2}
        &\qquad+ \lft(\sum_{j = \eta_k +1 }^{\eta_k + \gamma}\lft( Y_j - \lb X_j,  \beta^*_\st \rb_\lt \rgt)^2 - \sum_{j = \eta_k + 1}^{\eta_k + \gamma} \lft( Y_j - \lb X_j,  \beta^*_{\eta_k} \rb_\lt \rgt)^2  \rgt).
    \end{align}
    We show in Step~$1$ that
    \begin{align}\label{eq: Step I}
        & \max_{\substack{t \in \se \\ \gamma  \in \lft(1/\kappa_k^2 , \eta_{k+1} - \eta_k \rgt)}} \frac{1}{\mathcal{H}_1(t-s,\gamma)} \lft(\sum_{j = \eta_k + 1}^{\eta_k + \gamma}\lft( Y_j - \lb X_j,  \widehat{\beta}_\st \rb_\lt \rgt)^2 - \sum_{j = \eta_k + 1}^{\eta_k + \gamma} \lft( Y_j - \lb X_j,  \beta^*_{\st} \rb_\lt \rgt)^2  \rgt) = O_p(1),
    \end{align}
    where
    \begin{align*}
        \mathcal{H}_1(t-s,\gamma) = \gamma \bigg\{& \lft( \delta_{t-s}^{1/4} + \delta_{\gamma}^{1/4} \rgt) \delta_{\gamma}^{1/2} \sqrt{\log^{1+\xi} (\gamma)} 
        \\
        &+  \lft(\kappa_k + \delta_{t-s}^{1/2}\sqrt{\log^{1+\xi} (t-s)} \rgt) \delta_{t-s}^{1/4} \bigg\}  \delta_{t-s}^{1/4} \sqrt{\log^{1+\xi} (t-s)}. 
    \end{align*}
We show in Step~$2$ that
    \begin{align}\nonumber
            \max_{\substack{t \in \se \\ \gamma  \in \lft(1/\kappa_k^2 , \eta_{k+1} - \eta_k \rgt)}} \frac{1}{\mathcal{H}_2(\gamma)} \Bigg(\sum_{j = \eta_k +1 }^{ \eta_k + \gamma}\lft( Y_j - \lb X_j,  \beta^*_\st \rb_\lt \rgt)^2 &- \sum_{j = \eta_k +1}^{ \eta_k + \gamma} \lft( Y_j - \lb X_j,  \beta^*_{\eta_k} \rb_\lt \rgt)^2  - \lft\{ \lft(\frac{\eta_k - s}{t - s} \rgt)^2 - 1\rgt\} \gamma\kappa_k^2 \Bigg) 
    \\\label{eq: Step II}
    =& O_p\lft( 1\rgt).
    \end{align}
where
$$
\mathcal{H}_2(\gamma) = \kappa_k\sqrt{\gamma}\lft\{\log^{1+\xi} \lft(\gamma\kappa_k^2 \rgt) + 1\rgt\}.
$$
The bound for \eqref{eq: 3a1} follows from \eqref{eq: Step I} and the bound for \eqref{eq: 3a2} follows from \eqref{eq: Step II} and the realization that $\lft\{1 - \lft(\frac{\eta_k - s}{t - s}  \rgt)^2 \rgt\} \ge \lft(1 - \frac{\eta_k - s}{t-s} \rgt) = \lft(\frac{t - \eta_k}{t-s} \rgt)  $.


\
\\
{\bf Step 1:}
Observe that
        \begin{align}\nonumber
            &\lft| \sum_{j = \eta_k + 1}^{ \eta_k + \gamma}\lft( Y_j - \lb X_j,  \widehat{\beta}_\st \rb_\lt \rgt)^2 - \sum_{j = \eta_k + 1}^{ \eta_k + \gamma} \lft( Y_j - \lb X_j,  \beta^*_{\st} \rb_\lt \rgt)^2 \rgt|
            \\\nonumber
            =& \lft| \sum_{j = \eta_k + 1}^{ \eta_k + \gamma} \lft( \lft\lb X_j, \bhast - \beta^*_\st \rgt\rb_\lt^2 - 2 \lft\lb X_j, \bhast - \beta^*_\st \rgt\rb_\lt \lft\lb X_j, \beta^*_{\eta_{k+1}} - \beta^*_\st \rgt\rb_\lt - 2 \lft\lb X_j, \bhast - \beta^*_\st \rgt\rb_\lt \eps_j \rgt) \rgt|
            \\\label{eq: 3a11}
            \le & \lft| \sum_{j = \eta_k + 1}^{ \eta_k + \gamma}  \lft\lb X_j, \bhast - \beta^*_\st \rgt\rb_\lt^2 \rgt| 
            \\\label{eq: 3a12}
            +&2\lft| \sum_{j = \eta_k + 1}^{ \eta_k + \gamma}   \lft\lb X_j, \bhast - \beta^*_\st \rgt\rb_\lt \lft\lb X_j, \beta^*_{\eta_{k+1}} - \beta^*_\st \rgt\rb_\lt \rgt|
            \\\label{eq: 3a13}
            +& 2\lft| \sum_{j = \eta_k + 1}^{ \eta_k + \gamma}   \lft\lb X_j, \bhast - \beta^*_\st \rgt\rb_\lt \eps_j\rgt|.
        \end{align}
We are going to bound \eqref{eq: 3a11}, \eqref{eq: 3a12} and \eqref{eq: 3a13} in the following three sub-steps. Following from \Cref{lemma: empirical excess risk II} we have that
    \begin{equation}\nonumber\label{eq: 6a5}
        \eqref{eq: 3a11} = O_p\lft( \gamma \lft\{ \lft( \delta_{t-s}^{1/4} + \delta_{\gamma}^{1/4} \rgt) \delta_{\gamma}^{1/2} \sqrt{\log^{1+\xi} (\gamma)} + \delta_{t-s}^{3/4}\sqrt{\log^{1+\xi} (t-s)} \rgt\} \sqrt{\log^{1+\xi} (t-s)} \delta_{t-s}^{1/4}  \rgt).
    \end{equation}
Following from \Cref{lemma: unfold empirical bound II}, we have that
    \begin{equation}\nonumber\label{eq: 6b5}
        \eqref{eq: 3a12} = O_p\lft( \gamma \lft\{ \lft( \delta_{t-s}^{1/4} + \delta_{\gamma}^{1/4} \rgt) \delta_{\gamma}^{1/2} \sqrt{\log^{1+\xi} (\gamma)} + \kappa_k\delta_{t-s}^{1/4}\rgt\} \sqrt{\log^{1+\xi} (t-s)} \delta_{t-s}^{1/4}  \rgt).
    \end{equation}
And following from \Cref{lemma: consistency II}, we have that
\begin{equation*}
    \eqref{eq: 3a13} = O_p\lft(  \lft\{1 + \lft(\frac{\delta_{\gamma}}{\delta_{t-s}} \rgt)^{1/4}\rgt\} \delta_{\gamma}^{1/2} \sqrt{\log^{1+\xi} (\gamma)}  \delta_{t-s}^{1/2} \sqrt{\log^{1+\xi} (t - s)}  \rgt).
\end{equation*}
    The stochastic bound \eqref{eq: Step I} now follows directly from these three bounds on \eqref{eq: 3a11}, \eqref{eq: 3a12} and \eqref{eq: 3a13}. 

\
\\
{\bf Step 2:}
Observe that $\beta^*_\st - \beta^*_{\eta_k} = \frac{t-\eta_k}{t-s} (\beta^*_{\eta_{k+1}} - \beta^*_{\eta_k})$. We may write the expansion
    \begin{align*}
        &\lft(\sum_{j = \eta_k}^{ \eta_k + \gamma}\lft( Y_j - \lb X_j,  \beta^*_\st \rb_\lt \rgt)^2 - \sum_{j = \eta_k}^ { \eta_k + \gamma} \lft( Y_j - \lb X_j,  \beta^*_{\eta_k} \rb_\lt \rgt)^2 - \lft\{ \lft(\frac{\eta_k - s}{t - s} \rgt)^2 - 1\rgt\} \gamma\kappa^2 \rgt)
        \\
        =& \lft\{ \lft(\frac{s - \eta_k}{t - s} \rgt)^2 - 1\rgt\}  \sum_{j=\eta_k +1}^{ \eta_k + \gamma} \lft( \lft\lb X_j, \beta^*_{\eta_{k+1}} - \beta^*_{\eta_k} \rgt\rb_\lt^2 -\kappa_k^2\rgt) - 2\lft( \frac{t-\eta_k}{t-s} \rgt)\sum_{j=\eta_k +1}^{  \eta_k + \gamma} \lft\lb X_j, \beta^*_{\eta_{k+1}} - \beta^*_{\eta_k} \rgt\rb_\lt\eps_j.
    \end{align*}
Consequently, we have that
\begin{align}\nonumber
    &\lft|\sum_{j = \eta_k + 1}^{ \eta_k + \gamma}\lft( Y_j - \lb X_j,  \beta^*_\st \rb_\lt \rgt)^2 - \sum_{j = \eta_k + 1}^{ \eta_k + \gamma} \lft( Y_j - \lb X_j,  \beta^*_{\eta_k +1} \rb_\lt \rgt)^2 - \lft\{ \lft(\frac{\eta_k - s}{t - s} \rgt)^2 - 1\rgt\}  \gamma\kappa_k^2   \rgt|
        \\\nonumber
        \le & \lft| \lft\{ \lft(\frac{\eta_k - s}{t - s} \rgt)^2 - 1\rgt\}  \sum_{j=\eta_k +1}^{ \eta_k + \gamma} \lft( \lft\lb X_j, \beta^*_{\eta_{k+1}} - \beta^*_{\eta_k} \rgt\rb_\lt^2 -\kappa_k^2\rgt)\rgt| + 2\lft| \lft( \frac{t-\eta_k}{t-s} \rgt)\sum_{j=\eta_k +1}^{  \eta_k + \gamma} \lft\lb X_j, \beta^*_{\eta_{k+1}} - \beta^*_{\eta_k} \rgt\rb_\lt\eps_j \rgt| 
        \\\label{eq: 77a}
        \le &  \lft|\sum_{j=\eta_k +1}^{ \eta_k + \gamma} \lft( \lft\lb X_j, \beta^*_{\eta_{k+1}} - \beta^*_{\eta_k} \rgt\rb_\lt^2 - \kappa_k^2 \rgt) \rgt| 
        \\\label{eq: 77b}
        &+ 2\lft| \sum_{j=\eta_k +1}^{ \eta_k + \gamma} \lft\lb X_j, \beta^*_{\eta_{k+1}} - \beta^*_{\eta_k} \rgt\rb_\lt\eps_j \rgt|,
\end{align}
where the last two line follows from $0 \le \lft\{ 1 - \lft(\frac{\eta_k - s}{t - s} \rgt)^2\rgt\} \le 1$ and $\lft(\frac{t -\eta_k}{t-s}\rgt) \le 1$. For the expression \eqref{eq: 77a}, using \eqref{eq: IV term in thm2 proof}, we get that
\begin{equation}\label{eq: 7a11}
    \lft|\sum_{j=\eta_k +1}^{ \eta_k + \gamma} \lft( \lft\lb X_j, \beta^*_{\eta_{k+1}} - \beta^*_{\eta_k} \rgt\rb_\lt^2 - \kappa_k^2 \rgt) \rgt|  = O_p\lft( \sqrt{\gamma}\kappa_k \lft\{ \log^{1+\xi} \lft(\gamma \kappa_k^2\rgt) + 1 \rgt\} \rgt).    
\end{equation}
For the expression \eqref{eq: 77b}, we use \eqref{eq: III term in thm2 proof} to have
\begin{equation}\label{eq: 7a12}
    \lft| \sum_{j=\eta_k +1}^{ \eta_k + \gamma} \lft\lb X_j, \beta^*_{\eta_{k+1}} - \beta^*_{\eta_k} \rgt\rb_\lt\eps_j \rgt| = O_p\lft( \sqrt{\gamma}\kappa_k \lft\{ \log^{1+\xi} \lft(\gamma \kappa_k^2\rgt) + 1 \rgt\} \rgt).
\end{equation}
Bringing \eqref{eq: 7a11} and \eqref{eq: 7a12} together shall establish \eqref{eq: Step II}.
\end{proof}

\
\\
\begin{lemma}\label{lemma: refined interval}
Let $(s_k, e_k]$ be the refined interval constructed in \eqref{eq: def s_k and e_k}. Then, under the event $\mathcal{A}$ defined in \eqref{eq: event A}, $\eta_k$ is the one and only change point lying in $(s_k, e_k]$.
\
Additionally, under the same event $\mathcal{A}$, we have
$$
\min\bigg\{ e_k -\eta_k, s_k - \eta_k \bigg\} \ge \Delta/5 .
$$

Since, event $\mathcal{A}$ is asymptotically almost sure (\Cref{lemma: probability bound first lemma}). These results holds with probability converging to $1$ as $n \to \infty$.
\end{lemma}
\begin{proof}
    In the last seeded intervals layer, we have $\mathfrak{l}_k = \Delta$ and $\mathfrak{b}_k = \Delta/2$. Let ${\weta}_k \in (i \Delta/2, (i+1)\Delta/2]$, for some $i$. Without loss of generality, we assume $i=1$; if not, we translate the intervals by $(i-1)\Delta/2$ unit to right. Then $\weta_k$ would be contained in $(0, \Delta]$ and $(\Delta/2, 3\Delta/2]$ in this last seeded intervals layer. By construction, we have $s_k = 0$ and $e_k = 3\Delta/2$. Following from \Cref{thm: Consistency}, under the event $\mathcal{A}$, we have $|\weta_k - \eta_k| \le \Delta/4$. Therefore $\eta_k \in (\Delta/4, 5 \Delta/4]$.
\end{proof}

\section{Proof of \Cref{thm: long run variance estimation}}
\begin{proof}
Let $\calP = \{J_1, J_2, \ldots, J_S \}$. Let $J_1 = \{t_1, t_1 + 1, \ldots, t_1 + (q - 1), \ldots,  t_1 + (2q -1) \}$. Denote $\wtilde{J}_1 = J_1 \setminus (J_1 + q) = \{t_1, t_1, \ldots, t_1 + (q-1)\}$ and $\Bar{J}_1 = J_1 \setminus \wtilde{J_1} =\{ t_1 + q, \ldots, t_1 + (2q-1)\}$ as the two equal partition of the block $J_1$. Recall that $\delta_a \asymp a^{-2r/2r+1}$ for any $a>0$. Denote the population version of the process $\{F^*_{J_v}\}_{v=1}^{S}$ as
    \begin{equation}\nonumber
        F_{J_1}^* = \sqrt{\frac{2}{q}} \lft\{ \sum_{t\in \wtilde{J}_1} \lft( Z_t^*   -  Z_{t+q}^* \rgt\} \rgt) = \sqrt{\frac{2}{q}} \lft\{ \sum_{t\in \wtilde{J}_1} Z_t^*   - \sum_{t \in \Bar{J_1}} Z_t^* \rgt\}.
    \end{equation}
    where $Z_t^* = \frac{1}{\kappa_k}\lft\lb X_t, \beta^*_{\eta_k} - \beta^*_{\eta_{k+1}} \rgt\rb\eps_t$.
    
This proof is further divided into two steps. Firstly, we establish the consistency of the population version of the estimate. Secondly, we conclude the proof by demonstrating that the deviation of the estimate from the estimator is small in probability. The last redundant step is replacing $\kappa_k$ with $\what{\kappa}_k$ and applying \Cref{lemma: consistent kappa} along with the Slutsky's theorem.

\
\\
{\bf Step 1a:} Note that
$$
\E\lft[ (F_{J_1}^*)^2 \rgt] = 2\, \E\lft[ \lft( \frac{1}{\sqrt{q}}\sum_{t\in \wtilde{J}_1} Z_t^* \rgt)^2 \rgt] + 2\, \E\lft[ \lft( \frac{1}{\sqrt{q}}\sum_{t\in \Bar{J}_1} Z_t^* \rgt)^2 \rgt] - \frac{4}{q} \E\lft[ \lft(\sum_{t\in \wtilde{J}_1} Z_t^* \rgt) \lft( \sum_{t\in \Bar{J}_1} Z_t^*\rgt) \rgt].
$$
Following stationarity, we may write
\begin{align}\nonumber
    \frac{1}{q} \E\lft[ \lft(\sum_{t\in \wtilde{J}_1} Z_t^* \rgt) \lft( \sum_{t\in \Bar{J}_1} Z_t^*\rgt) \rgt] &= \frac{1}{q} \lft\{ \sum_{t=1}^q t \E[Z_1 Z_{1+t}] + \sum_{t=q+1}^{2q} (2q-t) \E[Z_1 Z_{1+t}]  \rgt\}
    \\\nonumber
    &\le \frac{1}{q} \sum_{t=1}^{2q} t \lft| \E\lft[  Z_1 Z_{1+t}  \rgt] \rgt| \le \frac{(2q)^{2/3}}{q} \sum_{t=1}^{2q}  t^{1/3} \|Z_1\|_3^2 \alpha^{1/3}(t) = O\lft(\frac{1}{q^{1/3}} \rgt) \to 0,
\end{align}
where the second last inequality follows from \Cref{lemma: mixing covariance inequality} and stationarity and the last inequality follows from $\sum_{k=1}^\infty k^{1/3} \alpha^{1/3}(k) < \infty$ and
\begin{equation}\label{eq: lrv proof1}
    \|Z_1\|_3 = \E^{1/3}\lft[ \frac{1}{\kappa_k^3}\lft\lb X_1, \beta^*_{\eta_k} - \beta^*_{\eta_{k+1}} \rgt\rb^3 \E\lft[\eps_1^3|X_1\rgt]  \rgt] \le O \lft(\E^{1/6}\lft[ \frac{1}{\kappa_k^6}\lft\lb X_1, \beta^*_{\eta_k} - \beta^*_{\eta_{k+1}} \rgt\rb^6  \rgt] \rgt) = O(1),    
\end{equation}
which follows from \Cref{assume: model assumption flr CP}.

\
\\
From this, the definition of the long run variance and stationarity at \eqref{eq: def long-run variance}, we can write
\begin{equation}\label{eq: lrv proof2}
    \E\lft[ (F_{J_v}^*)^2 \rgt] \to  \sigma_\infty^2(k), \qquad\text{as}\qquad q \to \infty,
\end{equation}
for all $J_v \in \calP$.

\
\\
{\bf Step 1b:}
We have from \Cref{assume: model assumption flr CP} that $\sum_{k=1}^\infty (k+1)^{8/3 -1}\alpha^{(4/3)/(8/3+4/3)}(k) < \infty$, $\E\lft[ Z_t^* \rgt] = 0$ and similar to \eqref{eq: lrv proof1} that 
$$
\|Z_1^*\|_{4} =  \E^{1/4}\lft[ \frac{1}{\kappa_k^4}\lft\lb X_1, \beta^*_{\eta_k} - \beta^*_{\eta_{k+1}} \rgt\rb^4 \E\lft[\eps_1^4|X_1\rgt]  \rgt] \le O \lft(\E^{1/6}\lft[ \frac{1}{\kappa_k^6}\lft\lb X_1, \beta^*_{\eta_k} - \beta^*_{\eta_{k+1}} \rgt\rb^6  \rgt] \rgt) < \infty.
$$
All of the conditions of Theorem 1 of \citet{yokoyama1980moment} are satisfied and therefore
$$
\E\lft[ \lft| \sum_{t\in \wtilde{J}_1} Z_t^*\rgt|^{8/3} \rgt] = O(q^{4/3}).
$$
Following stationarity for all $v \in \{1, \ldots, S\}$, it implies 
$$\E\lft[\lft| (F_{J_v}^*)^2 - \E\lft[(F_{J_v}^*)^2 \rgt]\rgt|^{4/3} \rgt] \le 2^{1/3}\E\lft[\lft| (F_{J_v}^*)^2 \rgt|^{4/3}\rgt] \le  4 \lft( \E\lft[\lft| \frac{1}{\sqrt{q}} \sum_{t\in \wtilde{J}_v} Z_t^* \rgt|^{8/3} \rgt] +  \E\lft[\lft| \frac{1}{\sqrt{q}} \sum_{t\in \wtilde{J}_v} Z_t^* \rgt|^{8/3} \rgt]   \rgt) < \infty,
$$
where we used $(a+b)^{4/3} \le 2^{1/3}(a^{4/3} + b^{4/3})$ in the last and the second last inequality. We also have $\alpha(k) = O\lft( \frac{1}{k^4}\rgt) $ which follows from the summability of $\{k^{1/3}\alpha^{1/3}(k)\}_{k=1}^\infty$. With $\rho = 8$ and $p=4/3$, it is what follows that
$$
\sum_{v=1}^S \frac{\E^{2/\rho}\lft[ \lft| (F_{J_v}^*)^2 - \E\lft[(F_{J_v}^*)^2 \rgt]\rgt|^p \rgt]}{v^p} \le O(1) \sum_{v=1}^\infty \frac{1}{v^p} < \infty.
$$
We have all the condition of \Cref{SLLN: alpha mixing} satisfied with $\rho = 8$ and $p=4/3$, therefore, 
$$
\frac{1}{S} \sum_{v=1}^S (F_{J_v}^*)^2 - \E\lft[(F^*_{J_v})^2 \rgt] \to 0 \qquad \text{a.s.}
$$
Combining this with \eqref{eq: lrv proof2} and and the stationarity of $\{F_{J_v}^*\}_{v=1}^S$, we write
\begin{equation}\label{eq: lrv proof3}
    \frac{1}{S} \sum_{v=1}^S (F_{J_v}^*)^2 \to \sigma^2_\infty(k) \qquad \text{a.s.}
\end{equation}

\
\\
{\bf Step 2:}
Let
$$
\what{F}_{J_v} = \sqrt{\frac{2}{q}} \lft\{ \sum_{t\in \wtilde{J}_v} \frac{\lft\lb X_t, \what{\beta}_{k}^\1 - \what{\beta}_{k}^\2 \rgt\rb_\lt}{\kappa_k} \lft( Y_t - \lft\lb X_t, \what{\beta}_{J} \rgt\rb_\lt \rgt) - \sum_{t\in \Bar{J}_1} \frac{\lft\lb X_t, \what{\beta}_{k}^\1 - \what{\beta}_{k}^\2 \rgt\rb_\lt}{\kappa_k} \lft( Y_t - \lft\lb X_t, \what{\beta}_{J} \rgt\rb_\lt \rgt)  \rgt\}.
$$
Observe that
\begin{align}\nonumber
   \lft(\what{F}_{J_v}\rgt)^2  = \lft( A_v + F_{J_v}^* + B_v \rgt)^2 ,
\end{align}
where
\begin{align*}
    A_v = \frac{\sqrt{2}}{\kappa_k\sqrt{q}} \Bigg\{ &  \sum_{t\in \wtilde{J}_v} \lft\lb X_t, \what{\beta}_k^\1 - \beta^*_{\eta_k} \rgt\rb_\lt\eps_t + \sum_{t\in \wtilde{J}_v} \lft\lb X_t, \beta^*_{\eta_{k+1}} - \what{\beta}_{k}^\2  \rgt\rb_\lt\eps_t 
    \\
    &+  \sum_{t\in \wtilde{J}_v} \lft\lb X_t, \what{\beta}_k^\1 - \beta^*_{\eta_k} \rgt\rb_\lt\lft\lb X_t, \beta^*_{J_v} - \what{\beta}_{J_v} \rgt\rb_\lt  
    \\
    &+ \sum_{t\in \wtilde{J}_v} \lft\lb X_t, \beta^*_{\eta_{k+1}} - \what{\beta}_{k}^\2 \rgt\rb_\lt \lft\lb X_t, \beta^*_{J_v} - \what{\beta}_{J_v} \rgt\rb_\lt   
    \\
    & - \sum_{t\in \Bar{J}_v} \lft\lb X_t, \what{\beta}_k^\1 - \beta^*_{\eta_k} \rgt\rb_\lt\eps_t - \sum_{t\in \Bar{J}_1} \lft\lb X_t, \beta^*_{\eta_{k+1}} - \what{\beta}_{k}^\2  \rgt\rb_\lt\eps_t 
    \\
     &- \sum_{t\in \Bar{J}_v} \lft\lb X_t, \what{\beta}_k^\1 - \beta^*_{\eta_k} \rgt\rb_\lt\lft\lb X_t, \beta^*_{J_v} - \what{\beta}_{J_v} \rgt\rb_\lt  
    \\
    &- \sum_{t\in \Bar{J}_v} \lft\lb X_t, \beta^*_{\eta_{k+1}} - \what{\beta}_{k}^\2 \rgt\rb_\lt \lft\lb X_t, \beta^*_{J_v} - \what{\beta}_{J_v} \rgt\rb_\lt  \Bigg\}
\end{align*}
and 
\begin{align*}
    B_v = \frac{\sqrt{2}}{\kappa_k\sqrt{q}} \Bigg\{ &\sum_{t\in \wtilde{J}_v} \lft\lb X_t, \beta^*_{\eta_k} - \beta^*_{\eta_{k+1}} \rgt\rb_\lt \lft\lb X_t, \beta^*_{J_v} - \what{\beta}_{J_v}\rgt\rb_\lt 
    \\
    & - \sum_{t\in \Bar{J}_v} \lft\lb X_t, \beta^*_{\eta_k} - \beta^*_{\eta_{k+1}} \rgt\rb_\lt \lft\lb X_t, \beta^*_{J_v} - \what{\beta}_{J_v} \rgt\rb_\lt   \Bigg\}.
\end{align*}
Such an expansion is possible because, under the event outlined in \Cref{lemma: probability bound first lemma}, for $1\le v \le S,\; J_v$ have no change point. This follows from their construction in \Cref{alg: LRV} and conditions specified in \eqref{eq: long run q conditions}. As a consequence, $\beta^*_{\widetilde{J}_v} = \beta^*_{\Bar{J}_v} = \beta^*_{J_v}$.
Following from \Cref{lemma: empirical excess risk}, \Cref{lemma: consistency flr}, \Cref{lemma: empirical excess risk II}, \Cref{lemma: consistency II} and the choice of the tuning parameter $q$ detailed in \eqref{eq: long run q conditions}, we may write
$$
A_v = O_p\lft( \frac{1}{\kappa_k}\sqrt{q} \delta_q \log^{1+\xi}(q) \rgt) = O_p\lft( \frac{1}{\kappa_k} q^{\frac{1/2-r}{2r+1}} \log^{1+\xi}(q) \rgt) = o_p(1).
$$
For the term $B_v$, we may write
\begin{align*}
    B_v = &\frac{\sqrt{2}}{\kappa_k\sqrt{q}} \Bigg\{ \sum_{t\in \wtilde{J}_v} \lft( \lft\lb X_t, \beta^*_{\eta_k} - \beta^*_{\eta_{k+1}} \rgt\rb_\lt \lft\lb X_t, \beta^*_{J_v} - \what{\beta}_{J_v}\rgt\rb_\lt - \Sigma\lft[\beta^*_{\eta_k} - \beta^*_{\eta_{k+1}},  \beta^*_{J_v} - \what{\beta}_{J_v}\rgt] \rgt)
    \\
    &\qquad \qquad - \sum_{t\in \Bar{J}_v} \lft( \lft\lb X_t, \beta^*_{\eta_k} - \beta^*_{\eta_{k+1}} \rgt\rb_\lt \lft\lb X_t, \beta^*_{J_v} - \what{\beta}_{J_v} \rgt\rb_\lt - \Sigma\lft[\beta^*_{\eta_k} - \beta^*_{\eta_{k+1}},  \beta^*_{J_v} - \what{\beta}_{J_v}\rgt] \rgt)  \Bigg\}
    \\
    & = \frac{\sqrt{2q}}{\kappa_k} \Bigg\{ \lft( \what{\Sigma}_{\widetilde{J}_v} - \Sigma \rgt)\lft[ \beta^*_{\eta_k} - \beta^*_{\eta_{k+1}},  \beta^*_{J_v} - \what{\beta}_{J_v} \rgt] - \lft( \what{\Sigma}_{\Bar{J}_v} - \Sigma \rgt)\lft[ \beta^*_{\eta_k} - \beta^*_{\eta_{k+1}},  \beta^*_{J} - \what{\beta}_{J} \rgt]  \Bigg\} 
    \\
    & = O_p\lft( \frac{\sqrt{q}}{\kappa_k} \sqrt{\frac{\delta_q}{\sqrt{q}}} \log^{1+\xi}(q) \rgt) = O_p\lft( \frac{1}{\kappa_k} \sqrt{q^{\frac{2r-1}{2r+1}}} \log^{1+\xi}(q) \rgt) = o_p(1),
\end{align*}
where the first equality in the last line follows from the Holders inequality $\lft( \Sigma\lft[a,b\rgt] \le \sqrt{\Sigma\lft[a,a\rgt]\Sigma\lft[b,b\rgt]}\rgt)$, \eqref{eq: markov I second equation} of \Cref{lemma: general markov type bound I}, \Cref{lemma: Excess Risk} and \Cref{lemma: empirical excess risk}, the last equality follows from \eqref{eq: long run q conditions}. Since $A_v = o_p(1), B_v = o_P(1)$ and $F_{J_v}^* = O_p(1)$, we can write
\begin{equation}\nonumber
    \lft(\what{F}_{J_v}\rgt)^2 - \lft(F_{J_v}^*\rgt)^2 = \lft( A_v + F_{J_v}^* + B_v \rgt)^2 - \lft(F_{J_v}^*\rgt)^2 = o_p(1).
\end{equation}
Therefore, from \eqref{eq: lrv proof3}
\begin{equation}\nonumber
    \frac{1}{S} \sum_{v=1}^S (\what{F}_{J_v})^2 \xrightarrow{\mathbb{P}} \sigma^2_\infty(k).
\end{equation}
The main result now follows from the Slutsky's theorem because $F_{J_v} = \frac{\kappa_k}{\what{\kappa}_k} \what{F}_{J_v}$, and $\frac{\kappa_k}{\what{\kappa}_k} \xrightarrow{\mathbb P} 1$ by \Cref{lemma: consistent kappa}.

\end{proof}

\subsection{Technical results for the Proof of \Cref{thm: long run variance estimation}}
\begin{lemma}\label{lemma: consistent kappa}
    Suppose that assumptions of \Cref{thm: Consistency} holds. The estimator $\what{\kappa}_k$ defined in \eqref{eq: kappa hat} satisfies
    %
    $$
    \what{\kappa}_k^2 - \kappa_k^2 = O_p\lft( \Delta^{-r/(2r+1)} \log^{1+\xi} \Delta \rgt).
    $$
    Consequently,
    $$
     \frac{\what{\kappa}_k^2}{\kappa_k^2}  \xrightarrow{\mathbb{P}} 1 , \qquad n \to \infty.
    $$
\end{lemma}
\
\\
\begin{proof}
    WLOG let $\weta_k \ge \eta_k$. Observe that
    $$
    \beta^*_{(s_k, \weta_k]} - \beta^*_{(\weta_k, e_k]} = \lft(\frac{\eta_k - s_k}{\weta_k - s_k} \rgt) \lft[ \beta^*_{\eta_k} - \beta^*_{\eta_{k+1}} \rgt]
    $$
    and because $\Delta/5 \le \weta - s_k \le \Delta$ from \Cref{lemma: refined interval}, following from \Cref{thm: Consistency} we have that
    $$
    1 - \frac{\weta_k -\eta_k}{\weta_k - s_k} =  \lft(\frac{\eta_k - s_k}{\weta_k - s_k} \rgt)  \xrightarrow{\mathbb{P}} 1, \qquad n \to \infty.
    $$
We may write the expansion 
\begin{align}\nonumber
    \what{\kappa}_k^2 - \kappa_k^2 &= \what{\Sigma}_{(s_k,e_k]} \lft[ \what{\beta}_{(s_k, \weta_k]} - \what{\beta}_{(\weta_k, e_k]}, \what{\beta}_{(s_k, \weta_k]} - \what{\beta}_{(\weta_k, e_k]} \rgt] -\Sigma\lft[ \beta^*_{\eta_k} - \beta^*_{\eta_{k+1}}, \beta^*_{\eta_k} - \beta^*_{\eta_{k+1}} \rgt]
    \\\nonumber
    &= \underbrace{\what{\Sigma}_{(s_k,e_k]} \lft[ \what{\beta}_{(s_k, \weta_k]} - \beta^*_{(s_k, \weta_k]}, \what{\beta}_{(s_k, \weta_k]} - \beta^*_{(s_k, \weta_k]} \rgt]}_{A_1} + \underbrace{\what{\Sigma}_{(s_k,e_k]} \lft[ \beta^*_{(\weta_k, e_k]} - \what{\beta}_{(\weta_k, e_k]}, \beta^*_{(\weta_k, e_k]} - \what{\beta}_{(\weta_k, e_k]} \rgt]}_{A_2}
    \\\nonumber
    &\quad + \underbrace{2\what{\Sigma}_{(s_k,e_k]} \lft[ \beta^*_{(s_k, \weta_k]} - \beta^*_{(\weta_k, e_k]}, \what{\beta}_{(s_k, \weta_k]} - \beta^*_{(s_k, \weta_k]}\rgt]}_{A_3} + \underbrace{2\what{\Sigma}_{(s_k,e_k]} \lft[ \what{\beta}_{(s_k, \weta_k]} - \beta^*_{(s_k, \weta_k]}, \beta^*_{(\weta_k, e_k]} - \what{\beta}_{(\weta_k, e_k]} \rgt]}_{A_4}
    \\\nonumber
    & \quad + \underbrace{2\what{\Sigma}_{(s_k,e_k]} \lft[ \beta^*_{(\weta_k, e_k]} - \what{\beta}_{(\weta_k, e_k]}, \beta^*_{(s_k, \weta_k]} - \beta^*_{(\weta_k, e_k]} \rgt]}_{A_5} + \underbrace{\lft( \what{\Sigma}_{(s_k,e_k]} - \Sigma  \rgt)\lft[ \beta^*_{(s_k, \weta_k]} - \beta^*_{(\weta_k, e_k]}, \beta^*_{(s_k, \weta_k]} - \beta^*_{(\weta_k, e_k]} \rgt]}_{A_6}
    \\\nonumber
    & \qquad + \underbrace{\lft[  \lft(\frac{\eta_k - s_k}{\weta_k - s_k} \rgt)^2 - 1  \rgt] \Sigma\lft[ \beta^*_{\eta_k} - \beta^*_{\eta_{k+1}}, \beta^*_{\eta_k} - \beta^*_{\eta_{k+1}} \rgt]}_{A_7}.
\end{align}
Observing $ \weta_k - s_k = O(\Delta) $, it follows from \Cref{lemma: empirical excess risk II} that for $j=1,2,4 $ we have
$$
\big|A_j\big| = O_p\lft( \delta_{\Delta} \log^{1+\xi} \Delta \rgt).
$$
For the expression $A_3$, it follows from \Cref{lemma: unfold empirical bound II} that
\begin{align*}
    \big|A_3 \big| &= 2\what{\Sigma}_{(s_k,e_k]} \lft[ \beta^*_{(s_k, \weta_k]} - \beta^*_{(\weta_k, e_k]}, \what{\beta}_{(s_k, \weta_k]} - \beta^*_{(s_k, \weta_k]}\rgt] 
    \\
    &\le 2\what{\Sigma}_{(s_k,e_k]} \lft[ \beta^*_{\eta_k} - \beta^*_{\eta_{k+1}}, \what{\beta}_{(s_k, \weta_k]} - \beta^*_{(s_k, \weta_k]}\rgt] = O_p\lft( \sqrt{\delta_{\Delta}} \log^{1+\xi} \Delta\rgt),    
\end{align*}
and for the fifth expression we have $ \big|A_5 \big| = O_p\lft( \sqrt{\delta_{\Delta}} \log^{1+\xi} \Delta\rgt)$ following the same argument.
For the expression $A_6$, we have
\begin{align*}
    \big|A_6 \big| &=  \lft(\frac{\eta_k - s_k}{\weta_k - s_k} \rgt)^2  \lft( \what{\Sigma}_{(s_k,e_k]} - \Sigma  \rgt) \lft[ \beta^*_{\eta_k} - \beta^*_{\eta_{k+1}}, \beta^*_{\eta_k} - \beta^*_{\eta_{k+1}} \rgt] 
    \\
    &\le \lft( \what{\Sigma}_{(s_k,e_k]} - \Sigma  \rgt) \lft[ \beta^*_{\eta_k} - \beta^*_{\eta_{k+1}}, \beta^*_{\eta_k} - \beta^*_{\eta_{k+1}} \rgt] = O_p\lft( \frac{1}{\sqrt{\Delta}} \kappa_k^2\rgt), 
\end{align*}
where the last equality follows from \Cref{lemma: general markov type bound I} and $e_k - s_k = O(\Delta)$. The deviation for the last expression
$$
|A_7| =  \lft( 1 + \frac{\eta_k -s_k}{\weta_k - s_k} \rgt) \lft(\frac{\weta_k - \eta_k}{\weta_k - s_k} \rgt)  \kappa_k^2 \le 6 \frac{5}{\Delta} \kappa_k^2 (\weta_k - \eta_k) = O_p \lft( \delta_{\Delta} \log^{1+\xi} (\Delta) \rgt)
$$ follows from the earlier observation $ \Delta/5 \le \eta_k - s_k \le \weta_k -s_k \le \Delta$ and \eqref{eq: consistency localization}. The first part this current lemma $
    \what{\kappa}_k^2 - \kappa_k^2 = O_p\lft( \Delta^{-r/(2r+1)} \log^{1+\xi} \Delta \rgt)
    $ follows by combining this seven deviation bounds.

\
\\
The deviation from the first part lead us to
$$
 \lft| \frac{\what{\kappa}_k^2 - \kappa_k^2}{\kappa_k^2}\rgt| = O_p\lft( \frac{1}{\kappa^2} \Delta^{-r/(2r+1)} \log^{1+\xi} (\Delta) \rgt) = o_p(1),
$$
where the last equality follows from \Cref{assume: SNR Signal to Noise ratio}.

\end{proof}

\newpage
\section{Proof of \Cref{thm: CI vanishing}}
\begin{proof}
Let $k \in [1, \ldots, \Khat] $ be given. For notational simplicity, we denote $\what{u} = \what{u}^{(b)}_k$ and $z_j = z_j^{(b)}$. 

The proof follows a similar pattern as the proof of \Cref{thm: Limiting distribution}. In the first step, we establish the uniform tightness of the minimizer. In the second step, we demonstrate the convergence of the objective function on a compact domain and use the Argmax continuous mapping theorem.

\
\\
{\bf Step 1.} Let $\what{u}$ be a minimizer. Without loss of generality, assume $\what{u}\ge 0$. 
Since $\what{\sigma}_\infty^2(k) = O_p(1)$, we may write
\begin{align}\nonumber
     \what{u} \le - \what{\sigma}_\infty^2(k) \frac{1}{\sqrt{n}} \sum_{j=1}^{\lfloor n\what{u} \rfloor} z_j = O_p\lft( \sqrt{\what{u} \log^{1+\xi}(\what{u}}) \rgt),
\end{align}
where the stochastic bound follows from the uniform result \Cref{lemma: maximal inequality useful}. Therefore, $\what{u} = O_p(1)$.

\
\\
{\bf Step 2.} Let $M>0$. We have $\what{\sigma}_\infty^2(k) \xrightarrow{\mathbb{P}} \sigma_\infty^2(k)$ from \Cref{thm: long run variance estimation}. From functional CLT, we have
$$
\frac{1}{\sqrt{n}} \sum_{j=1}^{\lfloor nr \rfloor} z_j \xrightarrow{\calD} \mathbb{B}_1(r),
$$
uniformly for all $0 \le r\le M$. Therefore, with the Argmax continuous mapping theorem ( e.g. Theorem 3.2.2 of \citet*{van1996weak}), we have

$$
\what{u} \xrightarrow{\calD} \argmin_{r \in \mathbb R} \big\{ |r| + \sigma_\infty(k) \mathbb W(r) \big\}, \qquad n \to \infty.
$$
The main result now follows from the Slutsky's theorem.
\end{proof}

\newpage
\section{Deviation bounds in functional linear regression}\label{sec: FLR}
\subsection{Notations}\label{sec: notation E1}
For any $a>0$, we denote $\delta_a \asymp a^{-2r/(2r+1)}$. Also, $\lambda_a \asymp a^{-2r/(2r+1)}$. This is used in the observation \eqref{def: fhast} to denote $\fhast$ which is estimator of $f_\st^*$ from \eqref{def: bhast to fhast}. The operator $T$ is defined in \eqref{eq: def T} and its plug-in estimate $T_\I$ is defined in \eqref{eq: sample operator action}. We use ${\bf I}$ to denote the identity operator. The expression for $g_\st$ and $H_\st$ is defined in \Cref{pro: rough estimator} and \Cref{lemma: bias expansion} respectively.
\subsection{Kernel tools}
Following Riesz representation theorem, the norm associated with $\hk$ from \eqref{eq: RKHS defintion} can be equivalently defined through,
$$
\lb f, L_K(g) \rb_\hk := \lb f, g \rb_\lt.
$$
One may note that
\begin{equation}\label{Expt square}
    \E \left[ \lb X, f \rb^2 \right] = \int f(s) \sx(s,t) f(t) \; ds \; dt = \lb L_\sx (f), f\rb_\lt = \Sigma[f, f].
\end{equation}
\\
Moving forward, the main operator of our interest is the linear operator corresponding to the bi-linear function $\kh\sx\kh$ and the eigenvalues and eigenfunctions from its expansion.
\\
The linear operator on $\lt$ corresponding to $K^{1/2}\sx K^{1/2}$ is given by
$$
L_{K^{1/2}\sx K^{1/2}} (f) (*) = \lb K^{1/2}\sx K^{1/2}(\cdot, *), f(\cdot) \rb_\lt.
$$
We denote the linear operator
\begin{equation}\label{eq: def T}
    T = L_{K^{1/2}\sx K^{1/2}},
\end{equation}
and by \Cref{assume: model assumption spectral}
$$
T(\phi_l) = \mathfrak{s}_l \phi_l.
$$
Following this, for any $a\in \mathbb{R}$, the operator $T^a$ is defined through the operation $T^a(\phi_l) = \mathfrak{s}_l^a \phi_l$. Also for any $\beta \in \hk$ such that $f = L_\kmh (\beta)$,
\begin{equation}\label{eq: sigma to T}
    \begin{aligned}
        \Sigma[\beta, \beta] = \Sigma[L_\kh(f), L_\kh(f)] &= \lb L_\Sigma L_\kh (f), L_\kh (f)\rb_\lt  
        \\
        &= \lb L_{\kh\Sigma\kh}(f), f \rb_\lt = \lb T(f), f\rb_\lt = \|T^{1/2} (f)\|^2_\lt.
    \end{aligned}    
\end{equation}

\
\\
The estimator of covariance function based on the sub-sample $\I \subset (0,n]$ is given by

$$
\shati(u,v) = \frac{1}{\I} \sum_{j\in \I} X_j (u) X_j(v).
$$
\
\\
The empirical version of $T$ is $T_\I := L_{\kh\shati\kh} $ and its action can be viewed as
\begin{equation}\label{eq: sample operator action}
    T_\I(h) = L_\kh\circ L_{\shati} \circ L_\kh (h) = L_\kh \left( \frac{1}{|\I|} \sum_{j\in \I} \lb X_j , L_\kh(h) \rb X_j \right) = \frac{1}{|\I|} \sum_{j\in \I} \lb X_j, L_\kh (h) \rb L_\kh(X_j).
\end{equation}


Since, $\lt$ is bijectively mapped to $\hk$, we may have $f^*_\st$ and $\fhast$ defined as
\begin{equation}\label{def: bhast to fhast}
 \frac{1}{t-s}\sum_{j=s+1}^t f_j^*= f^*_\st = L_\kmh \beta^*_\st \qquad \text{and} \qquad\fhast = L_\kmh(\bhast). 
\end{equation}

\
\\
We may also observe that 

\begin{equation}\label{def: fhast}
    \fhast = \argmin_{f\in \lt} \left\{ \frac{1}{(t-s)}\sum_{j = s+1}^t \left( y_i - \lb X_i, L_\kh(f) \rb_\lt \right)^2  + \lambda_\st \|f\|^2_\lt \right\}
\end{equation}
Given $(y^\star, X^\star)$ a copy of $(y,X)$ independent of the training data, the excess risk based on $\st$ is defined as 
\begin{align}\label{def: excess risk st}
\E [ \lb X^\star, \bhast - \beta^*_\st \rb^2 ] &= \int \int (\bhast(x) - \beta^*_\st(x) ) \sx(x,y) (\bhast(y) - \beta^*_\st(y) ) dxdy\\\nonumber
&= \sx[\bhast - \beta^*_\st, \bhast - \beta^*_\st] =  \left\|T^{1/2} (\fhast - f^*_\st)  \right\|^2_\lt   
\end{align}
the last form can be obtained using \eqref{Expt square}, \eqref{eq: sigma to T} and \eqref{def: bhast to fhast}.

\subsection{Roughness regularized estimator and its properties}\label{sec: estimator properties}
In order to evaluate the quality of estimation, we rely on the following lemmas. They help us control various deviation terms in the main result presented in this paper. All the proofs of the lemmas stated below are in the next section.

\begin{lemma}\label{lemma: Excess Risk}
    Let $\xi > 0$. Suppose $(s,e] \subset (0,n]$. Then
    $$
    \max_{s<t\le e} \frac{\delta_{t-s}^{-1}}{\log^{1+\xi} (t-s)} \Sigma\lft[ \bhast - \beta^*_\st, \bhast - \beta^*_\st \rgt] = O_p(1).
    $$
\end{lemma}

\
\\
\begin{lemma}\label{lemma: empirical excess risk}
    Let $\xi >0$. Suppose $(s,e] \subset (0,n]$. Then
    $$
    \max_{s < t\le e }\lft( \frac{\delta_{t-s}^{-1}}{\log^{1+\xi} (t-s)} \rgt) \shast \left[\bhast - \beta^*_\st,\bhast - \beta^*_\st \right] = O_p(1).
    $$
\end{lemma}

\
\\
\begin{lemma}\label{lemma: consistency flr}
    Let $\xi>0$. Suppose $(s,e] \subset (0,n]$. Then
    $$
    \max_{s < t\le e } \lft( \frac{\delta_{t-s}^{-1}}{\log^{1+\xi} (t-s)} \rgt) \frac{1}{t-s} \sum_{j=s+1}^t \lb X_j, \bhast - \beta^*_\st \rb_\lt \eps_j = O_p(1). 
    $$
\end{lemma}

\
\\
\begin{lemma}\label{lemma: unfold empirical bound}
    Let $\xi>0$. Suppose $(s,e] \subset (0,n]$. Then
    $$
        \max_{s<t\le e} \lft( \frac{\delta_{t-s}^{-1}}{\log^{1+\xi} (t-s)} \rgt) \frac{1}{t-s}\sum_{j = s+1}^t \lft\lb X_j, \bhast - \beta^*_\st \rgt\rb_\lt \lft\lb X_j, \beta_j^* - \beta_\st^* \rgt\rb_\lt = O_p(1).
    $$
\end{lemma}

\
\\
\begin{lemma}\label{lemma: empirical excess risk II}
    Let $\xi>0$. Suppose $(s',e']$ and $\se$ are the subsets of $(0,n]$. Then
    $$
    \max_{\substack{s<t\le e \\ s'<t'\le e'}} \; \frac{1}{\mathfrak{J}(t-s,t'-s')} \; \shat_{(s',t']}\lft[\bhast - \beta^*_\st, \bhast - \beta^*_\st \rgt] = O_p(1),
    $$
    where 
    $$
    \mathfrak{J}(t-s,t'-s') = \lft( 1 + \lft( \frac{\delta_{t'-s'}}{\delta_{t-s}}   \rgt)^{1/4} \rgt) \lft\{ \delta_{t-s}^{1/2}  \delta_{t'-s'}^{1/2} \sqrt{\log^{1+\xi}(t'-s' )} \sqrt{\log^{1+\xi}(t-s )} \rgt\} + \delta_{t-s}\log^{1+\xi} (t-s).
    $$
\end{lemma}

\
\\
\begin{lemma}\label{lemma: consistency II}
    Let $\xi>0$. Suppose $(s',e']$ and $\se$ are the subsets of $(0,n]$. Then
    $$
    \max_{\substack{s<t\le e \\ s'<t'\le e'}} \; \frac{1}{\mathfrak{H}(t-s,t'-s')} \;  \frac{1}{t'-s'} \sum_{j=s'+1}^{t'} \lb X_j, \bhast - \beta^*_\st \rb_\lt \eps_j = O_p(1).
    $$
    where 
    $$
    \mathfrak{H}(t-s,t'-s') = \lft( 1 + \lft(\frac{\delta_{t'-s'}}{\delta_{t-s}}   \rgt)^{1/4} \rgt)  \delta_{t-s}^{1/2}  \delta_{t'-s'}^{1/2} \sqrt{\log^{1+\xi}(t'-s' )} \sqrt{\log^{1+\xi}(t-s )}.
    $$
\end{lemma}

\
\\
\begin{lemma}\label{lemma: unfold empirical bound II}
    Let $\xi>0$. Suppose $(s',e']$ and $\se$ are the subsets of $(0,n]$. Then
    $$
    \max_{\substack{s<t\le e \\ s'<t'\le e'}} \; \frac{1}{\mathfrak{G}(t-s,t'-s')} \;  \frac{1}{t'-s'} \sum_{j=s'+1}^{t'} \lb X_j, \bhast - \beta^*_\st \rb_\lt \lb X_j, \beta^*_{\eta_{k+1}} - \beta^*_{\eta_k} \rb_\lt = O_p(1).
    $$
    where 
    $$
    \mathfrak{G}(t-s,t'-s') = \lft( 1 + \lft(\frac{\delta_{t'-s'}}{\delta_{t-s}}   \rgt)^{1/4} \rgt)  \delta_{t-s}^{1/2}  \delta_{t'-s'}^{1/2} \sqrt{\log^{1+\xi}(t'-s' )} \sqrt{\log^{1+\xi}(t-s )} + \kappa_k \sqrt{\delta_{t-s}\log^{1+\xi} (t-s)}.
    $$
\end{lemma}

\newpage
\subsection{Markov type probability bounds}

\begin{lemma}\label{lemma: general markov type bound I}
    Let $\{\mathfrak{f}_j\}_{j=1}^t$ and $h$ be non-random function in $\lt$. Suppose $\Sigma[\mathfrak{f}_j, \mathfrak{f}_j] \le M < \infty$, for all $1\le j \le t$, where $M$ is some absolute constant. Then
    \begin{equation}\label{eq: markov I first equation}
        \E\lft[ \lft| \sum_{j=1}^t \lb X_j, \mathfrak{f}_j\rb_\lt \lb X_j, h \rb_\lt - \Sigma[\mathfrak{f}_j, h] \rgt|^2\rgt] = O(t) \Sigma[h,h].
    \end{equation}
    When $\mathfrak{f}_1 = \ldots = \mathfrak{f}_t = \mathfrak{f}$, Then
    \begin{equation}\label{eq: markov I second equation}
    \E\lft[ \lft| \sum_{j=1}^t \lb X_j, \mathfrak{f}\rb_\lt \lb X_j, h \rb_\lt - \Sigma[\mathfrak{f}, h] \rgt|^2\rgt] = O(t) \Sigma[h,h]\Sigma[\mathfrak{f},\mathfrak{f}].
    \end{equation}
\end{lemma}

\
\\
\begin{proof}
    Given any sequence of  stationary random variables $\{W_j\}_{j=1}^t$ with finite second moment it holds
    \begin{equation}\label{eq: Variance formula}
        \text{Var}\lft( \sum_{j=1}^t W_j \rgt)^2 =  \sum_{j=1}^t Var\lft(  W_j^2 \rgt) + 2\sum_{j=1}^{t-1} (t - j) Cov\lft(  W_1, W_{1+j} \rgt).  
    \end{equation}
    We are going estblish \eqref{eq: markov I first equation}. Let $z_j = \lb X_j, \mathfrak{f}_j\rb_\lt \lb X_j, h \rb_\lt - \Sigma[\mathfrak{f}_j, h] $. Then
    \begin{align}\nonumber
        \E[z_j^2] &\le \E\lft[ \lb X_j, \mathfrak{f}_i\rb_\lt^2 \lb X_j, h \rb_\lt^2 \rgt]
        \\\nonumber
        & \le \sqrt{\E\lft[ \lb X_j, \mathfrak{f}_j\rb_\lt^4 \rgt]} \sqrt{\E\lft[  \lb X_j, h \rb_\lt^4 \rgt]}
        \\\label{eq: f18a}
        &\le c^2\E\lft[ \lb X_j, \mathfrak{f}_j\rb_\lt^2 \rgt] c^2 \E\lft[ \lb X_j, h \rb_\lt^2 \rgt]
        \\\nonumber
        &= c^4 \Sigma[\mathfrak{f}_j,\mathfrak{f}_j] \Sigma[h,h] \le c^4 M \Sigma[h,h].
    \end{align}
The \eqref{eq: f18a} follows from the \Cref{assume: model assumption flr CP}, where we have the sixth moment bounded by the second moment up to a constant factor $c$.

\
\\
We have
\begin{align}\nonumber
    \E[|z_j z_{j+k}|] = Cov(|z_j|, |z_{j+k}|) \le \|z_j\|_3 \|z_{j+k}\|_{3} \alpha^{1/3} (k),
\end{align}
following from \Cref{lemma: mixing covariance inequality}. Following from
$$
\|z_j\|_3 = \|\lb X_j, \mathfrak{f}_j\rb_\lt \lb X_j, h \rb_\lt - \Sigma[\mathfrak{f}_j, h]\|_3 \le \|\lb X_j, \mathfrak{f}_j\rb_\lt \lb X_j, h \rb_\lt \|_3 +  \|\Sigma[\mathfrak{f}_j, h]\|_3 \le 2 \|\lb X_j, \mathfrak{f}_i\rb_\lt \lb X_j, h \rb_\lt \|_3,
$$
one may write
\begin{align}\nonumber
   \|z_j\|_3 &\le 2\|\lb X_j, \mathfrak{f}_j\rb_\lt \lb X_j, h \rb_\lt \|_3
   \\\nonumber
   & \le 2\|\lb X_j, \mathfrak{f}_j\rb_\lt \|_6 \|\lb X_j, h \rb_\lt\|_6
   \\\nonumber
   & \le 2 c\|\lb X_j, \mathfrak{f}_j\rb_\lt \|_2 c\|\lb X_j, h \rb_\lt\|_2 = 2c^2\sqrt{\Sigma[\mathfrak{f}_i,\mathfrak{f}_j] \Sigma[h,h]} \le 2c^2\sqrt{M} \sqrt{\Sigma[h,h]}.
\end{align}
The last line here follows the same argument as \eqref{eq: f18a}. Similarly
$$
\|z_{1+j}\|_3 \le 2c\sqrt{M} \sqrt{\Sigma[h,h]}.
$$
Therefore 
$$
\E[|z_1 z_{1+j}|] \le 4c^4 M \Sigma[h,h] \alpha^{1/3}(j).
$$
Following \eqref{eq: Variance formula}, one may have the expansion 
\begin{align}\nonumber
    \E\left[ \left| \sum_{j=1}^t z_j \right|^2 \right] &= \E\left[ \sum_{j=1}^t z_j^2 \right] + 2 \sum_{j=1}^{t-1} (t-j) \E\left[  z_1 z_{1+j} \right]
    \\\label{eq: f18b}
    & \le \E\left[ \sum_{j=1}^t z_j^2 \right] + 2 \sum_{k=1}^{t-1} (t-k) \E\left[  |z_i z_{i+k}| \right].
\end{align}
Using \eqref{eq: f18b}, we may write 
\begin{align*}
    \E\lft[\lft|\sum_{j=1}^t z_i\rgt|^2\rgt] &\le \E\left[ \sum_{j=1}^t z_j^2 \right] + 2 \sum_{j=1}^{t-1} (t-j) \E\left[  |z_1 z_{1+j}| \right]
    \\
    &\le \sum_{i=1}^t c^4 M \Sigma[h,h] + 2 \sum_{j=1}^{t-1} (t-k) 4c^4 M \Sigma[h,h] \alpha^{1/3}(j)
    \\
    &\le t c^4 M \Sigma[h,h] + 8 c^4 M t \Sigma[h,h]\sum_{j=1}^{t-1} \alpha^{1/3}(j)
    \\
    &\le t c^4 M \Sigma[h,h] + 8 c^4 M t \Sigma[h,h]\sum_{j=1}^{\infty} \alpha^{1/3}(j) = \lft(t \Sigma[h,h] \rgt) O(1).
\end{align*}
The last line follows from $\sum_{j\ge 1} \alpha^{1/3}(j) <\infty$.

\
\\
The proof for \eqref{eq: markov I second equation} is very similar and therefore omitted.
\end{proof}

\
\\
\begin{lemma}\label{lemma: general markov type bound II}
     Let $h$ be non-random function in $\lt$. Then
     $$
     \E\lft[ \lft| \sum_{j=1}^t \lb X_j, h\rb_\lt\eps_j \rgt|^2\rgt] = O\lft( t \Sigma[h,h] \rgt).
     $$
\end{lemma}

\
\\
\begin{proof}
The proof here closely follows the proof of the \Cref{lemma: general markov type bound I}.
    Let $z_j =  \lb X_j, h\rb_\lt\eps_j $. We can see $\E \lft[z_j\rgt] =0$.

\
\\
Observe that
\begin{equation}\label{eq: cite once}
    \E\lft[ z_j^2 \rgt] = \E\lft[ \lb X_j, h\rb_\lt^2 \E \lft[ \eps_j^2| X_i \rgt]\rgt] \le O(1) \|\lb X_j, h\rb_\lt\|_2^2 = O(1) \Sigma[h, h],  
\end{equation}
here we use the moment assumption outlined \Cref{assume: model assumption flr CP}.

\
\\
Following from
$$
\E\lft[ \lft| z_1 z_{1+j} \rgt| \rgt] \le \|z_1\|_3 \|z_{1+j}\|_3 \alpha^{1/3}(j),
$$
and 
$$
\|z_j\|_3 \le  \lft( \E\lft[ \lb X_j, h\rb_\lt^3 \E \lft[ \eps_j^3| X_j \rgt]\rgt] \rgt)^{1/3} \le O(1) \|\lb X_j, h \rb_\lt\|_3  = O(1) \sqrt{\Sigma[h, h]},
$$
we may have
\begin{equation}\label{eq: proof of long run variance}
\E\lft[ \lft| \lb X_j,h\rb \eps_j  \lb X_{j+k},h\rb \eps_{j+k}  \rgt| \rgt] = \E\lft[ \lft| z_j z_{j+k} \rgt| \rgt] = O( \Sigma[h,h]) \alpha^{1/3}(k).   
\end{equation}

The rest of proof follows from the exactly same arguments as the proof of \Cref{lemma: general markov type bound I} and therefore omitted.

\end{proof}

\subsection{Proofs of Lemmas from \Cref{sec: estimator properties}}
All the proofs in this section used the notations from \Cref{sec: notation E1}.
\subsubsection{Proof of \Cref{lemma: Excess Risk}}

The proof of \Cref{lemma: Excess Risk} follows from \Cref{lemma: general excess risk type} with $a=1/2$ and $b=1$.

\subsubsection{Proof of \Cref{lemma: empirical excess risk}}
\begin{proof}
Let $ 0 < \nu < \frac{1}{2} - \frac{1}{4r}$.
    Observe that
    \begin{align*}
        &\lft( \frac{\delta_{t-s}^{-1}}{\log^{1+\xi} (t-s)} \rgt) \shast \left[\bhast - \beta^*_\st,\bhast - \beta^*_\st \right]
        \\
        =&\lft( \frac{\delta_{t-s}^{-1}}{\log^{1+\xi} (t-s)} \rgt) \lft\lb \fhast - f^*_\st, T_\st (\fhast - f^*_\st) \rgt\rb_\lt
        \\
        =& \lft( \frac{\delta_{t-s}^{-1}}{\log^{1+\xi} (t-s)} \rgt) \lft\{   \lft\lb \fhast - f^*_\st, (T_\st - T) (\fhast - f^*_\st) \rgt\rb_\lt +   \lft\lb \fhast - f^*_\st, T (\fhast - f^*_\st) \rgt\rb_\lt\rgt\}
        \\
        \le & \lft( \frac{\delta_{t-s}^{-1}}{\log^{1+\xi} (t-s)} \rgt) \lft\{ \lft\|T^{\nu} ( \fhast - f^*_\st) \rgt\|_\lt \lft\|T^{-\nu} (T_\st - T)( \fhast - f^*_\st) \rgt\|_\lt +  \lft\|T^{1/2} ( \fhast - f^*_\st) \rgt\|_\lt \rgt\}.
    \end{align*}
    The term on the right is bounded by using \Cref{lemma: mu Excess risk} and \Cref{lemma: general emp/consistency type}. The term on the left is bounded by using \Cref{lemma: Excess Risk}.
\end{proof}

\
\\
\subsubsection{Proof of \Cref{lemma: consistency flr}}
\begin{proof}
    Observe that

\begin{align}\nonumber
        \frac{1}{(t-s)} \sum_{j=s+1}^t \lb X_j, \bhast - \beta^*_\st \rb_\lt \eps_j &= \lft\lb \frac{1}{(t-s)} \sum_{j=s+1}^t X_j \eps_j,  \bhast - \beta^*_\st   \rgt\rb_\lt 
        \\\nonumber
        &= \lft\lb \frac{1}{(t-s)} \sum_{j=s+1}^t L_\kh(X_j) \eps_j,  \fhast - f^*_\st   \rgt\rb_\lt    
        \\\nonumber
        &= \lft\lb g_\st, \fhast - f^*_\st \rgt\rb_\lt
        \\\nonumber
        \nonumber
        &= \lft\lb T^{-1/4} \lft( T + \lambda_{t-s} \rgt)^{-1/4} g_\st, T^{1/4} \lft( T + \lambda_{t-s} \rgt)^{1/4} \fhast - f^*_\st \rgt\rb_\lt
        \\\nonumber
        &\le \lft\| T^{-1/4} \lft( T + \lambda_{t-s} \rgt)^{-1/4} g_\st \rgt\|_\lt \lft\| T^{1/4} \lft( T + \lambda_{t-s} \rgt)^{1/4} \fhast - f^*_\st \rgt\|_\lt
\end{align}
where the last line follows from Cauchy-Schwarz inequality.

\
\\
From \Cref{lemma: gt bound}, we have
$$
     \max_{s < t\le e } \sqrt{ \frac{\delta_{t-s}^{-1}}{\log^{1+\xi}(t-s)}}\lft\| T^{-1/4} \lft( T + \lambda_{t-s} \rgt)^{-1/4} g_\st \rgt\|_\lt = O_p(1), 
$$
and from \Cref{lemma: general excess risk type}
$$
     \max_{s < t\le e } \sqrt{ \frac{\delta_{t-s}^{-1}}{\log^{1+\xi}(t-s)}}\lft\| T^{1/4} \lft( T + \lambda_{t-s} \rgt)^{1/4} \fhast - f^*_\st \rgt\|_\lt = O_p(1).
$$
The above two bounds establish the result.
\end{proof}

\subsubsection{Proof of \Cref{lemma: unfold empirical bound}}
\begin{proof}
    Observe that
    $$
    \frac{1}{t-s}\sum_{j=s+1}^t\Sigma\lft[\bhast - \beta^*_\st, \beta_j^* - \beta_\st^* \rgt] = 0,
    $$
because $\beta^*_\st = \sum_{j=s+1}^t \beta^*_j/(t-s)$.

\
\\
We may write
\begin{align*}
    &\frac{1}{t-s}\sum_{j = s+1}^t \lft\lb X_j, \bhast - \beta^*_\st \rgt\rb_\lt \lft\lb X_j, \beta_j^* - \beta_\st^* \rgt\rb_\lt
    \\
    =& \frac{1}{t-s}\sum_{j = s+1}^t \lft( \lft\lb X_j, \bhast - \beta^*_\st \rgt\rb_\lt \lft\lb X_j, \beta_j^* - \beta_\st^* \rgt\rb_\lt - \Sigma\lft[\bhast - \beta^*_\st, \beta_j^* - \beta_\st^* \rgt] \rgt)
    \\
    =& \frac{1}{t-s}\sum_{j = s+1}^t \lft( \lft\lb X_j, \bhast - \beta^*_\st \rgt\rb_\lt \lft\lb X_j, \beta_j^* - \beta_\st^* \rgt\rb_\lt - \lft\lb\fhast - f^*_\st, T(f_j^* - f_\st^*)\rgt\rb_\lt \rgt)
    \\
    =& \lft\lb \frac{1}{t-s}\sum_{j = s+1}^t \lft(\lft\lb L_\kh(X_j), f_j^* - f_\st^* \rgt\rb_\lt L_\kh(X_j) - T(f_j^* - f_\st^*) \rgt) , \fhast - f^*_\st \rgt\rb
    \\
    =&\lft\lb G_\st , \fhast - f^*_\st \rgt\rb_\lt
    \\
    =& \lft\lb T^{-1/4} \lft( T + \lambda_{t-s} \boldI \rgt)^{-1/4} G_\st , T^{1/4} \lft( T + \lambda_{t-s} \boldI \rgt)^{1/4}  (\fhast - f^*_\st) \rgt\rb_\lt
    \\
    \le& \lft\| T^{-1/4} \lft( T + \lambda_{t-s} \boldI \rgt)^{-1/4} G_\st \rgt\|_\lt \lft\| T^{1/4} \lft( T + \lambda_{t-s} \boldI \rgt)^{1/4}  (\fhast - f^*_\st) \rgt\|_\lt,
\end{align*}
    where $G_\st = \frac{1}{t-s}\sum_{j = s+1}^t \lft(\lft\lb L_\kh(X_j), f_j^* - f_\st^* \rgt\rb_\lt L_\kh(X_j) - T(f_j^* - f_\st^*) \rgt) $.

\
\\
From \Cref{lemma: Ht bound}, we have
$$
     \max_{s < t\le e }  \sqrt{\frac{\delta_{t-s}^{-1}}{\log^{1+\xi}(t-s)}}\lft\| T^{-1/4} \lft( T + \lambda_{t-s} \boldI \rgt)^{-1/4} G_\st \rgt\|_\lt = O_p(1), 
$$
and from \Cref{lemma: general excess risk type}
$$
     \max_{s < t\le e }  \sqrt{\frac{\delta_{t-s}^{-1}}{\log^{1+\xi}(t-s)}}\lft\| T^{1/4} \lft( T + \lambda_{t-s} \boldI \rgt)^{1/4}  (\fhast - f^*_\st) \rgt\|_\lt = O_p(1).
$$
The above two bounds establish the result.
\end{proof}

\subsubsection{Proof of \Cref{lemma: empirical excess risk II}}
\begin{proof}
Let $\nu < 1/2 - 1/4r$. We may write
    \begin{align}\nonumber
        & \shat_{(s',t']}\lft[\bhast - \beta^*_\st, \bhast - \beta^*_\st \rgt]
        \\\nonumber
        \le &  \lft| \lft( \what{\Sigma}_{(s', t']} - \Sigma \rgt) \lft[ \bhast - \beta^*_\st, \bhast - \beta^*_\st \rgt] \rgt| + \lft| \Sigma  \lft[ \bhast - \beta^*_\st, \bhast - \beta^*_\st \rgt] \rgt| 
        \\\nonumber
        =& \lft| \lft\lb (T_{(s', t']} - T) (\fhast - f^*_\st), \fhast - f^*_\st \rgt\rb_\lt \rgt|  + \lft| \Sigma  \lft[ \bhast - \beta^*_\st, \bhast - \beta^*_\st \rgt] \rgt| 
        \\\nonumber
        =&\lft| \lft\lb T^{-\nu} (T_{(s', t']} - T) (\fhast - f^*_\st), T^{\nu}(\fhast - f^*_\st) \rgt\rb_\lt \rgt| + \lft| \Sigma  \lft[ \bhast - \beta^*_\st, \bhast - \beta^*_\st \rgt] \rgt|
        \\\label{eq: t111}
        \le& \lft\| T^{-\nu}(T_{(s', t']} - T) (\fhast - f^*_\st) \rgt\|_\lt \lft\| T^{\nu} (\fhast - f^*_\st) \rgt\|_\lt + \lft| \Sigma  \lft[ \bhast - \beta^*_\st, \bhast - \beta^*_\st \rgt] \rgt|,
    \end{align}
where the second line follows from the triangle inequality and the last line follows from the Cauchy-Schwarz inequality. Observe that
    \begin{align}\nonumber
        &\lft\| T^{-\nu}(T_{(s', t']} - T) (\fhast - f^*_\st) \rgt\|_\lt
        \\\nonumber
        \le & \lft(1 + \lft( \frac{\delta_{t'-s'}}{\delta_{t-s}}   \rgt)^{1/4}\rgt) \Bigg[\lft\| T^{-\nu}(T_{(s', t']} - T) (T + \lambda_{t'-s'} \boldI)^{-1/4} T^{-1/4}\rgt\|_{\mathrm{op}} \cdot \lft\| T^{1/4} (T + \lambda_{t-s} \boldI)^{1/4} (\fhast - f^*_\st) \rgt\|_\lt \Bigg]
        \\\label{eq: t111a}
        = & O_p\lft(  \lft[1 + \lft( \frac{\delta_{t'-s'}}{\delta_{t-s}}   \rgt)^{1/4}\rgt] \delta_{t'-s'}^{1/2} \sqrt{\log^{1+\xi} (t' -s')} \cdot  \delta_{t-s}^{1/2} \sqrt{\log^{1+\xi} (t - s)}  \rgt),
    \end{align}
    where the first line follows from \Cref{lemma:lambda transfer useful} and the last line follows from \Cref{lemma: operator bound} and \Cref{lemma: general excess risk type}. Following from \Cref{lemma: general excess risk type}, we have 
    \begin{equation}\label{eq: t111b}
            \lft\| T^{-\nu} (\fhast - f^*_\st) \rgt\|_\lt = O_p\lft( \delta_{t-s}^\nu \sqrt{\log^{1+\xi} (t-s)} \rgt) = O_p(1),
    \end{equation}
    and from \Cref{lemma: Excess Risk} we have
    \begin{equation}\label{eq: t111c}
        \lft| \Sigma  \lft[ \bhast - \beta^*_\st, \bhast - \beta^*_\st \rgt] \rgt| = O_p\lft( \delta_{t-s} \log^{1+\xi} (t-s)  \rgt).
    \end{equation}
    The result now follows using \eqref{eq: t111a}, \eqref{eq: t111b} and \eqref{eq: t111c} to bound \eqref{eq: t111}.

\end{proof}
    
\subsubsection{Proof of \Cref{lemma: consistency II}}
\begin{proof}
    Observe that
     \begin{align}\nonumber
        &\lft| \frac{1}{t'-s'} \sum_{j = s' + 1}^{ t'}   \lft\lb X_j, \bhast - \beta^*_\st \rgt\rb_\lt \eps_j\rgt|
        \\\nonumber
        = & \lft|  \lft\lb \frac{1}{t' - s'} \sum_{j = s' + 1}^{ t'}  L_\kh(X_j) \eps_j, \fhast - f^*_\st \rgt\rb_\lt\rgt|
        \\\nonumber
        = & \lft| \lft\lb g_{(s', t']}, \fhast - f^*_\st \rgt\rb_\lt\rgt|
        \\\nonumber
        \le & \lft(1 + \lft( \frac{\delta_{t'-s'}}{\delta_{t-s}}   \rgt)^{1/4}\rgt) \lft\|T^{-1/4}  (T + \lambda_{t'-s'} \boldI)^{-1/4}  g_{(s', t']}\rgt\|_\lt \lft\| T^{1/4} (T + \lambda_{t-s} \boldI)^{1/4} (\fhast - f^*_\st) \rgt\|_\lt
        \\\nonumber
        = & O_p\lft(  \lft[ 1 + \lft( \frac{\delta_{t'-s'}}{\delta_{t-s}}   \rgt)^{1/4} \rgt] \delta_{t'-s'}^{1/2} \sqrt{\log^ {1+\xi} (t' -s')} \cdot  \delta_{t-s}^{1/2} \sqrt{\log^{1+\xi} (t - s)}  \rgt),
    \end{align}
    where the second last line follows from \Cref{lemma:lambda transfer useful} and the last line follows from \Cref{lemma: gt bound} and \Cref{lemma: Excess Risk}.
    
\end{proof}

\subsubsection{Proof of \Cref{lemma: unfold empirical bound II}}
\begin{proof}
    Observe that
    \begin{align}\nonumber
        & \lft| \frac{1}{t'-s'} \sum_{j = s' + 1}^{ t'}   \lft\lb X_j, \bhast - \beta^*_\st \rgt\rb_\lt \lft\lb X_j, \beta^*_{\eta_{k+1}} - \beta^*_{\eta_k} \rgt\rb_\lt \rgt| 
        \\\nonumber
        \le &  \lft| \lft( \what{\Sigma}_{(s', t']} - \Sigma \rgt) \lft[ \bhast - \beta^*_\st, \beta^*_{\eta_{k+1}} - \beta^*_{\eta_k} \rgt] \rgt| + \lft| \Sigma  \lft[ \bhast - \beta^*_\st, \beta^*_{\eta_{k+1}} - \beta^*_{\eta_k} \rgt] \rgt| 
        \\\nonumber
        =& \lft| \lft\lb (T_{(s', t']} - T) (\fhast - f^*_\st), f^*_{\eta_{k+1}} - f^*_{\eta_{k}} \rgt\rb_\lt \rgt|  + \lft| \Sigma  \lft[ \bhast - \beta^*_\st, \bhast - \beta^*_\st \rgt] \rgt| 
        \\\nonumber
        =&\lft| \lft\lb  (T_{(s', t']} - T) (\fhast - f^*_\st), (f^*_{\eta_{k+1}} - f^*_{\eta_{k}}) \rgt\rb_\lt \rgt| + \lft| \Sigma  \lft[ \bhast - \beta^*_\st, \beta^*_{\eta_{k+1}} - \beta^*_{\eta_k} \rgt] \rgt|
        \\\label{eq: t113}
        \le& \lft\|(T_{(s', t']} - T) (\fhast - f^*_\st) \rgt\|_\lt \lft\|  f^*_{\eta_{k+1}} - f^*_{\eta_{k}} \rgt\|_\lt + \lft| \Sigma  \lft[ \bhast - \beta^*_\st, \beta^*_{\eta_{k+1}} - \beta^*_{\eta_k} \rgt] \rgt|,
    \end{align}
where the second line follows from the triangle inequality and the last line follows from the Cauchy-Schwarz inequality. Observe that

    \begin{align}\nonumber
        &\lft\| (T_{(s', t']} - T) (\fhast - f^*_\st) \rgt\|_\lt
        \\\nonumber
        \le & \lft(1 + \lft( \frac{\delta_{t'-s'}}{\delta_{t-s}}   \rgt)^{1/4}\rgt) \Bigg[\lft\| (T_{(s', t']} - T) (T + \lambda_{t'-s'} \boldI)^{-1/4} T^{-1/4}\rgt\|_{\mathrm{op}} \cdot  \lft\| T^{1/4} (T + \lambda_{t-s} \boldI)^{1/4} (\fhast - f^*_\st) \rgt\|_\lt \Bigg]
        \\\label{eq: t113a}
        = & O_p\lft(  \lft[ 1 + \lft( \frac{\delta_{t'-s'}}{\delta_{t-s}}   \rgt)^{1/4} \rgt] \delta_{t'-s'}^{1/2} \sqrt{\log^{1+\xi} (t' -s')} \cdot  \delta_{t-s}^{1/2} \sqrt{\log^{1+\xi} (t - s)}  \rgt),
    \end{align}
    where the first line follows from \Cref{lemma:lambda transfer useful} and the last line follows from \Cref{lemma: operator bound} and \Cref{lemma: general excess risk type}. Following from \Cref{lemma: Excess Risk} we have that
    \begin{align}\nonumber
        &\lft| \Sigma  \lft[ \bhast - \beta^*_\st, \bhast - \beta^*_\st \rgt] \rgt| 
        \\\nonumber
        \le &\sqrt{ \Sigma  \lft[ \bhast - \beta^*_\st, \bhast - \beta^*_\st \rgt] } \sqrt{ \Sigma  \lft[ \beta^*_{\eta_{k+1}} - \beta^*_{\eta_k}, \beta^*_{\eta_{k+1}} - \beta^*_{\eta_k} \rgt]}
        \\\label{eq: t113b}
        =&  O_p\lft( \kappa_k \sqrt{\delta_{t-s} \log^{1+\xi} (t-s)}  \rgt).
    \end{align}
The result now follows using \eqref{eq: t113a} and \eqref{eq: t113b}  to bound \eqref{eq: t113}.
\end{proof}

\newpage
\subsubsection{Technical results for this section}
\begin{proposition}\label{pro: rough estimator}
    The analytical expression for the estimator in \eqref{def: fhast} is given by
    $$
    \fhast = (T_\st + \lambda_{t-s} {\bf I} )^{-1} \left( \frac{1}{t-s} \sum_{j=s+1}^t \lb L_\kh(X_j), f_j^* \rb_\lt L_\kh(X_j)  + g_\st \right)
    $$
    where $g_\st = \frac{1}{(t-s)} \sum_{j=s+1}^t \eps_j L_\kh(X_j) $.
\end{proposition}

\
\\
\begin{proof} 

Observe that
$$
\frac{\dau \lb f, f \rb_\lt}{\dau f} = 2f \qquad \text{ and } \qquad \frac{\dau \lb f, g \rb_\lt}{\dau f} = g.
$$

\
\\
Since the objective function is a quadratic form, we just need to differentiate and make it zero to find the minima. We may have
\begin{align*}
0 = &\frac{\dau}{\dau f} \left(\frac{1}{t-s}\sum_{j = s+1}^t \left( y_j - \lb X_j, L_\kh(f) \rb_\lt \right)^2 + \lambda\|f\|^2_\lt \right) \Bigg|_{f = \fhast}
\\
= &\frac{\dau}{\dau f} \left(\frac{1}{t-s}\sum_{j = s+1}^t\left(  \lb L_\kh( X_j), f \rb^2_\lt + y_j^2 - 2y_j\lb L_\kh( X_j), f \rb_\lt \right) + \lambda\|f\|^2_\lt \right) \Bigg|_{f = \fhast}
\\
= & \frac{1}{t-s}\sum_{j = s+1}^t \left(  2\lb L_\kh( X_j), \fhast \rb_\lt L_\kh( X_j) - 2y_j L_\kh( X_j) \right) + 2\lambda \fhast.
\end{align*}

And it lead us to 
\begin{align*}
    0 = & \frac{1}{t-s}\sum_{j = s+1}^t \left(  \lb L_\kh( X_j), \fhast \rb_\lt L_\kh( X_j) - y_j L_\kh( X_j) \right) + \lambda \fhast ,
    \\
     = & \frac{1}{t-s}\sum_{j = s+1}^t \lb L_\kh( X_j), \fhast \rb_\lt L_\kh( X_j) - \frac{1}{t-s} \sum_{j = s+1}^t \lb L_\kh(X_j), f^*_j \rb_\lt L_\kh( X_j)
    \\
    & \qquad- \frac{1}{t-s}\sum_{j = s+1}^t  \eps_j L_\kh( X_j) + \lambda \fhast
    \\
    = & T_\st\fhast  - \frac{1}{t-s} \sum_{j = s+1}^t \lb L_\kh(X_j), f^*_j \rb_\lt L_\kh( X_j) - g_\st + \lambda \fhast.
\end{align*}
The last equation follows from the action of $T_\I$ illustrated at \eqref{eq: sample operator action} in the previous subsection and the result follows.

\end{proof}

\
\\
One key component is the expansion of variance term in the error bound. The next lemma is to structure this variance term. Define
\begin{equation}\label{eq: flst}
    \flst = \left(T + \lambda_{t-s} {\bf I} \right)^{-1} T f^*_\st. 
\end{equation}

\
\\
\begin{lemma}\label{lemma: bias expansion}
    Given \eqref{eq: flst} and the form of the estimator in \Cref{pro: rough estimator}, the following holds
    $$
     \lft(\fhast - \flst \rgt)  = \lft( T + \lambda_{t-s} \boldI \rgt)^{-1} \bigg( (T - T_\st) \lft( \fhast - f^*_\st \rgt) + g_\st + (T - T_\st) f^*_\st + H_\st \bigg)
    $$
    where $H_\st = (t -s)^{-1} \sum_{j=s+1}^t \lft( \lb L_\kh(X_j), f_j^* \rb_\lt L_\kh(X_j) - T f^*_j \rgt)$ and $g_\st$ defined in \Cref{pro: rough estimator}.
\end{lemma}

\
\\
\begin{proof}

    \begin{align}\nonumber
        &\fhast - \flst 
        \\\nonumber
        =& \lft( T + \lambda_{t-s} \boldI \rgt)^{-1} \lft( (T - T_\st) \fhast  + (T_\st + \lambda_{t-s} \boldI ) \fhast - (T + \lambda_{t-s} \boldI)\flst \rgt)
        \\\nonumber
        =& \lft( T + \lambda_{t-s} \boldI \rgt)^{-1} \bigg( (T - T_\st) \fhast + g_\st + \frac{1}{t-s} \sum_{j=s+1}^t \lb L_\kh(X_j), f_j^* \rb_\lt L_\kh(X_j)  - T f^*_\st \bigg)
        \\\nonumber
        =& \lft( T + \lambda_{t-s} \boldI \rgt)^{-1} \bigg( (T - T_\st) \lft( \fhast - f^*_\st \rgt) + g_\st +\frac{1}{t-s} \sum_{j=s+1}^t \lb L_\kh(X_j), f_j^* \rb_\lt L_\kh(X_j) - T_\st f^*_\st \bigg)
        \\\nonumber
        =& \lft( T + \lambda_{t-s} \boldI \rgt)^{-1} \bigg( (T - T_\st) \lft( \fhast - f^*_\st \rgt) + g_\st + (T - T_\st) f^*_\st
        \\\nonumber
        & \qquad \qquad \qquad \qquad+\frac{1}{t-s} \sum_{j=s+1}^t \lb L_\kh(X_j), f_j^* \rb_\lt L_\kh(X_j) - \frac{1}{t-s} \sum_{j=s+1}^t T f^*_j  \bigg)
    \end{align}
    In the last line, we use the fact that $f^*_\st = \sum_{j=s+1}^t f^*_j/(t-s)$ and linearity of the operator $T$.
    The changes at the third last line follows from \Cref{pro: rough estimator} and \eqref{eq: flst}.
\end{proof}

\
\\
\begin{lemma}\label{lemma: det bound}
    Let $\xi > 0$ and $1 \ge b > a  > 0$.
    $$
    \max_{s<t\le e} \delta_{t-s}^{2(b-a-1)} \left\| T^a (T + \lambda_{t-s} {\bf I})^{1-b} (\flst - f^*_\st)  \right\|^2_\lt = O(1).
    $$
\end{lemma}

\
\\
\begin{proof}
Note that because $f_j$ is bounded, the population average $f^*_\st$ is also bounded. Precisely, if $f^*_\st = \sum_{l\ge1} a_l^{s,t} \phi_l $, then $\sum_{l\ge 1} (a_l^{s,t})^2 \le M < \infty$, for some absolute constant $M>0$, where $\{\phi_l\}_{l\ge1}$ is the $\lt$ basis coming from the spectral decomposition of $\kh\Sigma\kh$.

    \begin{align}\nonumber
  \| T^a (T + \lambda_{t-s} {\bf I})^{1-b}(\flst - f^*_\st) \|^2_\lt &= \| T^a (T + \lambda_{t-s} {\bf I})^{1-b} \left( ( T + \lambda_{t-s} {\bf I})^{-1}Tf^*_\st - f^*_\st \right) \|^2_\lt \\\nonumber
  &= \sum_{l\ge1} \mathfrak{s}_l^{2a} (\mathfrak{s}_l + \lambda_{t-s})^{2-2b}\left( \frac{\mathfrak{s}_l}{\mathfrak{s}_l + \lambda_{t-s}} - 1 \right)^2 (a_l^{s,t})^2
  \\\nonumber
  &= \sum_{l\ge1} \mathfrak{s}_l^{2a} \frac{\lambda_{t-s}^2}{(\mathfrak{s}_l + \lambda_{t-s})^{2b}}  (a_l^{s,t})^2
  \\\label{eq: mis det 1}
  &\le \left\{\max_{l\ge1} \mathfrak{s}_l^{2a}\frac{\lambda_{t-s}^2}{(\mathfrak{s}_l + \lambda_{t-s})^{2b}} \right\} \sum_{l\ge1} (a_l^{s,t})^2
 \\\label{eq: mis det 2}
  &\le \left\{\max_{l\ge1} \mathfrak{s}_l^{2a} \frac{\lambda_{t-s}^2}{ \lft[ (1-a/b)^{-(1-a/b)}\lambda_{t-s}^{1-a/b} (a/b)^{-a/b} \mathfrak{s}_l^{a/b} \rgt]^{2b} } \right\} \sum_{l\ge1} (a_l^{s,t})^2
  \\\nonumber
  &= (1-a/b)^{2(1-a/b)} (a/b)^{2a/b} \lambda_{t-s}^{2(1+a-b)} \sum_{l\ge1} (a_l^{s,t})^2
  \\\nonumber
  &= O(1) \lambda_{t-s}^{2(1+a-b)} M =  O\lft(\lambda_{t-s}^{2(1+a-b)}\rgt) = O\lft(\delta_{t-s}^{2(1+a-b)}\rgt),
\end{align}
where the inequality in \eqref{eq: mis det 1} is from Holder's inequality and \eqref{eq: mis det 2} follows from Young's inequality in the following form
$$
a + b \le (pa)^{1/p} (qb)^{1/q},
$$
where $a, b, p, q$ are positive real numbers and $p^{-1} + q^{-1} = 1$.
\end{proof}



\
\\
\begin{lemma}\label{lemma: mu Excess risk}
    Let $\xi >0$. Let $ 0 < \nu < \frac{1}{2} - \frac{1}{4r}$. Then
    $$
    \max_{s<t\le e} \lft\|T^\nu \lft( \fhast - f^*_\st \rgt) \rgt\| \frac{\delta_{t-s}^{-\nu}}{\sqrt{\log^{1+\xi} (t-s)}} = O_p(1).
    $$
\end{lemma}

\
\\

\begin{proof}
    Following from triangle inequality and the decomposition at \Cref{lemma: bias expansion} we have,
\begin{align} \nonumber
    \|T^{\nu} (\fhast - f^*_\st) \|_\lt \le & \|T^{\nu} (\flst - f^*_\st) \|_\lt + \|T^{\nu} (\fhast - \flst) \|_\lt
    \\\label{eq: f10aaa}
    \le & \quad  \|T^{\nu} (\flst - f^*_\st) \|_\lt
    \\\label{eq: f10a}
    & +  \|T^{\nu}(T + \lambda_{t-s} {\bf I})^{-1} (T - T_\st)T^{-\nu}\|_{\mathrm{op}} \|T^\nu (\fhast - \flst)\|_\lt
    \\\label{eq: f10aa}
    & + \|T^{\nu}(T + \lambda_{t-s} {\bf I})^{-1} (T - T_\st)T^{-\nu}\|_{\mathrm{op}} \|T^\nu (\flst - f^*_\st)\|_\lt
    \\\label{eq: f10b}
    &+ \|T^{\nu}(T + \lambda_{t-s} {\bf I})^{-1}g_\st\|_\lt 
    \\\label{eq: f10c}
    &+ \|T^{\nu}(T + \lambda_{t-s} {\bf I})^{-1}(T_\st -T)\|_{\mathrm{op}} \|f^*_\st\|_\lt
    \\\label{eq: f10d}
    &+ \|T^{\nu}(T + \lambda_{t-s} {\bf I})^{-1} H_\st\|_\lt 
\end{align}
We are going to bound each of the four term uniformly to have result.

\
\\
For \eqref{eq: f10aaa}, from \Cref{lemma: det bound}, we write that with high probability
$$
\forall t \in \se, \qquad  \|T^{\nu} (\flst - f^*_\st) \|_\lt \lesssim \delta_{t-s}^\nu.
$$

\
\\
For \eqref{eq: f10a}, from \Cref{lemma: operator bound}, we write that in high probability
$$
\forall t \in \se, \qquad \|T^{\nu}(T + \lambda_{t-s} {\bf I})^{-1} (T - T_\st)T^{-\nu}\|_{\mathrm{op}} \lesssim \delta_{t-s}^\nu \sqrt{\log^{1+\xi} (t-s)},
$$
which would give us 
$$
\eqref{eq: f10a} \le o(1) \|T^\nu (\fhast - \flst)\|_\lt
$$
in probability, in uniform sense.

\
\\
Similarly for \eqref{eq: f10aa}, from \Cref{lemma: operator bound} and \Cref{lemma: det bound}, we write that with high probability
$$
\forall t \in \se, \qquad \eqref{eq: f10aa}\lesssim \delta_{t-s}^{2\nu} \sqrt{\log^{1+\xi} (t-s)}.
$$

\
\\
For \eqref{eq: f10b}, from \Cref{lemma: gt bound}, we write that with high probability
$$
\forall t \in \se, \qquad \|T^{\nu}(T + \lambda_{t-s} {\bf I})^{-1}g_\st\|_\lt \lesssim \delta_{t-s}^\nu \sqrt{\log^{1+\xi} (t-s)}.
$$

\
\\
For \eqref{eq: f10c}, from \Cref{lemma: operator bound}, we write that with high probability
$$
\forall t \in \se, \qquad \|T^{\nu}(T + \lambda_{t-s} {\bf I})^{-1}(T_\st -T)\|_{\mathrm{op}} \|f^*_\st\|_\lt \lesssim \|T^{\nu}(T + \lambda_{t-s} {\bf I})^{-1}(T_\st -T)\|_{\mathrm{op}} \lesssim \delta_{t-s}^\nu \sqrt{\log^{1+\xi} (t-s)},
$$
here we used the fact that $\|f^*_\st\|_\lt < \infty.$

\
\\
For \eqref{eq: f10d}, from \Cref{lemma: Ht bound}, we write that with high probability
$$
\forall t \in \se, \qquad \|T^{\nu}(T + \lambda_{t-s} {\bf I})^{-1} H_\st\|_\lt \lesssim \delta_{t-s}^\nu \sqrt{\log^{1+\xi} (t-s)}.
$$

\
\\
This six individual bounds come together to give us the required result.
\end{proof}

\begin{lemma}\label{lemma: general excess risk type}
Let $\xi > 0$ and $1 \ge b \ge a + 1/2 > 0$.
$$
\max_{s<t\le e} \frac{\delta_{t-s}^{(b-a-1)}}{\sqrt{\log^{1+\xi}(t-s)}} \left\| T^a (T + \lambda_{t-s} {\bf I})^{1-b} (\fhast - f^*_\st)  \right\|_\lt = O_p(1).
$$  
\end{lemma}

\begin{proof}
    Using triangle inequality we may write,
    \begin{equation}\label{eq: f8a}
        \left\| T^a (T + \lambda_{t-s} {\bf I})^{1-b} (\fhast - f^*_\st)  \right\|_\lt \le \left\| T^a (T + \lambda_{t-s} {\bf I})^{1-b} (\fhast - \flst)  \right\|_\lt + \left\| T^a (T + \lambda_{t-s} {\bf I})^{1-b} (\flst - f^*_\st)  \right\|_\lt
    \end{equation}
    where $\flst$ is defined at \eqref{eq: flst}.
The second term on the right of \eqref{eq: f8a} can be bounded using \Cref{lemma: det bound}, which gives us
\begin{align*}
   &\max_{s<t\le e} \frac{\delta_{t-s}^{(b-a-1)}}{\sqrt{\log^{1+\xi}(t-s)}} \left\| T^a (T + \lambda_{t-s} {\bf I})^{1-b} (\flst - f^*_\st)  \right\|_\lt
   \\
   \le & \max_{s<t\le e}  \delta_{t-s}^{(b-a-1)}\left\| T^a (T + \lambda_{t-s} {\bf I})^{1-b} (\flst - f^*_\st)  \right\|_\lt = O(1).
\end{align*}
Now, it is suffice is to bound the first term on the right of \eqref{eq: f8a}. Let $0 < \nu< \frac{1}{2} - \frac{1}{4r}$. Following the decomposition at \Cref{lemma: bias expansion}, we may write

\begin{align}\label{eq: f8b}
    \|T^a (T + \lambda_{t-s} {\bf I})^{1-b} (\fhast - \flst) \|_\lt \le& \quad \|T^{a}(T + \lambda_{t-s} {\bf I})^{-b} (T - T_\st)T^{-\nu}\|_{\mathrm{op}} \|T^\nu (\fhast - f^*_\st)\|_\lt
    \\\label{eq: f8c}
    &+ \|T^{a}(T + \lambda_{t-s} {\bf I})^{-b} g_\st\|_\lt 
    \\\label{eq: f8d}
    &+ \|T^{a}(T + \lambda_{t-s} {\bf I})^{-b}(T_\st -T)\|_{\mathrm{op}} \|f^*_\st\|_\lt
    \\\label{eq: f8e}
    &+ \|T^{a}(T + \lambda_{t-s} {\bf I})^{-b} H_\st\|_\lt 
\end{align}
We are going to bound each of the four terms \eqref{eq: f8b}, \eqref{eq: f8c}, \eqref{eq: f8d}  and \eqref{eq: f8e} uniformly over $s < t \le e$ to have the required result.

\
\\
For \eqref{eq: f8b}, using \Cref{lemma: operator bound} and \Cref{lemma: mu Excess risk}, we write that with high probability
\begin{align*}
    \forall t \in \se, \qquad &\|T^{a}(T + \lambda_{t-s} {\bf I})^{-b} (T - T_\st)T^{-\nu}\|_{\mathrm{op}} \|T^\nu (\fhast - f^*_\st)\|_\lt
    \\
    \lesssim & \delta_{t-s}^{1+a-b} \sqrt{\log^{1+\xi} (t-s)} \delta_{t-s}^\nu \sqrt{\log^{1+\xi} (t-s)}
    \\
    \le & \delta_{t-s}^{1+a-b} \sqrt{\log^{1+\xi} (t-s)}.    
\end{align*}

\
\\
For \eqref{eq: f8c}, from \Cref{lemma: gt bound}, we write that with high probability
$$
\forall t \in \se, \qquad \|T^{a}(T + \lambda_{t-s} {\bf I})^{-b} g_\st\|_\lt \lesssim \delta_{t-s}^{1+a-b} \sqrt{\log^{1+\xi} (t-s)}.
$$

\
\\
For \eqref{eq: f8d}, from \Cref{lemma: operator bound}, we write that with high probability
\begin{align*}
    \forall t \in \se, \qquad &\|T^{a}(T + \lambda_{t-s} {\bf I})^{-b}(T_\st -T)\|_{\mathrm{op}} \|f^*_\st\|_\lt
    \\
    \lesssim &  \|T^{a}(T + \lambda_{t-s} {\bf I})^{-b}(T_\st -T)\|_{\mathrm{op}}
    \\
    \lesssim & \delta_{t-s}^{1+a-b} \sqrt{\log^{1+\xi} (t-s)}.    
\end{align*}
here we used the fact that $\|f^*_\st\|_\lt < \infty.$

\
\\
For \eqref{eq: f8e}, from \Cref{lemma: Ht bound}, we write that with high probability
$$
\forall t \in \se, \qquad \|T^{a}(T + \lambda_{t-s} {\bf I})^{-b} H_\st\|_\lt  \lesssim \delta_{t-s}^{1+a-b} \sqrt{\log^{1+\xi} (t-s)}.
$$

\
\\
This four individual bounds come together to give us the required bound for the first term on the right of \eqref{eq: f8a}.

\end{proof}

\
\\
\begin{lemma}\label{lemma: general emp/consistency type}
    Let $ 0 < \nu < \frac{1}{2} - \frac{1}{4r}$. Let $p\in \{0,\nu\}$ and $\xi>0$. Then
    $$
    \max_{s<t\le e} \frac{\delta_{t-s}^{-1}}{\log^{1+\xi} (t-s)} \|T^{-p}(T_\st - T)(\fhast - f^*_\st)\|_\lt  = O_p (1).
    $$
\end{lemma}

\
\\
\begin{proof}
    Using the linear operator norm inequality, we may have
 \begin{align}\nonumber
      &\|T^{-p}(T_\st - T)(\fhast - f^*_\st)\|_\lt 
      \\\nonumber
      \le& \lft\|T^{-p}(T_\st - T)(T + \lambda_{t-s} \boldI )^{-1/4} T^{-1/4} \rgt\|_{\mathrm{op}} \lft\| T^{1/4} (T + \lambda_{t-s} \boldI)^{1/4}\lft(\fhast - f^*_\st \rgt) \rgt\|_\lt.
 \end{align}
We are going to bound each of the two terms here. For the first one, using \Cref{lemma: operator bound}, we write that with high probability 
$$
\forall t \in \se \qquad \lft\|T^{-p}(T_\st - T)(T + \lambda_{t-s} \boldI )^{-1/4} T^{-1/4} \rgt\|_{\mathrm{op}} \lesssim \delta_{t-s}^{1/2} \sqrt{\log^{1+\xi}(t-s)}.
$$
And for the second term, we use \Cref{lemma: general excess risk type} to have
$$
\forall t \in \se  \qquad \lft\| T^{1/4} (T + \lambda_{t-s} \boldI)^{1/4}\lft(\fhast - f^*_\st \rgt) \rgt\|_\lt \lesssim \delta_{t-s}^{1/2} \sqrt{\log^{1+\xi}(t-s)}.
$$

\
\\
The two bounds come together to have the required result.
\end{proof}

\
\\
\begin{lemma}\label{lemma: gt bound}
Let $\xi > 0$ and $1 \ge b \ge a + 1/2 \ge 1/4$. Then we have
$$
\E\left[ \max_{s<t\le e} \frac{\delta_{t-s}^{2(b-a-1)}}{\log^{1+\xi} (t-s)} \|T^{a}(T + \lambda_{t-s} {\bf I})^{-b} g_\st \|_\lt^2 \right] = O(1),
$$
where $g_\st = \frac{1}{(t-s)} \sum_{j=s+1}^t \eps_j L_\kh(X_j) $ defined in \Cref{pro: rough estimator}.
\end{lemma}
\
\\
\begin{proof}
    We may write
\begin{align}\nonumber
    \|T^a (T + \lambda_{t-s} {\bf I})^{-b}g_\st\|^2_\lt & = \sum_{l\ge1} \lb T^a (T + \lambda_{t-s} {\bf I})^{-b}g_\st, \phi_l \rb^2_\lt
    \\\nonumber
    & = \sum_{l\ge1} \lb g_\st, T^a (T + \lambda_{t-s} {\bf I})^{-b} \phi_l \rb^2_\lt
    \\\nonumber
    & = \sum_{l\ge1} \frac{\mathfrak{s}_l^{2a}}{(\mathfrak{s}_l + \lambda_{t-s})^{2b}} \lb g_\st, \phi_l \rb^2_\lt
    \\\nonumber
    & = \sum_{l\ge1} \frac{\mathfrak{s}_l^{2a}}{(\mathfrak{s}_l + \lambda_{t-s})^{2b}} \lb \frac{1}{t-s}\sum_{i=s+1}^t \eps_i L_\kh(X_i), \phi_l \rb^2_\lt
    \\\nonumber
    & = \sum_{l\ge1} \frac{\mathfrak{s}_l^{2a}}{(\mathfrak{s}_l + \lambda_{t-s})^{2b}} \left( \frac{1}{t-s} \sum_{j=s+1}^t \eps_j \lb  X_j, L_\kh \phi_l \rb_\lt \right)^2
\end{align}
By linearity of the expectation it lead us to
\begin{align}\nonumber
    &\E\left[ \max_{s < t\le e} \|T^a (T+ \lambda_{t-s} {\bf I})^{-b} g_\st \|^2 \frac{\delta_t^{2(b-a-1)}}{\log^{1+\xi} t} \right] 
    \\\nonumber
    \le& \sum_{l\ge1} \E\left[ \max_{s < t\le e} \left( \frac{\delta_{t-s}^{(b-a-1)}}{(\log (t-s))^{(1+\xi)/2}} \frac{\mathfrak{s}_l^{a}}{(t-s)(\mathfrak{s}_l + \lambda_{t-s})^{b}} \sum_{j=1}^t \eps_j \lb  X_j, L_\kh \phi_l \rb_\lt \right)^2 \right]
    \\\label{eq: TL gt1}
    = & O(1).
\end{align}
Here \eqref{eq: TL gt1} follows from \Cref{lemma: TL first general}.
\end{proof}

\
\\
\begin{lemma}\label{lemma: operator bound}
Let $\xi > 0$ and $1 \ge b \ge a + 1/2 \ge 1/4$. Let $ 0 < \nu < \frac{1}{2} - \frac{1}{4r}$. Suppose $p\in \{0,\nu\}$. Then
$$
\E\left[ \max_{s<t\le e} \frac{\delta_{t-s}^{2(b-a-1)}}{\log^{1+\xi} (t-s)} \|T^{a}(T + \lambda_{t-s} {\bf I})^{-b} (T_\st - T) T^{-p}\|_{\mathrm{op}}^2 \right] = O(1).
$$
\end{lemma}

\
\\
\begin{proof}
    Using the definition of operator norm, we may write
$$
\| T^a (T + \lambda_{t-s} {\bf I})^{-b} (T - T_\st) T^{-p} \|_{\mathrm{op}} := \sup_{\substack{ \|h\|_\lt = 1}} \left| \lb h, T^a (T + \lambda_{t-s} {\bf I})^{-b} (T - T_\st) T^{-p} h \rb_\lt \right|.
$$
Let $h \in \lt$ such that $ \|h\|_\lt = 1 $. This means $h = \sum_{j\ge1} h_j \phi_j$ and $\sum_{j\ge1} h_j^2 =1$. Then 

 \begin{align}\nonumber
     \lb h,T^a (T + \lambda_{t-s} {\bf I})^{-b} (T - T_\st) T^{-p} h \rb_\lt &= \sum_{j \ge 1} \sum_{m \ge 1} h_j h_m \lb \phi_j , T^a(T + \lambda_{t-s} {\bf I})^{-b} (T - T_\st) T^{-p} \phi_m \rb_\lt
     \\\nonumber
     & = \sum_{j \ge 1} \sum_{m \ge 1} h_j h_m \lb T^a(T + \lambda_{t-s} {\bf I})^{-b}\phi_j , (T - T_\st) T^{-p} \phi_m \rb_\lt
     \\\nonumber
     & = \sum_{j \ge 1} \sum_{m \ge 1} h_j h_m \left\lb \frac{\mathfrak{s}_j^a}{(\mathfrak{s}_j + \lambda_{t-s})^b}\phi_j , (T - T_\st) \mathfrak{s}_m^{-p} \phi_m \right\rb_\lt
     \\\nonumber
     & = \sum_{j \ge 1} \sum_{m \ge 1} h_j h_m \frac{\mathfrak{s}_j^a}{(\mathfrak{s}_j + \lambda_{t-s})^b} \mathfrak{s}_m^{-p} \left\lb \phi_j , (T - T_\st)  \phi_m \right\rb_\lt
     \\\label{eq: ot1}
     &\le \sqrt{\sum_{j \ge 1} \sum_{m \ge 1}h_j^2 h_m^2} \sqrt{\sum_{j \ge 1} \sum_{m \ge 1}  \frac{\mathfrak{s}_j^{2a}}{(\mathfrak{s}_j + \lambda_{t-s})^{2b}} \mathfrak{s}_m^{-2p} \left\lb \phi_j , (T - T_\st)  \phi_m \right \rb^2_\lt  }
     \\\label{eq: ot2}
     & = \sqrt{\sum_{j \ge 1} \sum_{m \ge 1}  \frac{\mathfrak{s}_j^{2a}}{(\mathfrak{s}_j + \lambda_{t-s})^{2b}} \mathfrak{s}_m^{-2\nu} \left\lb \phi_j , (T - T_\st)  \phi_m \right \rb^2_\lt  }
 \end{align}
 The second last inequality \eqref{eq: ot1} follows from Cauchy-Schwarz, where one may think $\lb A, B \rb = \sum_{j \ge 1} \sum_{m \ge 1} A_{jm} B_{jm} $. The last equality \eqref{eq: ot2} follows from 
 $$
 \sum_{j \ge 1} \sum_{m \ge 1}h_j^2 h_m^2 =  \sum_{j \ge 1}h_j^2\sum_{m \ge 1} h_m^2 = 1,
 $$
 by definition of $h$.
\\ 
We have,
\begin{equation}\label{eq: operator expansion term}
     \| T^a (T + \lambda_{t-s} {\bf I})^{-b} (T - T_\st) T^{-p} \|_{\mathrm{op}}^2 \le \sum_{j \ge 1} \sum_{m \ge 1}  \frac{\mathfrak{s}_j^{2a}}{(\mathfrak{s}_j + \lambda_{t-s})^{2b}} \mathfrak{s}_m^{-2p} \left\lb \phi_j , (T - T_\st)  \phi_m \right \rb^2_\lt  
\end{equation}
\\
By linearity of expectation
\begin{align}\nonumber
    & \E\left[ \max_{s<t\le e} \|T^{a}(T + \lambda_{t-s} {\bf I})^{-b} (T_\st - T) T^{-p}\|_{\mathrm{op}}^2  \frac{\delta_{t-s}^{2(b-a-1)}}{\log^{1+\xi} t} \right]
    \\\nonumber
    \le & \sum_{m \ge 1} \mathfrak{s}_m^{-2p} \E \left[  \max_{s< t \le e} \sum_{j\ge1} \frac{\mathfrak{s}_j^{2a}}{(\mathfrak{s}_j + \lambda_{t-s})^{2b} } \frac{\delta_t^{2(b-a -1)} }{\log^{1+\xi}t }\left| \lb \phi_k, (T_\st - T)\phi_j \rb_\lt \right|^2 \right]
    \\\nonumber
    \lesssim &  \sum_{m \ge 1} \mathfrak{s}_m^{1-2p} < \infty,
\end{align}
where the last line follows from \Cref{lemma: TL second general}, and $ \sum_{m \ge 1} \mathfrak{s}_m^{1-2p} \asymp  \sum_{m \ge 1} m^{ -(1-2p)2r}$ is summable given we have $(1-2p)2r > 1$.
\end{proof}

\
\\
\begin{lemma}\label{lemma: Ht bound}
Let $\xi > 0$ and $1 \ge b \ge a + 1/2 \ge 1/4$. Suppose $\{h_j\}$ be some $\lt$ sequence that satisfies $\Sigma\lft[ L_\kh(h_j), L_\kh(h_j) \rgt]\le M < \infty$. Then we have
\begin{equation}\label{eq: 22  Ht general}
    \E\left[ \max_{s<t\le e} \frac{\delta_{t-s}^{2(b-a-1)}}{\log^{1+\xi} (t-s)} \lft\|T^{a} \lft(T + \lambda_{t-s} {\bf I} \rgt)^{-b} \lft( \frac{1}{t-s} \sum_{j=s+1}^t \lft( \lb L_\kh(X_j), h_j \rb_\lt L_\kh(X_j) - T h_j \rgt) \rgt)\rgt\|_{\mathrm{op}}^2  \right] = O(1).
\end{equation}

\
\\
Consequently, it holds that
\begin{equation}\label{eq: 22 Ht specific}
    \E\left[ \max_{s<t\le e} \frac{\delta_{t-s}^{2(b-a-1)}}{\log^{1+\xi} (t-s)} \|T^{a} \lft(T + \lambda_{t-s} {\bf I} \rgt)^{-b} H_\st \|_{\mathrm{op}}^2 \right] = O(1),
\end{equation}
where $H_\st = \frac{1}{t-s} \sum_{j=s+1}^t \lft( \lb L_\kh(X_j), f_j^* \rb_\lt L_\kh(X_j) - T f^*_j \rgt)$  defined in \Cref{lemma: bias expansion}, and that

\begin{equation}\label{eq: 22 Gt bound}
    \E\left[ \max_{s<t\le e} \frac{\delta_{t-s}^{2(b-a-1)}}{\log^{1+\xi} (t-s)} \|T^{a} \lft(T + \lambda_{t-s} {\bf I} \rgt)^{-b} G_\st \|_{\mathrm{op}}^2 \right] = O(1),    
\end{equation}

where $G_\st = \frac{1}{t-s} \sum_{j=s+1}^t \lft( \lb L_\kh(X_j), f_j^* - f^*_\st \rb_\lt L_\kh(X_j) - T (f^*_j - f^*_\st) \rgt)$.

\end{lemma}

\
\\
\begin{proof}
    We may write
\begin{align}\nonumber
    &\lft\|T^{a} \lft(T + \lambda_{t-s} {\bf I} \rgt)^{-b} \lft( \frac{1}{t-s} \sum_{j=s+1}^t \lft( \lb L_\kh(X_j), h_j \rb_\lt L_\kh(X_j) - T h_j \rgt) \rgt)\rgt\|_{\mathrm{op}}^2  
    \\\nonumber
    & = \sum_{m\ge1} \lft\lb T^a (T + \lambda_{t-s} {\bf I})^{-b} \lft( \frac{1}{t-s} \sum_{j=s+1}^t \lft( \lb L_\kh(X_j), h_j \rb_\lt L_\kh(X_j) - T h_j \rgt) \rgt), \phi_m \rgt\rb^2_\lt
    \\\nonumber
    & = \sum_{m\ge1} \lft\lb \lft( \frac{1}{t-s} \sum_{j=s+1}^t \lft( \lb L_\kh(X_j), h_j \rb_\lt L_\kh(X_j) - T h_j \rgt) \rgt), T^a (T + \lambda_{t-s} {\bf I})^{-b} \phi_m \rgt\rb^2_\lt
    \\\nonumber
    & = \sum_{m\ge1} \frac{\mathfrak{s}_m^{2a}}{(\mathfrak{s}_m + \lambda_{t-s})^{2b}} \lft\lb \lft( \frac{1}{t-s} \sum_{j=s+1}^t \lft( \lb L_\kh(X_j), h_j \rb_\lt L_\kh(X_j) - T h_j \rgt) \rgt), \phi_m \rgt\rb^2_\lt
    \\\nonumber
    & = \sum_{m\ge1} \frac{\mathfrak{s}_m^{2a}}{(\mathfrak{s}_m + \lambda_{t-s})^{2b}} \left( \frac{1}{t-s} \sum_{j=s+1}^t \lb X_i, L_\kh(h_j \rb_\lt \lb  X_i, L_\kh \phi_m \rb_\lt - \lb h_j, T\phi_m \rb_\lt \right)^2
\end{align}
By linearity of the expectation it lead us to
\begin{align}\nonumber
    &\E\left[ \max_{1 < t\le n} \lft\|T^a (T+ \lambda_{t-s} {\bf I})^{-b} \lft( \frac{1}{t-s} \sum_{j=s+1}^t \lft( \lb L_\kh(X_j), h_j \rb_\lt L_\kh(X_j) - T h_j \rgt) \rgt) \rgt\|^2 \frac{\delta_t^{2(b-a-1)}}{\log^{1+\xi} t} \right] 
    \\\nonumber
    \le& \sum_{m\ge1} \E\left[ \max_{1 < t\le n} \left( \frac{\delta_{t-s}^{(b-a-1)}}{(\log (t-s))^{(1+\xi)/2}} \frac{\mathfrak{s}_m^{a}}{(t-s)(\mathfrak{s}_m + \lambda_{t-s})^{b}} \sum_{j=s+1}^t \lb X_i, L_\kh(h_j \rb_\lt \lb  X_i, L_\kh \phi_m \rb_\lt - \lb h_j, T\phi_m \rb_\lt \right)^2 \right]
    \\\label{eq: TL ht1}
    = & O(1).
\end{align}
Here \eqref{eq: TL ht1} follows from \Cref{lemma: TL third general} because we have  $\Sigma \lft[L_\kh(h_j), L_\kh(h_j) \rgt] < \infty$.

\
\\
The result \eqref{eq: 22 Ht specific} follows from \eqref{eq: 22  Ht general} because $\|\beta^*_j\|_\hk < \infty$. For \eqref{eq: 22 Gt bound}, we can again use \eqref{eq: 22  Ht general} as we have $\Sigma\lft[ \beta^*_j - \beta^*_\st, \beta^*_j - \beta^*_\st\rgt] \le O(1) \max_{1\le k \le \Ktilde} \|\beta^*_{\eta_k}\|_\hk < \infty$.
\end{proof}




\
\\
\begin{lemma}\label{lemma: TL first general}
Let $\xi > 0$ and $1 \ge b \ge a + 1/2 \ge 1/4$. Then we have
$$
\E \left[  \max_{s< t \le e} \sum_{m\ge1} \frac{\mathfrak{s}_m^{2a}}{(\mathfrak{s}_m + \lambda_{t-s})^{2b} } \frac{\delta_{t-s}^{2(b-a -1)} }{\log^{1+\xi} (t-s) }\left| \frac{1}{ t-s }  \sum_{j=s+1}^t \lb X_j, L_\kh(\phi_m) \rb_\lt \eps_j \right|^2 \right] = O(1).
$$
\end{lemma}

\
\\
\begin{proof}
    For simplicity, denote $Y_{j,m} = \lb X_j, L_\kh(\phi_m) \rb_\lt \eps_j $. Observe that
    $$
    \Sigma\lft[L_\kh(\phi_m), L_\kh(\phi_m) \rgt] = \lb T \phi_m, \phi_m \rb_\lt = \mathfrak{s}_m.
    $$
\
\\
{ We are going to prove the result for a general interval $\{1,\ldots, T\}$, the result for the $\se$ follows from translation and stationarity.}

\
\\
    Using \Cref{lemma: general markov type bound II}, we may write 
    $$
    \E\lft[ \lft|\sum_{j=1}^t Y_{j,m} \rgt|^2 \rgt] \le O(t) \mathfrak{s}_m.
    $$
    We apply \Cref{lemma: maximal inequality useful} and to have this result: for any non-decreasing sequence $ \{\gamma_{t}\}_{t = 1}^T$
    \begin{equation}\label{eq: max ineq g term}
        \E\lft[ \max_{1<t\le T} \lft|\frac{1}{\gamma_{t}}\sum_{j=1}^t Y_{j,m} \rgt|^2 \rgt] = C \sum_{t=1}^T \frac{1}{\gamma_{t}^2} \mathfrak{s}_m,
    \end{equation}
for some constant $C>0$.

\
\\
Observe that
$$
\frac{\delta_{t}^{2 (b-a -1)} }{(\mathfrak{s}_m + \lambda_{t})^{2b - (1+2a) + (1 +2a)} } \le \frac{\delta_{t}^{2( b-a -1)} }{\lambda_{t}^{2b - (1+2a)}(\mathfrak{s}_m + \lambda_{t})^{ (1 +2a)} } \lesssim \frac{\delta_{t}^{-1}}{(\mathfrak{s}_m +\lambda_{t})^{1+2a}}.
$$
It led us to
\begin{align}\nonumber
    &\frac{\mathfrak{s}_m^{2a}}{(\mathfrak{s}_m + \lambda_{t})^{2b} } \frac{\delta_{t}^{2(b-a -1)} }{\log^{1+\xi}t }\left| \frac{1}{t}  \sum_{j=1}^t Y_{j,m} \right|^2 \lesssim \frac{\mathfrak{s}_m^{2a}}{(\mathfrak{s}_m + \lambda_{t})^{1+2a} } \frac{\delta_{t}^{ -1} }{\log^{1+\xi}t }\left| \frac{1}{t}  \sum_{j=1}^t Y_{j,m} \right|^2 
    \\\nonumber
    &\le \left\{ \frac{\mathfrak{s}_m^{2a}}{\mathfrak{s}_m^{1+2a}} \wedge \frac{\mathfrak{s}_m^{2a}}{\lambda_{t}^{1+2a}} \right\}\frac{\delta_{t}^{ -1} }{\log^{1+\xi}t }\left| \frac{1}{t}  \sum_{j=1}^t Y_{j,m} \right|^2
    \\\label{eq: f15a}
    \implies & \sum_{m\ge1} \frac{\mathfrak{s}_m^{2a}}{(\mathfrak{s}_m + \lambda_{t})^{2b} } \frac{\delta_{t}^{2(b-a -1)} } {\log^{1+\xi}t }\left| \frac{1}{t}  \sum_{j=1}^t Y_{j,m}\right|^2 \lesssim \sum_{m\ge1}  \left\{ \frac{\mathfrak{s}_m^{2a}}{\mathfrak{s}_m^{1+2a}} \wedge \frac{\mathfrak{s}_m^{2a}}{\lambda_{t}^{1+2a}} \right\}\frac{\delta_{t}^{ -1} }{\log^{1+\xi}t }\left| \frac{1}{t}  \sum_{j=1}^t Y_{j,m} \right|^2
\end{align}

\
\\
{\bf Case I: $a\le0$}

\
\\
Let $f_m = \lfloor{m^{(2r +1)}} \rfloor \wedge T$. Using \eqref{eq: f15a}, we write
\begin{align}\nonumber
    &\sum_{m\ge1} \frac{\mathfrak{s}_m^{2a}}{(\mathfrak{s}_m + \lambda_t)^{2b} } \frac{\delta_t^{2(b-a -1)} }{\log^{1+\xi}t }\left| \frac{1}{ t }  \sum_{j=1}^t Y_{j,m}\right|^2 
    \\\label{eq: tln 2}
    \le & \sum_{m\ge 1} {\bf I}\left\{ t \le f_m \right\}  \mathfrak{s}_m^{2a}\left|\frac{\delta_t^{ -1/2} }{(\log t)^{(1+\xi)/2} } \frac{1}{ \lambda_t^{a + 1/2} t }  \sum_{j=1}^t Y_{j,m} \right|^2 + \sum_{m\ge 1} {\bf I}\left\{ t > f_m \right\}
    \mathfrak{s}_m^{-1}\left|\frac{\delta_t^{ -1/2} }{(\log t)^{(1+\xi)/2} } \frac{1}{ t }  \sum_{j=1}^t Y_{j,m} \right|^2 .
\end{align}
Observe that for $ 2 \le t \le T$,
$$
\frac{d}{dt} \left( t \delta_t^{1/2} (\log t)^{(1+\xi)/2} \right) = t^{(1+r)/(2r+1)} (\log t)^{(\xi - 1)/2} \left( \frac{1 +r}{2r +1} + \frac{1+\xi}{2} \log t \right) >0  
$$
and 
$$
\frac{d}{dt} \left( t \lambda_t^{a + 1/2} \delta_t^{1/2} (\log t)^{(1+\xi)/2} \right) = \lambda_t^{a + 1/2} \delta_t^{1/2}(\log t)^{(\xi - 1)/2} \left( \frac{1+r - (a+1/2)2r}{2r+1} \log t
+\frac{1+\xi}{2}\right) > 0.
$$
This says that $\left\{t\delta_t^{1/2} (\log t)^{(1+\xi)/2} \right\}$ and $\left\{t \lambda_t^{a + 1/2} \delta_t^{1/2} (\log t)^{(1+\xi)/2} \right\}$ satisfies the criteria for $\left\{ \gamma_t \right\}$ in 
\eqref{eq: max ineq g term}.

\
\\
This observation on derivatives and \eqref{eq: tln 2} helps us to write
\begin{align}\nonumber
    &\E \left[ \max_{1<t\le T} \sum_{m\ge1} \frac{\mathfrak{s}_m^{2a}}{(\mathfrak{s}_m + \lambda_t)^{2b} } \frac{\delta_t^{2(b-a -1)} }{\log^{1+\xi}t }\left| \frac{1}{ t }  \sum_{j=1}^t Y_{j,m}\right|^2 \right]
    \\\nonumber
    \le & \sum_{m\ge 1} \mathfrak{s}_m^{2a}  \E \left[ \max_{1<t\le f_m} \left|\frac{\delta_t^{ -1/2} }{(\log t)^{(1+\xi)/2} } \frac{1}{ \lambda_t^{a + 1/2} t }  \sum_{j=1}^t Y_{j,m} \right|^2 \right] + \sum_{m\ge 1} \mathfrak{s}_m^{-1} \E \left[\max_{f_m <t\le T}\left|\frac{\delta_t^{ -1/2} }{(\log t)^{(1+\xi)/2} } \frac{1}{ t }  \sum_{j=1}^t Y_{j,m} \right|^2   \right]
    \\\label{eq: tln 3}
    \le & \sum_{m\ge 1} \mathfrak{s}_m^{2a} c\sum_{t\le f_m } \frac{\delta_t^{-1}}{\log^{1+\xi} t \lambda_t^{1+2a} t^2} \mathfrak{s}_m + \sum_{m\ge 1}\mathfrak{s}_m^{-1} c\sum_{t>f_m} \frac{\delta_t^{-1}}{t^2 \log^{1+\xi} t} \mathfrak{s}_m
    \\\nonumber
    = &    c \sum_{1<t\le T} \frac{\delta_t^{-1}}{t^2\log^{1+\xi} t } \sum_{m \ge \delta_t^{-1/2r}} \frac{\mathfrak{s}_m^{1+2a}}{\lambda_t^{1+2a}} + c \sum_{1<t\le T}  \frac{\delta_t^{-1}}{t^2 \log^{1+\xi} t} \sum_{m < \delta_t^{-1/2r}}  1
    \\\label{eq: tln 4}
    = & c \sum_{1<t\le T}  \frac{\delta_t^{-1}}{t^2 \log^{1+\xi} t} O(\delta_t^{-1/2r}) + c\sum_{1<t\le T}  \frac{\delta_t^{-1}}{t^2 \log^{1+\xi} t} O(\delta_t^{-1/2r})
    \\\nonumber
    =  &  \sum_{1<t\le T}  \frac{1}{t \log^{1+\xi} t} O\left(\frac{\delta_t^{-1-1/2r}}{t} \right) = \sum_{1<t\le T}  \frac{1}{t \log^{1+\xi} t} O(1) < \infty.
\end{align}
The \eqref{eq: tln 3} follows from \eqref{eq: max ineq g term} and \eqref{eq: tln 4} follows from \eqref{test 1} with the observation
$$
\sum_{m \ge \delta_t^{-1/2r}} \frac{\mathfrak{s}_m^{1+2a}}{\lambda_t^{1+2a}} \lesssim \sum_{m \ge \delta_t^{-1/2r}} \frac{1}{\lft( m^{2r} \delta_t \rgt)^{1+2a}} \le O(1) \int_{\delta_t^{-1/2r}}^\infty  \frac{1}{\lft( x^{2r} \delta_t \rgt)^{1+2a}} dx.
$$

\vspace{2.0cm}
\
\\
{\bf Case II: $a>0$}

\
\\
Let $f_m = \lfloor{m^{(2r +1)}} \rfloor \wedge T$. Using \eqref{eq: f15a}, we write
\begin{align}\nonumber
    &\sum_{m\ge1} \frac{\mathfrak{s}_m^{2a}}{(\mathfrak{s}_m + \lambda_t)^{2b} } \frac{\delta_t^{2(b-a -1)} }{\log^{1+\xi}t }\left| \frac{1}{ t }  \sum_{j=1}^t Y_{j,m}\right|^2 
    \\\nonumber
    \le & \sum_{m\ge 1} {\bf I}\left\{ t < f_m \right\}  \mathfrak{s}_m^{2a}\left|\frac{\delta_t^{ -1/2} }{(\log t)^{(1+\xi)/2} } \frac{1}{ \lambda_t^{a + 1/2} t }  \sum_{j=1}^t Y_{j,m} \right|^2 + \sum_{m\ge 1} {\bf I}\left\{ t \ge f_m \right\}
    \mathfrak{s}_m^{-1}\left|\frac{\delta_t^{ -1/2} }{(\log t)^{(1+\xi)/2} } \frac{1}{ t }  \sum_{j=1}^t Y_{j,m} \right|^2
    \\\label{eq: tln 5}
    \le & \sum_{m\ge 1} {\bf I}\left\{ t < f_m \right\}  \frac{\mathfrak{s}_m^{2a}}{\lambda_{f_m}^{2a}}\left|\frac{\delta_t^{ -1/2} }{(\log t)^{(1+\xi)/2} } \frac{1}{ \lambda_t^{ 1/2} t }  \sum_{j=1}^t Y_{j,m} \right|^2 + \sum_{m\ge 1} {\bf I}\left\{ t \ge f_m \right\}
    \mathfrak{s}_m^{-1}\left|\frac{\delta_t^{ -1/2} }{(\log t)^{(1+\xi)/2} } \frac{1}{ t }  \sum_{j=1}^t Y_{j,m} \right|^2
\end{align}
We have $t <f_m \Rightarrow \lambda_t > \lambda_{f_m}   $ which gives us \eqref{eq: tln 5}.
\\
Observe that for $ 2 \le t \le n$,
$$
\frac{d}{dt} \left( t \delta_t^{1/2} (\log t)^{(1+\xi)/2} \right) = t^{(1+r)/(2r+1)} (\log t)^{(\xi - 1)/2} \left( \frac{1 +r}{2r +1} + \frac{1+\xi}{2} \log t \right) >0  
$$
and 
$$
\frac{d}{dt} \left( t \lambda_t^{ 1/2} \delta_t^{1/2} (\log t)^{(1+\xi)/2} \right) = \lambda_t^{1/2} \delta_t^{1/2}(\log t)^{(\xi - 1)/2} \left( \frac{1}{2r+1}\log t
+\frac{1+\xi}{2} \right) > 0.
$$
This says that $\left\{t\delta_t^{1/2} (\log t)^{(1+\xi)/2} \right\}$ and $\left\{ t\lambda_t^{ 1/2} \delta_t^{1/2} (\log t)^{(1+\xi)/2} \right\}$ satisfies the criteria for $\left\{ \gamma_t \right\}$ in 
\eqref{eq: max ineq g term}.

\
\\
This observation on derivatives and \eqref{eq: tln 5} helps us to write
\begin{align}\nonumber
    &\E \left[ \max_{1<t\le T} \sum_{m\ge1} \frac{\mathfrak{s}_m^{2a}}{(\mathfrak{s}_m + \lambda_t)^{2b} } \frac{\delta_t^{2(b-a -1)} }{\log^{1+\xi}t }\left| \frac{1}{ t }  \sum_{j=1}^t Y_{j,m}\right|^2 \right]
    \\\nonumber
    \le & \sum_{m\ge 1} \frac{\mathfrak{s}_m^{2a}}{\lambda_{f_m}^{2a}}  \E \left[ \max_{1<t < f_m} \left|\frac{\delta_t^{ -1/2} }{(\log t)^{(1+\xi)/2} } \frac{1}{ \lambda_t^{1/2} t }  \sum_{j=1}^t Y_{j,m} \right|^2 \right] + \sum_{m\ge 1} \mathfrak{s}_m^{-1} \E \left[\max_{f_m \le t\le T}\left|\frac{\delta_t^{ -1/2} }{(\log t)^{(1+\xi)/2} } \frac{1}{ t }  \sum_{j=1}^t Y_{j,m} \right|^2   \right]
    \\\label{eq: tln 6}
    \le & \sum_{m\ge 1} \frac{\mathfrak{s}_m^{2a}}{\lambda_{f_m}^{2a}} c\sum_{t< f_m } \frac{\delta_t^{-1}}{\log^{1+\xi} t \lambda_t t^2} \mathfrak{s}_m + \sum_{m\ge 1}\mathfrak{s}_m^{-1} c\sum_{t\ge f_m} \frac{\delta_t^{-1}}{t^2 \log^{1+\xi} t} \mathfrak{s}_m
    \\\nonumber
    = &    c \sum_{1<t\le T} \frac{\delta_t^{-1}}{t^2\log^{1+\xi} t } \sum_{m > \delta_t^{-1/2r}} \frac{\mathfrak{s}_m^{1+2a}}{\lambda_t\lambda_{f_m}^{2a}} + c \sum_{1<t\le T}  \frac{\delta_t^{-1}}{t^2 \log^{1+\xi} t} \sum_{m \le \delta_t^{-1/2r}} 1
    \\\label{eq: tln 7}
    = & c \sum_{1<t\le T}  \frac{\delta_t^{-1}}{t^2 \log^{1+\xi} t} O(\delta_t^{-1/2r}) + c\sum_{1<t\le T}  \frac{\delta_t^{-1}}{t^2 \log^{1+\xi} t} O(\delta_t^{-1/2r})
    \\\nonumber
    =  &  \sum_{1<t\le T}  \frac{1}{t \log^{1+\xi} t} O\left(\frac{\delta_t^{-1-1/2r}}{t} \right) = \sum_{1<t\le T}  \frac{1}{t \log^{1+\xi} t} O(1) < \infty.
\end{align}
The \eqref{eq: tln 6} follows from \eqref{eq: max ineq g term}. For \eqref{eq: tln 7}, with the realization $\lambda_{f_m} \asymp \mathfrak{s}_m$ we may write
$$
\frac{\mathfrak{s}_m^{1+2a}}{\lambda_t \lambda_{f_m}^{2a} } \le c_2 \frac{1}{m^{2r}\lambda_t} \Longrightarrow \sum_{m> \delta_r^{-1/2r}} \frac{\mathfrak{s}_m^{1+2a}}{\lambda_t \lambda_{f_m}^{2a} } \le c_2 \sum_{m> \delta_t^{-1/2r}} \frac{1}{m^{2r} \lambda_t} \le c_2 \int_{\delta_t^{-1/2r}}^\infty \frac{1}{x^{2r} \lambda_t} dx = O(\delta_t^{-1/2r}),
$$
similar idea is outlined at \eqref{test 1} which comes as a consequence from \Cref{lemma: sum-integration bound}.
    
\end{proof}

\
\\
\begin{lemma}\label{lemma: TL second general}
Let $\xi > 0$ and $1 \ge b \ge a + 1/2 \ge 1/4$. Then for any $k\ge1$, we have
$$
\E \left[  \max_{s< t \le e} \sum_{j\ge1} \frac{\mathfrak{s}_j^{2a}}{(\mathfrak{s}_j + \lambda_{t-s})^{2b} } \frac{\delta_{t-s}^{2(b-a -1)} }{\log^{1+\xi}(t-s) }\left| \lb \phi_k, (T_\st - T)\phi_j \rb_\lt \right|^2 \right] = O(\mathfrak{s}_k).
$$

\end{lemma}

\
\\
\begin{proof}
    Denote $u_{j,k} = \lb X_j, L_\kh (\phi_k) \rb_\lt$ and $u_{j,m} = \lb X_j, L_\kh (\phi_m) \rb_\lt$. Let $Y_{j,m}^{k} = u_{j,k} u_{j,m} - \E[u_{j,k} u_{j,m}] =  \lb X_j, L_\kh (\phi_k) \rb_\lt\lb X_j, L_\kh (\phi_m) \rb_\lt - \lb\phi_k, T\phi_m\rb_\lt$. Observe that
    $$
    \lb \phi_k, (T_\st - T)\phi_j \rb_\lt  = \frac{1}{t-s} \sum_{i = s+1}^t Y_{j,m}^{k}.
    $$
    
\
\\
{ Again, We are going to prove the result for a general interval $\{1,\ldots, T\}$, the result for the $\se$ follows from translation and stationarity.}

\
\\
Using \Cref{lemma: general markov type bound I}, we may write
$$
 \E\lft[ \lft| \sum_{j = 1}^t Y_{j,m}^{k} \rgt|^2  \rgt] = O(t) \mathfrak{s}_k \mathfrak{s}_m.
$$
We use \Cref{lemma: maximal inequality useful} to establish for any non-decreasing sequence $\{\gamma_t\}$ 

\begin{equation}\label{eq: max ineq TT term}
        \E\lft[ \max_{1<t\le T} \lft|\frac{1}{\gamma_{t}}\sum_{j=1}^t Y_{j,m}^{k} \rgt|^2 \rgt] = C \sum_{t=1}^T \frac{1}{\gamma_{t}^2} \mathfrak{s}_k \mathfrak{s}_m,
\end{equation}
for some constant $C>0$.

The rest of proof follows exactly as the proof of \Cref{lemma: TL first general}, just by replacing $Y_{j,m}^{k}$ with $Y_{j,m}$ and therefore omitted.

\end{proof}

\
\\
\begin{lemma}\label{lemma: TL third general}
    Let $\xi > 0$ and $1 \ge b \ge a + 1/2 \ge 1/4$. Let $\{h_i\}$ be sequence of $\lt$ functions such that $\Sigma[L_\kh(h_i), L_\kh(h_i)] \le  M < \infty$. Then we have 
\begin{align*}
        &\E \left[  \max_{s< t \le e} \sum_{m\ge1} \frac{\mathfrak{s}_m^{2a}}{(\mathfrak{s}_m + \lambda_{t-s})^{2b} } \frac{\delta_{t-s}^{2(b-a -1)} }{\log^{1+\xi} (t-s) } \lft|\frac{1}{t-s} \sum_{j=s+1}^t \lft( \lb X_j, L_\kh(h_j) \rb_\lt \lb X_j, L_\kh(\phi_m) \rb_\lt - \lb  h_j, T \phi_m \rb_\lt \rgt) \rgt|^2 \right] 
        \\
        = & O(1).    
\end{align*}

\end{lemma}

\
\\
\begin{proof}
    Let 
    $$
    Y_{j, m}  =  \lb X_j, L_\kh(h_j) \rb_\lt \lb X_j, L_\kh(\phi_m) \rb_\lt - \lb  h_j, T \phi_m \rb_\lt
    $$
\
\\
Similar to the last two proofs, we are going to establish the result on a generic interval $\{1,\ldots, T\}$, the case in the lemma follows from translation and stationarity.

\
\\
    Observe that 
    $$
    \E\lft[ \lb X_j, L_\kh(h_j) \rb_\lt \lb X_j, L_\kh(\phi_m) \rb_\lt\rgt] = \lb  h_j, T \phi_m \rb_\lt,
    $$
    and
    $$
    \Sigma[L_\kh(\phi_m), L_\kh(\phi_m)] = \lb T\phi_m, \phi_m\rb_\lt = \mathfrak{s}_m.
    $$
    Using this, from \Cref{lemma: general markov type bound I}, we may establish
    $$
    \E\lft[ \lft|\sum_{j=1}^t Y_{j,m} \rgt|^2 \rgt] \le O(1) \sum_{j=1}^t \mathfrak{s}_m.
    $$

Now, similar to \eqref{eq: max ineq g term}, we apply \Cref{lemma: maximal inequality useful} to have: for any non-decreasing sequence $\{\gamma_{t}\}_{t = 1}^T$
    \begin{equation}
        \E\lft[ \max_{1<t\le T} \lft|\frac{1}{\gamma_{t}}\sum_{j=1}^t Y_{j,m} \rgt|^2 \rgt] = C \sum_{t=1}^T \frac{1}{\gamma_{t}^2} \mathfrak{s}_m,
    \end{equation}
for some constant $C>0$.

\
\\
The rest of proof follows exact same steps as the proof of \Cref{lemma: TL first general} and therefore omitted.

\end{proof}

\
\\
\begin{lemma}\label{lemma:lambda transfer useful}
    Let $a,b,q >0$. Let $p,r$ be some constant. Suppose $D: \lt \to \lt$ be some linear operator. Suppose $f, h \in \lt$.
    Then we have
    \begin{align}\label{eq: l27a}
         &\|T^p D f \|_\lt \le \lft( 1 +  \lft( \frac{\lambda_b}{\lambda_a} \rgt)^q \rgt) \|T^p D (T + \lambda_b \boldI )^{-q} T^{-r} \|_{\mathrm{op}} \| T^r (T + \lambda_a \boldI )^{q} f \|_\lt
         \\\label{eq: l27b}
         &|\lb h, f \rb_\lt | \le \lft( 1 +  \lft( \frac{\lambda_b}{\lambda_a} \rgt)^q \rgt) \|T^{-p} (T + \lambda_b \boldI )^{-q} h \|_\lt \|T^{p} (T + \lambda_a \boldI )^{q} h \|_\lt.
    \end{align}
        
\end{lemma}

\
\\
\begin{proof}
We are going to establish \eqref{eq: l27a} and the proof for \eqref{eq: l27b} follows similarly. The proof is divided in three steps. We establish some necessary result in Step 1 and Step 2 and complete the proof in Step 3 by using them.

\
\\
{\bf Step 1:} For $d\ge c$, we are going to establish the following in this step.
\begin{equation}\label{eq: D5 1}
    \|T^p (T + \lambda_d \boldI)^{q} f \|_\lt \le \|T^p (T + \lambda_c \boldI)^{q} f \|_\lt \le \lft( \frac{\lambda_c}{\lambda_d}\rgt)^q \|T^p (T + \lambda_d \boldI)^{q} f \|_\lt,
\end{equation}

\
\\
Let $f = \sum_{j \ge 1}a_j \phi_j$.
    Observe that $j^{-2r} \asymp \mathfrak{s}_j > 0$ and
    $$
    d \ge c \iff (\mathfrak{s}_j + \lambda_c) \ge (\mathfrak{s}_j + \lambda_d) \iff \frac{1}{(\mathfrak{s}_j + \lambda_d)} \ge \frac{1}{(\mathfrak{s}_j + \lambda_c)} \iff \frac{\lambda_c}{ \lambda_d}  \ge \frac{\mathfrak{s}_j + \lambda_c}{\mathfrak{s}_j + \lambda_d}.
    $$
    It lead us to
    $$
    \|T^p (T + \lambda_d \boldI)^{q} f \|^2_\lt = \sum_{j \ge 1} \mathfrak{s}_j^{2p} (\mathfrak{s}_j + \lambda_d )^{2q} a_j^2 \le \sum_{j \ge 1} \mathfrak{s}_j^{2p} (\mathfrak{s}_j + \lambda_c )^{2q} a_j^2 = \|T^p (T + \lambda_c \boldI)^{q} f \|^2_\lt
    $$
and
   \begin{align*}
   \|T^p (T + \lambda_c \boldI)^{q} f \|^2_\lt =   \sum_{j \ge 1} \mathfrak{s}_j^{2p}(\mathfrak{s}_j + \lambda_c )^{2q} a_j^2 =&  \sum_{j \ge 1}  \lft( \frac{\mathfrak{s}_j + \lambda_c}{\mathfrak{s}_j + \lambda_d}\rgt)^{2b} \mathfrak{s}_j^{2p}(\mathfrak{s}_j + \lambda_d )^{2q} a_j^2 
   \\
   \le& \lft( \frac{ \lambda_c}{\lambda_d}\rgt)^{2b} \sum_{j \ge 1}   \frac{\mathfrak{s}_j^{2p}}{(\mathfrak{s}_j + \lambda_d )^{2q}} a_j^2  = \lft( \frac{ \lambda_c}{\lambda_d}\rgt)^{2b} \|T^p (T + \lambda_d \boldI)^{-q} f \|^2_\lt.
    \end{align*}

\
\\
{\bf Step 2:} For $d\ge c$, we are going to establish the following in this step.
\begin{equation}\label{eq: D5 2}
    \| T^{p} D (T + \lambda_c \boldI)^{-q} T^{r}\|_{\mathrm{op}} \le \| T^{p} D (T + \lambda_d \boldI)^{-q} T^{r}\|_{\mathrm{op}}.
\end{equation}
 Observe that 
    $$
    b \ge a \iff (\mathfrak{s}_j + \lambda_c) \ge (\mathfrak{s}_j + \lambda_b) \iff \frac{1}{(\mathfrak{s}_j + \lambda_b)} \ge \frac{1}{(\mathfrak{s}_j + \lambda_c)}.
    $$
    It lead us to
    \begin{align*}
        &\| T^{p} D (T + \lambda_c \boldI)^{-q} T^{r}\|_{\mathrm{op}}
        \\
        =&\sup_{\substack{ h\in \lt \\ \|h\|_\lt = 1 }} \lft| \lft\lb h, T^{p} D (T + \lambda_c \boldI)^{-q} T^{r} h \rgt\rb_\lt \rgt|
        \\
        = & \sup_{\substack{ h\in \lt \\ \|h\|_\lt = 1 }} \lft| \sum_{j\ge 1} \sum_{m\ge 1} h_j h_m \lft\lb \phi_j, T^{p} D (T + \lambda_c \boldI)^{-q} T^{r} \phi_m \rgt\rb_\lt \rgt|
        \\
        = & \sup_{\substack{ h\in \lt \\ \|h\|_\lt = 1 }} \lft| \sum_{j\ge 1} \sum_{m\ge 1} h_j h_m \lft\lb T^{p} \phi_j, D (T + \lambda_c \boldI)^{-q} T^{r} \phi_m \rgt\rb_\lt \rgt|
        \\
        = & \sup_{\substack{ h\in \lt \\ \|h\|_\lt = 1 }} \lft| \sum_{j\ge 1} \sum_{m\ge 1} h_j h_m \frac{\mathfrak{s}_j^{2p} \mathfrak{s}_m^{2r}}{(\mathfrak{s}_m+\lambda_c)^{2q}} \lft\lb  \phi_j, D \phi_m \rgt\rb_\lt \rgt|
        \\
        \le & \sup_{\substack{ h\in \lt \\ \|h\|_\lt = 1 }} \lft| \sum_{j\ge 1} \sum_{m\ge 1} h_j h_m \frac{\mathfrak{s}_j^{2p}\mathfrak{s}_m^{2r}}{(\mathfrak{s}_m+\lambda_d)^{2q}} \lft\lb  \phi_j, D \phi_m \rgt\rb_\lt \rgt|
        \\
        = & \sup_{\substack{ h\in \lt \\ \|h\|_\lt = 1 }} \lft| \sum_{j\ge 1} \sum_{m\ge 1} h_j h_m \lft\lb T^{p} \phi_j, D (T + \lambda_d \boldI)^{-q} T^{r} \phi_m \rgt\rb_\lt \rgt|
        \\
        = & \sup_{\substack{ h\in \lt \\ \|h\|_\lt = 1 }} \lft| \sum_{j\ge 1} \sum_{m\ge 1} h_j h_m \lft\lb \phi_j, T^{p} D (T + \lambda_d \boldI)^{-q} T^{r} \phi_m \rgt\rb_\lt \rgt|
        \\
        =&\sup_{\substack{ h\in \lt \\ \|h\|_\lt = 1 }} \lft| \lft\lb h, T^{p} D (T + \lambda_d \boldI)^{-q} T^{r} h \rgt\rb_\lt \rgt|
        \\
        = & \| T^{p} D (T + \lambda_d \boldI)^{-q} T^{r}\|_{\mathrm{op}}.
    \end{align*}

\
\\
{\bf Step 3:}
Using \eqref{eq: D5 1} and \eqref{eq: D5 2}, we may write
    \begin{align*}
        &\|T^p D f \|_\lt 
        \\
        =& \boldI\{b \ge a\} \|T^p D (T + \lambda_b \boldI )^{-q} T^{-r} T^r (T + \lambda_b \boldI )^{q} f \|_\lt  + \boldI\{b \le  a\} \|T^p D (T + \lambda_b \boldI )^{-q} T^{-r} T^r (T + \lambda_b \boldI )^{q} f \|_\lt 
        \\
        \le & \boldI\{b \ge a\} \|T^p D (T + \lambda_b \boldI )^{-q} T^{-r} \|_{\mathrm{op}} \| T^r (T + \lambda_a \boldI )^{q} f \|_\lt  
        \\
        &\qquad + \boldI\{b \le  a\} \lft( \frac{\lambda_b}{\lambda_a} \rgt)^q \|T^p D (T + \lambda_b \boldI )^{-q} T^{-r} \|_{\mathrm{op}} \| T^r (T + \lambda_a \boldI )^{q} f \|_\lt 
        \\
        \le & \lft( 1 +  \lft( \frac{\lambda_b}{\lambda_a} \rgt)^q \rgt) \|T^p D (T + \lambda_b \boldI )^{-q} T^{-r} \|_{\mathrm{op}} \| T^r (T + \lambda_a \boldI )^{q} f \|_\lt .
    \end{align*}
\end{proof}

\newpage

\section{Lower bound}

\begin{proof}[Proof of \Cref{lemma: lower bound}]
    We prove a more general result and the required result follows as a special case.

    For $Z_j = (Y_j, X_j)$, let $P_0^n$ be the joint distribution of $\{Z_j\}_{j=1}^n$ following
    \begin{equation}\nonumber
        \begin{aligned}
            y_j &= \lb X_j, \beta \rb_\lt + \eps_j, \qquad\qquad\text{for}\; 1 \le j \le \Delta,
            \\
            y_j &=  \eps_j, \quad\qquad\qquad\qquad\qquad\text{for}\; \Delta < j \le n,
        \end{aligned}
    \end{equation}
    where $\{X_j\}_{j=1}^n$ is independent standard Brownian motion and $\{\eps_j\}_{j=1}^n \overset{iid}{\sim} N(0,1)$.
    Let $P_1^n$ be the joint distribution of $\{Z'_j  \}_{j=1}^n$ with $Z'_j = (Y'_j, X'_j)$ which follows
    \begin{equation}\nonumber
        \begin{aligned}
            y_j' &= \lb X_j', \beta \rb_\lt + \eps_j', \qquad\qquad\text{for}\; 1 \le j \le \Delta + \delta,
            \\
            y_j' &=  \eps_j', \quad\qquad\qquad\qquad\qquad\text{for}\; \Delta + \delta < j \le n.
        \end{aligned}
    \end{equation}
    where $\{X'_j\}_{j=1}^n$ is independent standard Brownian motion and $\{\eps'_j\}_{j=1}^n \overset{iid}{\sim} N(0,1)$. We assume that the two datasets are independent. Observe that
    $$
    KL\lft(P_0^n;P_1^n \rgt) = \sum_{j=\Delta+1}^{\Delta+\delta} KL\lft(P_0^{j,n};P_1^{j,n} \rgt), 
    $$
    where $P_0^{j,n}(y,x)$ and $P_1^{j,n}(y,x)$ are distributions of $(y_j, X_j)$ and $(y_j',X_j')$ respectively. For $\Delta < j \le \Delta + \delta$, one may write
    \begin{align}\nonumber
        KL\lft(P_0^{j,n};P_1^{j,n} \rgt) &= \iint \log \lft\{ \frac{p_0^{j,n}(y|x)}{p_1^{j,n}(y|x)} \rgt\} p^{j,n}_0(y|x) p(x) \; dy\;dx 
        \\\nonumber
        &= \iint \frac{1}{2} \lft( \lb x, \beta\rb_\lt^2 - 2y  \lb x, \beta\rb_\lt \rgt) p^{j,n}_0(y|x) p(x) \; dy\;dx 
        \\\nonumber
        &= \frac{1}{2}\int \lb x, \beta\rb_\lt^2 p(x) dx = \frac{\kappa^2}{2},
    \end{align}
    where in the first line we used the conditional density $p_0^{j,n}(y|x)$, $p_1^{j,n}(y|x)$, and $p(x)$ as the density of $X_j$; in the second and the last line we use the fact that $y_j|X_j \sim N\lft(0, 1\rgt)$ under $p_0^{j,n}(y|x)$. This lead us to $KL\lft(P_0^n;P_1^n \rgt) = \delta\kappa^2/2$ and we already have $\eta(P_1^n) - \eta(P_0^n) = \delta$. Following from LeCam's lemma (see e.g.\ \citealp{Binyu1997assouad} and Theorem 2.2 of \citealp{tsybakov_book}), we may write
    $$
        \inf_{\weta}\; \sup_{P \in \mathfrak{P}} \E\lft[ |\weta - \eta(P)| \rgt] \ge \frac{\delta}{4} e^{-\delta \kappa^2/2}.
    $$
    The result now follows by putting $\delta = \frac{4}{\kappa^2}$ with the realization that, for large $n$, $\frac{4}{\kappa^2} \ll \Delta< n/2$.

\end{proof}

\newpage

\section{$\alpha$-mixing}
The strong mixing or $\alpha$-mixing coefficient between two $\sigma-$fields $\mathcal{A}$ and $\mathcal{B}$ is defined as 
$$
\alpha(\mathcal{A}, \mathcal{B}) = \sup_{A\in \mathcal{A}, B \in \mathcal{B}} \left| P(A\cap B) - P(A)P(B) \right|.
$$

\begin{lemma}\label{lemma: mixing covariance inequality}
Let $X$ and $Y$ be random variables. Then for any positive numbers $p,q,r$ satisfying $\frac{1}{p} + \frac{1}{q} + \frac{1}{r} = 1 $, we have
$$
|Cov(X,Y)| \le 4 \|X\|_p \|Y\|_q \left\{ \alpha(\sigma(X), \sigma(Y))\right\}^{1/r}.
$$

\end{lemma} 

\subsection{Strong law of large numbers}
\begin{theorem}\label{SLLN: alpha mixing}
    Let $\{Z_t\}$ be centered alpha mixing time series such that $\alpha(k) = O\lft( \frac{1}{ \lft(k L_k^2\rgt)^{\rho/(\rho-2)} } \rgt)$ for some $\rho >2$, where $L_k$ is non-decreasing sequence satisfying 
    $$
    \sum_{k=1}^\infty\frac{1}{kL_k} < \infty \qquad\qquad \text{and} \qquad\qquad L_{k} - L_{k-1} = O\lft(L_k/k\rgt).
    $$
    Suppose for some $1 \le p \le \rho < \infty$ one has
    $$
        \sum_{t=1}^\infty \frac{\E^{2/\rho}\lft[ |Z_t|^p\rgt]}{t^p} < \infty.
    $$
Then $\sum_{t =1}^n Z_t/n$ converges a.s  to $0$.
\end{theorem}

\
\\
\begin{proof}
    Using $L_n$ in the Definition 1.4 of \citet{mcleish1975}, $\{\alpha(n)\}$ is sequence of size$-\rho/(\rho-2)$. Following their Remark 2.6b, the results directly follows from their Lemma 2.9 with $g_t(x) = x^p$, $d_t=1$ and $X_t = Z_t/t$.
\end{proof}


\subsection{Central limit theorem}

Below is the central limit theorem for $\alpha$-mixing random variable. For a proof, one may see \cite{doukhan2012mixing}.
\begin{theorem}\label{lemma: CLT alpha-mixing}
Let $\{Z_t\}$ be a centred $\alpha$-mixing stationary time series. Suppose that it holds for some $\delta>0$,
$$
\sum_{k=1}^\infty \alpha(k)^{\delta/(2+\delta)} < \infty \text{ \ and \ } \E(|Z_1|^{2+\delta}) <\infty.
$$
Denote $S_n = \sum_{t=1}^n Z_t$ and $\sigma_n^2 = \E\left[ |S_n|^2 \right]$. Then
$$
\frac{S_{ \lfloor{nt}\rfloor} }{\sigma_n} \to W(t),
$$
where the convergence is in Skorohod topology and $W(t)$ is the standard Brownian motion on $[0,1]$.
\end{theorem}

\section{Inequalities}

\begin{lemma}\label{lemma: sum-integration bound}
Let $f:[0,\infty] \to [0,\infty]$ be monotonically decreasing continuous function such that $\int_1^{\infty} f(x) dx < \infty$. Then
$$
\int_1^\infty f(x) dx \le \sum_{k\ge1} f(k) \le f(1) + \int_1^\infty f(x) dx \le f(0) + \int_1^\infty f(x) dx. 
$$

\end{lemma}

\begin{lemma}\label{sequence rate bound general}
Let $r>1$ be a constant. For any positive sequence $\mathfrak{s}_j \asymp j^{-2r}$ and $\varphi\ge1/2$ we have
\begin{equation}
    \sum _{j\ge 1} \frac{\mathfrak{s}_j^{\varphi}}{(\alpha + \mathfrak{s}_j)^{\varphi}} \le c_1 \alpha^{-1/2r} 
\end{equation}
given any $\alpha > 0$. Here $c_1 >0$ is some constant.
\end{lemma}

\begin{proof}
Given $\mathfrak{s}_j \asymp j^{-2r}$, we have some positive constants $c_2$ and $c_3$ such that $c_2 j^{-2r} \le \mathfrak{s}_j \le c_3 j^{-2r}$,
$$
\implies \sum _{j\ge 1} \frac{\mathfrak{s}_j^{\varphi}}{(\alpha + \mathfrak{s}_j)^{\varphi}} \le \sum_{j\ge1} \frac{ (c_3 j^{-2r})^{\varphi} }{ (\alpha + c_2 j^{-2r} )^{\varphi} } = c_3 \sum_{j\ge1} \frac{1}{ (\alpha j^{2r} + c_2 ) ^{\varphi} }.
$$
Now, we shall upper bound the quantity on the right of above equation using \Cref{lemma: sum-integration bound}. Observe that the function defined by $x \mapsto \frac{1}{ (\alpha x^{2r} + c_2 ) ^{\varphi} }$ satisfy the conditions of \Cref{lemma: sum-integration bound} and therefore

$$
\sum_{j\ge1} \frac{1}{ (\alpha j^{2r} + c_2 ) ^{\varphi} } \le c_4 + \int_1^\infty  \frac{1}{ (\alpha x^{2r} + c_2 ) ^{\varphi} } dx
$$
Observe that
\begin{equation}\label{test 1}
    \begin{aligned}
    &\int_1^{(c_2/\alpha)^{1/2r}}  \frac{1}{ (\alpha x^{2r} + c_2 ) ^{\varphi} } dx \le \int_1^{(c_2/\alpha)^{1/2r}} \frac{1}{ (c_2 ) ^{\varphi} } dx = \frac{c_5}{\alpha^{1/2r}} - c_4 \le \frac{c_5}{\alpha^{1/2r}}
    \\
    & \int_{(c_2/\alpha)^{1/2r}}^\infty \frac{1}{ (\alpha x^{2r} + c_2 ) ^{\varphi} } dx \le \int_{(c_2/\alpha)^{1/2r}}^\infty \frac{1}{ (\alpha x^{2r} )  ^{\varphi} } dx = \frac{c_6}{\alpha^{1/2r}}.
\end{aligned}
\end{equation}
Using \eqref{test 1} we may write,
$$
\int_1^\infty  \frac{1}{ (\alpha x^{2r} + c_2 ) ^{\varphi} } dx \le \frac{c_5 + c_6}{\alpha^{1/2r}}
$$
which lead us to the required result.

\end{proof}

\begin{corollary}\label{sequence rate bound required}
Let $\left\{ \alpha_n \right\}_{n\ge1} $ be positive sequence converging to 0. Under assumptions of \Cref{sequence rate bound general}, we have
\begin{equation}
    \sum _{j\ge 1} \frac{\mathfrak{s}_j^{1+2t}}{(\alpha_n + \mathfrak{s}_j)^{1+2t}} = O\left( \alpha_n^{- 1/2r} \right)
\end{equation}
\end{corollary}







\begin{lemma}\label{lemma: maximal inequality useful}
    Let $\{Z_i\}$  be a sequence of random variable. Let $\xi > 0$. Suppose 
    $$
    \E\lft[ \lft| \sum_{m=1}^t Z_j \rgt|^2  \rgt] = O(t).
    $$
    Then for any positive non-increasing sequence $\{\gamma_t\}$, we have
    \begin{equation}\label{eq: maximal useful general}
        \E\lft[ \max_{1<t\le n} \lft| \gamma_t \sum_{m=1}^t Z_j\rgt|^2  \rgt] = O(1) \sum_{t=2}^n \gamma_t^2.
    \end{equation}
    Consequently we have 
    \begin{equation}\label{eq: maximal useful specific}
        \E\lft[ \max_{1<t\le n} \lft|\frac{1}{\sqrt{t \log^{1+\xi} t}} \sum_{m=1}^t Z_j\rgt|^2  \rgt] = O(1).
    \end{equation}
    
\end{lemma}

\
\\
\begin{proof}
    Observe that \eqref{eq: maximal useful general} follows directly from Theorem B.3 of \citet{kirch2006resampling}.

    \
    \\
    Note that $\lft\{  \frac{1}{ \sqrt{t \log^{1+\xi} t}}\rgt\}_{t=2}^n$ is a non-increasing sequence and from \eqref{eq: maximal useful general}  
    $$
    \E\lft[ \max_{1<t\le n} \lft|\frac{1}{\sqrt{t \log^{1+\xi} t}} \sum_{m=1}^t Z_j\rgt|^2  \rgt] = O(1) \sum_{t=2}^n \frac{1}{t \log^{1+\xi} t},
    $$
    and the \eqref{eq: maximal useful specific} follows from the fact that $\sum_{t=2}^\infty \frac{1}{t \log^{1+\xi} t} < \infty$.
\end{proof}

\
\\
\begin{lemma}\label{lemma: max to Op by peeling}
Let $\nu > 0$.  Let $\{Z_i\}$  be a sequence of random variable. Suppose we have
$$
\E\left[ |S_i^j|^2 \right] \le c' (j-i),
$$
where $S_i^j = \sum_{k=i}^j Z_k$ and $c'>0$ is some absolute constant.
Then for any given $\eps>0$
$$
P\left( \forall r > 1/\nu, \qquad  |S_r| \le \frac{C_1}{\sqrt{\eps}} \sqrt{r} (\log r\nu + 1)  \right) > 1 -\eps,
$$
where $C_1 = \frac{\pi}{\log 2}\sqrt{\frac{c}{6}}$.

\
\\
In other words
$$
\max_{r > 1/\nu} \frac{1}{\sqrt{r}\lft( \log (r\nu) + 1\rgt) } \lft| \sum_{j=1}^r Z_j \rgt| = O_p(1).
$$
\end{lemma}

\begin{proof}
Observe that using \Cref{lemma: maximal inequality useful} with $\gamma_t = 1$, we can get 
$$
\E\left[ \max_{ i < t \le j} |S_i^t|^2 \right] \le c (j-i)
$$
for some absolute constant $c>0$.

\
\\
We are going to use the peeling argument for the proof. With
$$
\E\left[  \max_{m\le k \le 2m} \left|\frac{S_k}{\sqrt{k}}\right|^2 \right] \le \frac{1}{m} \E\left[  \max_{m\le k \le 2m} |S_k|^2 \right] \le c.
$$
Let's define $A_j = \left[ 2^{j-1}/\nu, 2^{j}/\nu \right]$.
Using Markov's inequality we may write,
\begin{align*}
    &P\left( \max_{m\le k \le 2m} \left|\frac{S_k}{\sqrt{k}}\right|  \ge x  \right) \le \frac{1}{x^2} \E\left[  \max_{m\le k \le 2m} \left|\frac{S_k}{\sqrt{k}}\right|^2 \right] \le cx^{-2} 
    \\
    \implies & P\left( \bigcup_{k\in A_j} \left\{ \left|\frac{S_k}{\sqrt{k}}\right|  \ge \alpha j \right\} \right) \le \frac{c}{\alpha^2 j^2}
    \\
    \implies & P\left( \bigcup_{k\in A_j} \left\{ \left|\frac{S_k}{\sqrt{k}}\right|  \ge \alpha (\log_2 {k\nu} + 1) \right\} \right) \le \frac{c}{\alpha^2 j^2}.
\end{align*}
The last equation follows from 
$$
2^{j-1}/\nu \le k \le  2^{j}/\nu \implies  j \le \log_2 {k\nu} + 1 \le j+1.
$$
And
\begin{align*}
    & P \left( \bigcup_{r \ge 1/\nu} \left\{ \left|\frac{S_r}{\sqrt{r}}\right|  \ge \alpha (\log_2 {r\nu} + 1) \right\} \right) 
    \\
    = & P \left( \bigcup_{j\ge1} \bigcup_{k \in A_j} \left\{ \left|\frac{S_k}{\sqrt{k}}\right|  \ge \alpha (\log_2 {k\nu} + 1) \right\} \right)
    \\
    \le & \sum_{j = 1}^{\infty} P\left( \bigcup_{k\in A_j} \left\{ \left|\frac{S_k}{\sqrt{k}}\right|  \ge \alpha (\log_2 {k\nu} + 1) \right\} \right)
    \\
    \le &\frac{c}{\alpha^2} \sum_{j = 1}^{\infty} \frac{1}{j^2} = \frac{c \pi^2}{6 \alpha^2}.
\end{align*}
With $\alpha^2 = \frac{c\pi^2}{6\eps \log^2{2}} $ and $\log (2) < 1$, we can have
$$
P \left( \max_{r\ge 1/\nu} |S_r|  \ge \frac{C_1}{\sqrt{\eps}} \sqrt{r}(\log{r\nu} + 1)  \right) \le \eps.
$$

\end{proof}

\end{document}